\documentclass[12pt]{article}
\usepackage{iftex, amssymb, amsmath, amsfonts, amsthm, tikz-cd, epigraph}
\usepackage{xcolor} 
\usepackage[mathscr]{eucal}
\ifpdf
  \usepackage{underscore}         % Only needed if you use pdflatex.
  \usepackage[T1]{fontenc}        % Recommended with pdflatex
\else
  \usepackage{breakurl}           % Not needed if you use pdflatex only.
\fi

\definecolor{mygray}{gray}{0.3}
\definecolor{mycolor}{rgb}{0.05, 0.4, 0.2}
\newcommand{\blue}{\color{blue}}
\newcommand{\gray}{\color{mycolor} \small }

\newcommand{\maxtensor}{\otimes_{\mbox{max}}}
\newcommand{\mintensor}{\otimes_{\mbox{min}}}

\newcommand{\0}{\mbox{\sf 0}}
\newcommand{\1}{\mbox{\sf 1}}
\newcommand{\A}{{\mathscr A}} %{\mathcal A}}
\newcommand{\Aff}{\mbox{Aff}}

\newcommand{\B}{{\mathscr B}} %{\mathcal B}}

\newcommand{\Boxworld}{\mbox{\bf Boxworld}}
\newcommand{\C}{{\mathbb C}}
\newcommand{\Cat}{{\mathcal C}}
\newcommand{\con}{\mbox{con}}
\newcommand{\Conv}{\mbox{\bf Conv}}

\newcommand{\D}{{\mathcal D}}
\newcommand{\Dat}{{\mathcal D}}

\newcommand{\E}{{\mathbb E}}
\newcommand{\Ev}{{\mathcal E}}
\renewcommand{\epsilon}{\varepsilon}
\newcommand{\Expo}{\mbox{Exp}}
\newcommand{\F}{{\mathscr F}}
\newcommand{\FinSet}{\mbox{\bf FinSet}}
\newcommand{\Func}{\mbox{\bf Func}}
\newcommand{\G}{{\mathscr G}}
\newcommand{\green}{\color{green}}
\newcommand{\Grp}{\mbox{\bf Grp}}
\renewcommand{\H}{{\boldsymbol{\mathscr H}}}
\newcommand{\Hilb}{\H} %\mbox{\sf Hilb}}

\newcommand{\Hom}{\mbox{Hom}}
\newcommand{\id}{\mbox{id}}

\newcommand{\K}{{\boldsymbol{\mathscr K}}}
\renewcommand{\L}{{\boldsymbol{\mathscr L}}} %{\mathcal L}}
\newcommand{\M}{{\mathscr M}}
\newcommand{\Meas}{\mbox{Meas}}
\newcommand{\Prob}{\Mod}
\newcommand{\N}{{\mathbb N}}

\newcommand{\Nat}{\mbox{\bf Nat}}
\newcommand{\op}{\mbox{op}}

\newcommand{\OUS}{\mbox{\bf OUS}}
\newcommand{\Mod}{\mbox{\bf Prob}}
\renewcommand{\P}{{\mathbb P}}
\newcommand{\Pow}{{\boldsymbol{\mathscr P}}} %{\mathcal P}}
\newcommand{\Power}{\Pow}
\newcommand{\pr}{\mbox{Pr}}

\newcommand{\Q}{{\mathbb H}}
\newcommand{\R}{{\mathbb R}}

\newcommand{\RVec}{\mbox{\bf RVec}}
\renewcommand{\S}{{\mathscr S}}
\newcommand{\sa}{\mbox{\small sa}}
\newcommand{\Set}{\mbox{\bf Set}}

\newcommand{\spn}{\mbox{Span}}
\newcommand{\supp}{\mbox{supp}}

\newcommand{\tensor}{\otimes}
\renewcommand{\tilde}{\widetilde}
\newcommand{\U}{{\mathbb U}}
\newcommand{\Unitary}{{\mathscr U}}
\newcommand{\V}{{\mathbb V}}
\renewcommand{\Vec}{\RVec}
\newcommand{\W}{{\mathbb W}}

\renewcommand{\hat}{\widehat}
\newcommand{\X}{{\mathbb X}}
\newcommand{\Z}{{\mathbb Z}}

\newcommand{\oc}{\mbox{ co }}
\newcommand{\co}{\oc}

\newcommand{\Exp}{{\mathbb E}}
\newcommand{\Aut}{\mbox{Aut}}
\newcommand{\for}[1]{\overrightarrow{#1}}
\newcommand{\back}[1]{\overleftarrow{#1}}
\newcommand{\bilat}[1]{\overleftrightarrow{#1}}

\newcommand{\tempout}[1]{{}}
\newcommand{\Tr}{\mbox{Tr}}
\newcommand{\tr}{\mbox{Tr}}

\newcounter{thaler}
\newenvironment{mlist}{\begin{list}{\arabic{thaler}}%
{\usecounter{thaler}
\setlength{\rightmargin}{\leftmargin}
\topsep=0pt
\itemsep=0pt
\parskip=0pt
\parsep=0pt
}}{\end{list}}

\usepackage{tikz}
\tikzset{help lines/.style=very thin}

\newtheorem{theorem}{Theorem}[section]

\newtheorem{lemma}[theorem]{Lemma}
\newtheorem{proposition}[theorem]{Proposition}
\newtheorem{corollary}[theorem]{Corollary}

\theoremstyle{definition} 
\newtheorem{definition}[theorem]{Definition}

\newtheorem{example}[theorem]{Example}

\makeatletter
\newtheoremstyle{indented}
  {11pt}% space before
  {11pt}% space after
  {\addtolength{\@totalleftmargin}{3.5em}
   \addtolength{\linewidth}{-3.5em}
   \parshape 1 3.5em \linewidth}% body font
  {}% indent
  {\bfseries}% header font
  {.}% punctuation
  {.5em}% after theorem header
  {}% header specification (empty for default)
\makeatother

\theoremstyle{indented}
\newtheorem{exercise}{Exercise}

\begin{document}

%\end{document} 

\begin{center}{\bf {\large Generalized Probability Theory} \\
Notes for a short course}

\vspace{.2in} 

Alex Wilce\footnote{wilce\@susqu.edu}\\
%PI, 
%{\blue Winter, Spring} 2024
%for HB, JH, CH, SS, JvdW
\end{center} 
%\title
{\fontsize{8}{9}\selectfont
\epigraph{\fontsize{9}{11}\selectfont The only way of discovering the limits of the possible is to venture a little way past them into the impossible.}{\textit{\fontsize{9}{11}\selectfont Arthur C. Clarke}}
}
%%EPIGRAMS: von Neumann, Strauss (?), Mackey, Hardy...

%\begin{abstract} 
%\end{abstract}

This is a slightly revised and expanded version of notes  distributed during a short course on GPTs given at the Perimeter Institute for Theoretical Physics in March and April of 2024.  I want to thank the members of th course, and my hosts, Lucien Hardy and Rob Spekkens, for their hospitality during that time. A special note of thanks goes to Maria Ciudad Ala\~{n}on, at whose suggestion I wrote section 4.4. \\

{\large \bf Introduction} 

The idea of regarding quantum theory as a "non-classical" probability calculus
%, within some suitably broad generalization of classical probability theory, 
goes back to the work of von Neumann \cite{vN}, and include the extensive mathematical literature on "quantum logics" from the 1960s, '70s and '80s, sparked by the work of Mackey \cite{Mackey}. Ideas from quantum information theory have led to a strong revival of this project during the past 20 years or so, beginning with seminal papers of Hardy \cite{Hardy} and Barrett \cite{Barrett}. A rapidly growing area of research, located somewhere between physics and mathematics, is now devoted to "generalized probabilistic theories", or GPTs (a phrase due to Barrett). This has greatly sharpened our understanding of many aspects of quantum theory, especially those having to do with entanglement: this turns out to be a generic feature of non-classical GPTs, and  thus, at least from a mathematical point of view, not a specifically "quantum" phenomenon at all. 

%As has become clear, many of the most striking qualitative %features of entangled quantum states (e.g., the possibility %of teleportation) are actually 
 %and provided tools for exploring how one might modify or extend it. 
%This newer 
%work is characterized by a distinctive emphasis on composite systems, and %a willingness to focus on finite-dimensional systems. 
%The former in one way broadening, and the latter in another way usefully narrowing, the scope of %investigations.

However, precisely because of its rapid growth, research in GPTs is somewhat scattered, with different groups making use of slightly different, "home-grown" mathematical frameworks. Moreover, all of this has taken place without much engagement with the earlier literature. This is unfortunate, since many ideas and techniques from that earlier period are still of interest and of use.  My aim in these notes is to give a unified, mathematically clear, and historically-informed outline of GPTs as a generalized probability theory. The approach I develop draws particularly from work of D. J. Foulis and C. H. Randall (e.g., \cite{RF70, RJF, FR-TP, FPR}) in the 1970s and '80s. This was originally presented, in part, as a  generalization --- and, implicitly, a criticism --- of then-prevailing methods and assumptions in quantum logic.  Suitably updated, the Foulis-Randall approach remains, in my view, the best (most flexible, most expressive) framework currently available for the study of GPTs. 
%, striking a balance between mathematical generality and elegance, and concreteness and sound operational motivation, on the other.    

As I hope will become clear, I regard this framework not so much as something belonging to physics, but as a very conservative generalization of classical probability theory. Indeed, I rather think the acronym GPT should be read as "general probability theory".  The only real departure from classical probability theory is that we abandon the (usually tacit) assumption that all statistical experiments can effectively be performed jointly.  The resulting theory's scope is very broad, allowing it to take in quantum theory, as well as various hypothetical "post-quantum" probabilistic physical theories.  
%Nevertheless, the degree to which this %looser probabailistic framework should be regarded as genuinely "non-classical" is unclear. I will come back to this point briefly at the end of these notes. 

{\em Prerequisites}  I assume the reader is comfortable with basic mathematical ideas, idioms, and notations, especially the essentials of (naive) set theory and linear algebra, the latter ideally including duality and tensor products. Beyond this, some familiarity with measure theory and functional analysis, a bit of point-set  topology, and a casual acquaintance with category theory will all be helpful, but I will briefly review some of this material as we go along, as it's needed to follow the discussion. Some less central material is marked off in {\gray green}. 

{\em Notational conventions} If $X$ and $Y$ are sets, 
I write $Y^{X}$ for the set of all mappings $f : X \rightarrow Y$ and $\Power(X)$ for the power set of $X$. If $\V$ is a vector space, the set $\V^{X}$ is a vector space under pointwise operations. In particular, $\R^{X}$ and $\C^{X}$ are vector spaces in this way.  If $\V, \W$ are vector spaces, $\L(\V,\W)$ denotes the set of all linear mappings $\V \rightarrow \W$; $\L(\V)$ is short for $\L(\V,\V)$, and $\V'$ is $\L(\V,\R)$, the algebraic {\em dual space} of $\V$. If $\V$ carries a linear topology, e.g., if $\V$ is a normed space, $\V^{\ast}$ always denotes the continuous dual. 
Physicists like to write an inner product on a complex vector space $\H$ as $\langle x | y \rangle$, understanding this as linear in the second argument. Mathematical tradition uses $\langle x, y \rangle$, usually linear in the first argument. In either case, the inner product is conjugate-symmetric, so if we simply define $\langle x | y \rangle = \langle y, x \rangle$, 
(where the side being defined depends on which side you're on), we can go back and forth between these notations as the mood takes us. On a real vector space, of course, it doesn't matter, and I'll stick to using a comma, rather than a bar, in that case. 

{\em Exercises} I've taken the liberty to sprinkle exercises throughout the text. Most are very simple, amounting to an invitation to confirm that something really is "evident" or "easily checked".  A few require somewhat more thought, and these are marked with $\star$.  The first sort should of course be attempted when encountered; the second sort should be read (as they sometimes mention useful facts), but needn't be tackled unless they seem especially interesting.  Exercises in {\gray green} are totally optional, and needn't even be read, let alone attempted. 

\tempout{
{%\green 
{\em ``Measurement"}  John Bell famously inveighed against the use of ``measurement" as a fundamental term in physics. This needn't concern us here, since we are developing probability theory, which is prior to physics. But one of his points is well-taken: ``measurement" suggests that there is something (some pre-existing value) being measured.  Bell suggests that ``experiment" would be better, and I largely agree: for the most part, that's the term I use. But ``measurement" is simply more euphonious than 
``experiment", and slips off the pen (keyboard?) very readily. To avoid having to impose on myself a draconian editorial regime in which all use to the M-word is suppressed in favor of the E-word,  I will instead just declare that, in these notes, the two are to be understood synonymously.}
}

%{\bf A Request}  I'm in the process of writing a book with John Harding on Quantum Structures (basically, all the different mathematical frameworks people use to talk about 
%non-relativistic quantum theory). Some of the material in these notes will likely end up in relevant chapters of that book. Looking further ahead, I am hoping to write a shorter and more focused second book specifically on GPTs, and these notes are a rough first pass at assembling material for that project.  
%These notes are a very rough first pass at assembling material that, I hope, will grow into a book 
%on GPTs. 
%\newpage
I intend at some point to expand these notes. In particular, I hope eventually to add a chapter on quantum reconstructions, and one on contextuality, ontic models, and locality. Meanwhile, I'll be very grateful for any feedback particularly if it involves pointing out mistakes! 

%\end{document}
\newpage
{\bf Outline} 

1. Probabilistic Models
\begin{mlist} 
\item[1.1] Test spaces and probability weights 
\item[1.2] Events, perspectivity, and a digression on quantum logics
\item[1.3] Sequential measurement and 
interference
\item[1.4] Mappings of models %[observables] 
%\item[1.5] Contextuality (??) [Add state prep here...]
%\item[1.3] Constructions with probabilistic models [Own section??]
%\item[1.4] Mappings of probabilistic models 
%\item[1.6] Historical remarks
\end{mlist} 

2.  Probabilistic Models Linearized 
\begin{mlist} 
\item[2.1] Ordered vector spaces
\item[2.2] Ordered vector spaces associated with a probabilistic model 
\item[2.3] Effect algebras
\item[2.4] Processes 
%\item[2.5] Historical remarks 
\end{mlist} 

3.  Composite Models and Entanglement 
\begin{mlist} 
\item[3.1] Bipartite states and non-signaling composites 
\item[3.2] Composites linearized 
\item[3.3] Local tomography and the maximal and minimal tensor products
\item[3.4] Entanglement and remote evaluation 
%\item[3.5] Historical remarks 
\end{mlist} 

4.  Probabilistic Theories 
\begin{mlist}
\item[4.1] Categories, functors, natural transformations 
\item[4.2] Monoidal categories and process theories 
\item[4.3] Probabilistic theories as functors
\item[4.4] Other frameworks
%\item[4.5] Compact closure ($\ast$)
%%\item[4.4] Monads on $\Prob$ 
%%\item[4.5] Compact Closure
\end{mlist} 

Appendices
\begin{mlist} 
\item[A] State of the ensembles
\item[B] Base-normed and order unit spaces
\item[C] Completeness results for $\V(A)$.
\end{mlist}

%5.  Sequential Measurement 
\tempout{
5.  Reconstructions 
\begin{mlist} 
\item[5.1] Quantum reconstructions in antiquity 
\item[5.2] Hardy-style reconstructions 
\item[5.3] The Pavian reconstruction 
\item[5.4] Jordan-algebraic reconstructions 
\end{mlist} 

6.  Further Topics ($\ast$)
\begin{mlist}
\item[6.1] Compounding, coarse-graining and interference
\item[6.2] Real and Quaternionic QM; InvQM
\item[6.3] Contextuality and classicality
\end{mlist} 
}
\newpage
\section{Probabilistic Models}

 \epigraph{\fontsize{9}{11} \selectfont [The outcome of a] physical operation is ... %merely 
 a {\em symbol}...}{\fontsize{9}{11} \selectfont \textit {C. H. Randall and D. J. Foulis} \cite{RF-Hooker}}
\vspace{-.5in}
\epigraph{
\fontsize{9}{11} \selectfont An outcome is a bump on the head.}{\fontsize{9}{11} 
	\selectfont \textit{ E. Wigner} \cite{Wigner}}

Very roughly, a {\em probabilistic model} for a physical system, is a mathematical structure in which the system is assigned a space of {\em states}, a space of {\em measurements} or {\em experiments}, and a way of assigning probabilities to the outcomes of the latter, given an element of the former. 
In this section, I will discuss various way in which one can represent such a structure mathematically. I begin with a framework (due to D.J. Foulis and C.H. Randall) that is  mathematically very simple and conceptually very conservative, but nevertheless remarkably general. In later sections, I will show how this leads naturally to a kind of linear representation of probabilistic models in terms of ordered vector spaces. 

\subsection{Test spaces and probability weights}
%\footnote{A general reference for this section is the survey paper \cite{AW-handbook}}

In elementary classical probability theory, a "probabilistic model" is a pair $(E,\alpha)$ where $E$ is the discrete outcome-set of some experiment, and $\alpha$ is a probability weight on $E$.  An obvious generalization is to allow both $E$ and $\alpha$ to vary.  Let's start with the outcome-set.

\begin{definition}\label{def: test space} A {\em test space}\footnote{Originally, Foulis and Randall called these {\em manuals}, and their elements, {\em operations}, which terminology I like better. But "test" and "test space" are well settled in the literature at this point.} \cite{FGR-II} is a collection $\M$ of outcome-sets $E, F, ....$ of various measurements, experiments, or {\em tests}.  If $X = \bigcup \M$ is the set of all outcomes of all tests, then a {\em probability weight} on $\M$ is a function $\alpha : X := \bigcup \M \rightarrow [0,1]$ with $\sum_{x \in E} \alpha(x) = 1$ for every test $E \in \M$.  
\end{definition}
Notice that if $E, F \in \M$ and $E \subseteq F$, then every probability weight $\alpha$ assigns probability $0$ to every outcome $x \in F \setminus E$. For this reason, it is usual to assume, and I will assume, that $\M$ is {\em irredundant}, meaning that if $E, F \in \M$ with $E \subseteq F$, then $E = F$. 

Denote the set of all probability weights on $\M$ by $\Pr(\M)$, and notice that this is a {\em convex} subset of $[0,1]^{X}$: weighted averages of probability weights are again probability weights. 

The simplest examples, of course, are those in which $\M$ consists of a single test: $\M = \{E\}$. In this case, $\Pr(\M) = \Delta(E)$, the simplex of all probability weights on $E$. Test spaces of this form are said to be {\em classical}. 

%{\bf Example 1:} 
\begin{example} 
A very simple (literally two-bit) non-classical test space consists of a pair 
\[\M = \{\{x,x'\}, \{y,y'\}\}\] 
of disjoint, two-outcome tests. A probability weight is determined by $\alpha(x)$ and $\alpha(y)$, and 
these can take any values in $[0,1]$, so 
$\Pr(\M)$ is isomorphic, as a convex set, to 
the unit square in $\R^{2}$.  
\end{example} 

More generally, a test space  $\M$ is {\em semi-classical} iff $E \cap F = \emptyset$ for distinct tests $E, F \in \M$. In this case, a probability weight $\alpha$ on $\M$ amounts to an assignment of a probability weight $\alpha_{E}$ to each test $E \in \M$, so $\Pr(\M)$ is effectively the Cartesian product of the simplices $\Delta(E)$, $E \in \M$.  
%\begin{exercise} Show that $\Pi_{i \in I} \Delta(E_i)$ is a simplex iff all but at most one 
%of the sets $E_i$ is a singleton. \end{exercise} 
In general, however, the tests $E \in \M$ can overlap, and the  combinatorial structure of $\M$ imposes constraints on the possible probability weights.

%{\bf Example 2:} 
\begin{example} Let $\M = \{u,x,v\}, \{v,y,w\}$.  A probability weight $\alpha$ on $\M$ must satisfy 
\[\alpha(x) = 1 - (\alpha(u) + \alpha(v)) \ \ \mbox{and} \ \ 
\alpha(y) = 1 - (\alpha(v) + \alpha(w))\]
and is thus determined by the triple $(\alpha(u), \alpha(v), \alpha(w)) \in \R^{3}$, subject 
to the conditions that $\alpha(u), \alpha(v), \alpha(w)$ are all non-negative, and  $\alpha(u) + \alpha(v)$ and $\alpha(v) + \alpha(w)$ are both $\leq 1$. 
\end{example}

\begin{exercise}  With $\M$ as in Example 1.3, show that $\Pr(\M)$, as represented in $\R^{3}$, is a (skewed) square-based pyramid with apex at $(0,1,0)$. 
\end{exercise} 

%{\bf Exercise:} 
\begin{exercise} 
Find simple examples of test spaces having (a) no probability weights, (b) exactly one probability weight. 
\end{exercise} 

An immediate question arises: which convex sets can be represented as $\Pr(\M)$ for a test space $\M$? It can 
be shown that {\em every compact convex set} arises 
in this way. (This follows from a sharper result 
due to Shulz \cite{Shulz}).)

%{\bf Exercise 
\begin{exercise}[requiring some general topology] Show that if $\M$ is {\em locally finite}, meaning that every test is a finite set, then $\Pr(\M)$ is closed, and hence compact, in the product topology on $[0,1]^{X}$. 
\end{exercise} 

%Pure states here??

%\subsection{Models} 

%In specific applications, whether classical or quantum, we will may be interested only in certain probability weights.  The following definition allows us the flexibility to restrict our attention in this way.

%{\bf Definition:} 
\begin{definition}\label{def: model} A (general) {\em probabilistic 
model} %can now be defined as
is a pair $(\M,\Omega)$ where 
$\M$ is a test space and $\Omega \subseteq \Pr(\M)$ is some designated set of probability weights. We call 
weights in $\Omega$ {\em states} of the model, and 
$\Omega$, its {\em state space}. 
\end{definition} 

Of course, we can always take $\Omega$ to be the whole set $\Pr(\M)$ of {\em all} probability weights on $\M$, in which case we'll say that model is {\em full}. In general, however, one will want to place some restrictions on which probability weights count as legitimate states in one's model, in ways that will 
vary with the case at hand. 
%and it would be unwise 
%to constrain this choice to severely at the outset. 

This said, one often imposes some structural constaints on 
the state spaces of models. A very weak example is the following: 

\begin{definition}\label{def: positive set of weights} A set $\Omega$ of probability weights on a test space $\M$ is {\em positive} iff, for every $x \in X = \bigcup \M$, there is at least one $\alpha \in \Omega$ with $\alpha(x) > 0$. 
\end{definition} 

If $\Omega$ is not positive, one can simply replace $X$ with the 
smaller set $X_{+} = \{ x | \exists \alpha \in \Omega \ \alpha(x) > 0\}$ and $\M$ with $\M_{+} = \{ E \cap X_{+} | E \in \M\}$. 
Identifying each state with its restriction to $X_{+}$, we 
end up with a positive model $(\M_+,\Omega)$. It is therefore 
usually harmless to assume that all models have positive 
state spaces.  {\bf \em This will be a standing assumption henceforth}.

{\bf Convexity and Closure}  
It's  usual to assume that the chosen state space $\Omega$ is convex in $\R^{X}$ --- if not, replace $\Omega$ by its convex hull. This reflects the idea that one can always prepare arbitrary random mixtures of states. 

Since states belong to the space $B(X)$ of bounded real-valued 
functions on $X$, we can also ask that $\Omega$ be closed with respect to the sup norm on the latter: equivalently, that 
uniform limits of states are states.  Although these are not by any means  necessary features of a probabilistic model (consider, e.g., pure-state QM, in which 
 $\Omega$ is not convex), in the interest of simplifying the exposition, I will make it another {\bf \em standing assumption} in these notes that {\bf \emph{state spaces are always uniformly closed and convex}}, unless otherwise indicated.

{\em Remark:} Note that requiring $\Omega$ to be closed 
under uniform limits is weaker than requiring it to be closed under pointwise sequential limits, and this in turn is weaker 
than requiring $\Omega$ to be closed under arbitrary 
pointwise-limits: the latter amounts to closure in the 
product topology on $[0,1]^{X}$, which would render $\Omega$ 
compact in this topology, a very powerful constraint that 
is nevertheless satisfied in many cases (recall Exercise 3!). More about this later. 

%\begin{exercise} A probability weight on a test space $\M$ is {\em dispersion-free} (d.f.) 
%iff it takes only the values $0$ and $1$. Show that a d.f. probability weight in pure 
%in $\Pr(\M)$, and hence, also in any smaller state space.  
%\end{exercise} 

{\bf Further Examples} Here are some further examples. The first shows that standard measure-theoretic probability theory is within the scope of present framework.  
%of probability weights on test spaces.

%{\bf Example 3 (Borel Models):} 
\begin{example}[Borel Models]\label{ex: Borel} Let $(S,\Sigma)$ be a measurable space. The corresponding locally countable and locally finite {\em Borel test spaces} 
$\B_{\sigma}(S,\Sigma)$ and $\B_{o}(S,\Sigma)$ consist, respectively of countable and of finite  measurable partitions of $S$. The probability weights on these correspond in a natural way to countably, respectively finitely, additive probability measures on $(S,\Sigma)$.  
A Borel test space becomes a {\em Borel model} if we equip it with any designated compact, convex set of probability measures. (For instance, one might consider the model consisting of the locally countable Borel test space of $\R$, equipped with the set of all Radon probability measures, or those absolutely continuous with respect to Lebesgue measure.  
%{\blue It is important to note here that $\Pr(\B_{\sigma}(S,\Sigma))$ is 
%generally not compact. [Check]}
\end{example} 

The following examples recover, in our slightly non-standard language, some very standard ways of modelling quantum systems. 

%\newpage
%{\bf Example 4 
\begin{example}[Hilbert Models]\label{ex: Hilbert} Let $\H$ be a (real or complex)  Hilbert space, and let $\F(\H)$ be the set of {\em frames} --- unordered orthonormal bases --- of $\H$. Note that the outcome space here, $X(\H) = \bigcup \F(\H)$, is exactly $\H$'s unit sphere. Every unit vector $v \in \H$ defines a probability weight on $\H$ by $\alpha_{v}(x) := |\langle v, x \rangle|^2$. Let $\Omega(\H)$ be the closed convex span of these in $\R^{X({\mathbf H})}$: then every state $\alpha \in \Omega(\H)$ has the form $\alpha_W (x) = \langle Wx, x \rangle$ for a unique density operator $W$ on $\H$. {\em Gleason's Theorem} \cite{Gleason} tells us that if $\dim(\H) > 2$, the model $(\F(\H), \Omega(\H))$ is full, i.e., {\em every} probability weight on $\F(\H)$ belongs to $\Omega(\H)$.  
(This is a distinctly nontrivial result, not at all easy to prove.)  

%{\gray 
In this model, different outcomes --- that is, unit vectors of $\H$ --- that differ by a scalar factor will have the same probability in all states: if $y = cx$ where $|c| = 1$, 
then 
\[\alpha_{W}(y) = \langle Wcx, cx \rangle = |c|^2 \langle Wx, x \rangle = \alpha_{W}(x).\]
We might want to identify such outcomes with one another.  A convenient way to do this is to replace $\F(\H)$ by the set $\F_{p}(\H)$ of {\em projective frames}, i.e., maximal pairwise orthogonal sets of rank-one projections on $\H$.  If $p = p_{x}$ is the projection onto the
span of unit vector $x$ and $W$ is a density operator, we have  $\langle Wx, x \rangle 
= \Tr(Wp)$. This defines a probability weight on $\F_{p}(\H)$, and the set of such weights 
gives us the state space for the {\em projective Hilbert model} based on $\H$.  In a sense made more precise later on, this is a quotient of the (as we might call it) vectorial  Hilbert model described above. As we will also see, passing to this quotient is not always safe. 
 %}

\tempout
{\gray Two slightly different models, both of which might be called {\em projective Hilbert model}, represent outcomes, not by unit vectors, but by one-dimensional projections.  
Let $\F_{p}(\H)$ denote the set of {\em projective frames}, i.e., maximal pairwise orthogonal sets of rank-one projections. Each unit density operator $\H$ gives rise to a 
probability weight, which I will again call $\alpha_{W}$, on $\F_{p}(\H)$,  $\alpha_{W}(p) = \Tr(Wp)$. Letting $\Omega_{p}(\H)$ denote the closed convex span of these weights, 
we see that $(\F_{p}(\H),\Omega_p(\H))$ and $(\F(\H),\Omega(\H))$ have essentially 
the same states, but slightly different test spaces. 

For a third quantum model, replace $\F_{p}(\H)$ by the set $\M(\H)$ of all 
maximal pairwise orthogonal sets of projections, of whatever rank. Again, we have for each 
density operator $W$ a probability weight given by $\alpha_{W}(p) = \Tr(Wp)$; 
and again, the state space has not changed in any essential way; rather, we've 
added more tests to $\F_{p}(\H)$. } 
\end{example}

%{\gray {\bf Example 5 
{\gray 
\begin{example}[von Neumann Models]\label{ex: vN} Let $\A$ be a von Neumann algebra without type $I_2$ summand, and let $\P(\A)$ be its projection lattice. The set $\M(\A)$ of countable (resp., finite) partitions of unity in $\P(\A)$ is a test space, and every normal (resp., arbitrary) state $f \in \A^{\ast}$ induces a probability weight on it. If $\A$ has no type $I_2$ factor, then conversely, every probability weight on $\M$ arises in this way from a normal (resp., arbitrary) state. This is the content of the (again, highly non-trivial) Christensen-Yeadon extension of Gleason's Theorem.  A still more general result, 
due to L. J. Bunce and J. D. M. Wright \cite{B-MW} shows that if $\X$ is any Banach 
space and $\mu : \P(\A) \rightarrow \X$ is any finitely-additive probability measure 
on $\P(\A)$, there exists a unique extension of $\mu$ to a bounded linear operator 
$\A \rightarrow \X$.  

In the special case in which $\A$ is the algebra of all bounded operators on a Hilbert space $\H$, we will write $\M(\A)$ as $\M(\H)$. Note that $\M(\H)$ includes $\F_{p}(\H)$ as a sub-test space, and all probability weights on $\M(\H)$ restrct to probability weights on the latter, If $\dim(\H) > 2$, these are  determined by density operators on $\H$, which, in turn, 
define states on $\M(\H)$ by the rule 
$p \mapsto \Tr(Wp)$. Thus, every probability weight on $\F_{p}(\H)$ extends uniquely to a probability weight on $\M(\H)$. 
\end{example} 
}

{\bf Pure states and dispersion-free states}  An point $a$ in a convex set $K$ is 
{\em extreme} iff it can't be expressed nontrivially as a convex combination of other points. 
That is, $a$ is pure iff, for any $0 < t < 1$, 
\[a = tb + (1-t)c \ \Rightarrow \ \ b = c = a.\]
A basic result in functional analysis (the {\em Krein-Milman Theorem}) tells us 
that every compact convex set is the closed convex hull of its extreme points.  Since 
we are assuming that our state-spaces $\Omega(A)$ are compact and convex, it follows that they have an abundance of extreme points. In physics, these are more usually called 
{\em pure states}.  

%A state $\alpha \in \Omega$ is {\em pure} iff it's an extreme point of $\Omega$; that is, 
%whenever $\alpha = t\beta + (1-t)\gamma$ for some $\beta, \gamma \in \Omega$ and %coefficient $0 < t < 1$, then  $\beta = \gamma = \alpha$.

One source of pure probability weights on $\B_{\sigma}(S,\Sigma)$ are the {\em point-masses} $\delta_{s}$, defined by $\delta_{s}(x) = 1$ iff $x = s \in S$ (so that $\delta_{s}(x) = 0$ if $x \not = s$). More generally, a probability measure $\mu$ on $\Sigma$ is pure iff it takes only the values $0$ and $1$\footnote{meaning that $\mu$ is an {\em ultrafilter} on $\Sigma$, and thus, the point mass associated with a point in the Stone space of $\Sigma$ --- but that's another story!}. For another example, the pure states on $\F(\H)$ are the vector states $\alpha_{v}$.  Thus, pure states on $\B_{\sigma}(S,\Sigma)$ are {\em dispersion-free} (d.f.), meaning that they take only values (probabilities) $0$ or $1$, while 
those on $\F(\H)$ are never d.f.: just consider any vector $x$ that is neither orthogonal to, nor a multiple of, $v$.  

%{\bf Exercise:} 
\begin{exercise} \hspace{.1in} 
%et $\M$ be any test space.
\begin{mlist} 
\item[(a)] Show that dispersion-free probability weights are always pure.   
%\item[(b)] Construct an example of a non-d.f. pure probability weight. 
\item[(b)] Show that the extreme points of the set of trace-one self-adjoint trace-class operators on  $\H$ are precisely the rank-one projections. 
\item[(c)] Using Gleason's Theorem, conclude that $\F(\H)$ has no d.f. probability weights if 
$\dim(\H) > 2$. 
\end{mlist} 
\end{exercise}  

{\bf Small test spaces and Greechie Diagrams} Small examples of test spaces, in  a small number of tests, each with only a few outcomes, intersect only in one or two outcomes, if at all, can be represented using {\em Greechie digrams}: 
outcomes are represented by nodes and sets of nodes belonging to a single test are connected by a line or other smooth curve (like beads on a wire) in such a way that tests correspond to maximal smooth curves. 
%constitute tests and curves representing distinct tests intersect transversally. 

For example, below are Greechie diagrams representing 
the test spaces $\{\{a,a'\}, \{b,b'\}$ and 
$\{a,x,b\}, \{b,y,c\}$ discussed earlier.
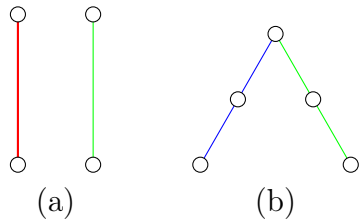
\[
\begin{array}{ccc}
\begin{tikzpicture} %1/2 PR box
\draw[thick, red] (0,0) -- (0,2);
\draw[green] (1,0) -- (1,2);
%\draw[blue] (0,1.74) -- (-.5,.87) -- (-1,0);
\draw[fill=white] (0,0) circle (.1cm) (0,2) circle (.1cm) (1,0) circle (.1cm) (1,2) circle (.1cm); 
\end{tikzpicture} 
& \hspace{.2in} & 
\begin{tikzpicture} %5-loop
%\draw[thick, red] (-1,0) -- (0,0) -- (1,0);
\draw[green] (1,0) -- (.5,.87) -- (0,1.74);
\draw[blue] (0,1.74) -- (-.5,.87) -- (-1,0);
\draw[fill=white] (-1,0) circle (.1cm) %(0,0) circle(.1cm) 
(1,0) circle (.1cm) (.5,.87) circle (.1cm) (0,1.74) circle (.1cm) (-.5,.87) circle (.1cm);
%\draw (-1.2,-.25) node{$a_0$} (0,-.25) node{$x_0$} (1.2,-.25) node{$a_1$} (.83,.9) node{$x_1$} (0,2.0) node{$a_2$} %(-.83,.9) node{$x_2$};
\end{tikzpicture}\\  
 \mbox{(a)} &  & \mbox{(b)} 
\end{array} 
\]

%{\bf Exercise:} 
\begin{exercise} 
Draw Greechie diagrams for the test spaces  $\{\{a,x,y,u\}, \{b,x,y,v\}$ and 
$\{\{x,y,z,u\}, \{x,y,z,v\}\}$. (Remember that curves 
representing tests needn't be straight.)
\end{exercise} 

The following small example is particularly interesting. 

%{\bf Example 6 
\begin{example}[The Firefly Box]\label{ex: Firefly}  A firefly is confined in a triangular box having an opaque top and bottom, but slightly translucent sides. The box is divided into three chambers, $a$, $b$ and $c$, 
in such a way that each side of the box gives a (cloudy) view 
of two chambers.  The walls between the chambers contain angled passageways, allowing the firefly --- but not light! --- to move from one chamber to another. An experiment consists in viewing the box from one side, and noting whether, on that viewing, a light appears on one side or the other, or not at all. 

We can represent this by a test space 
\[\M = \{\{a,x,b\}, \{b,y,c\}, \{c,z,a\}\}\]
where $x, y$ and $z$ represent the outcomes 
of seeing no light in the appropriate window.  Here is the corresponding Greechie 
diagram:
%\begin{columns}
%\begin{column}{0.2\textwidth}
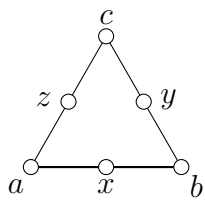
\[
\begin{array}{c} 
\begin{tikzpicture} %5-loop
\draw[thick] (-1,0) -- (0,0) -- (1,0);
\draw (1,0) -- (.5,.87) -- (0,1.74);
\draw (0,1.74) -- (-.5,.87) -- (-1,0);
\draw[fill=white] (-1,0) circle (.1cm) (0,0) circle (.1cm) (1,0) circle (.1cm) (.5,.87) circle (.1cm) (0,1.74) circle (.1cm) (-.5,.87) circle (.1cm);
\draw (-1.2,-.25) node{$a$} (0,-.25) node{$x$} (1.2,-.25) node{$b$} (.83,.9) node{$y$} (0,2.0) node{$c$} (-.83,.9) node{$z$};
\end{tikzpicture}\\  
%X = \mbox{nodes;} \   \M = \mbox{sides} 
%\M \ = \ \{a_0, x_0, a_1\}, \{a_1, x_1, %a_2\}, \{a_2, x_2, a_0\} 
\end{array} 
\] 
%\end{column}
%\end{columns}
%\end{frame} 
and here are some sample probability weights: 
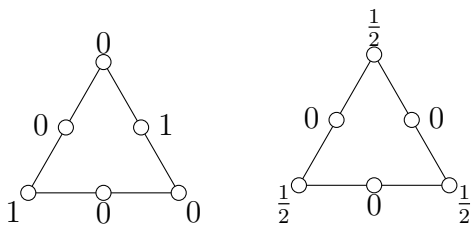
\[
\begin{array}{ccc} 
\begin{tikzpicture} %5-loop
\draw (-1,0) -- (0,0) -- (1,0) -- (.5,.87) -- (0,1.74) -- (-.5,.87) -- (-1,0);
\draw[fill=white] (-1,0) circle (.1cm) (0,0) circle (.1cm) (1,0) circle (.1cm) (.5,.87) circle (.1cm) (0,1.74) circle (.1cm) (-.5,.87) circle (.1cm);
\draw (-1.2,-.25) node{$1$} (0,-.26) node{$0$} (1.2,-.25) node{$0$} (.83,.9) node{$1$} (0,2.0) node{$0$} (-.83,.9) node{$0$};
\end{tikzpicture} 
& & 
\begin{tikzpicture} %5-loop
\draw (-1,0) -- (0,0) -- (1,0) -- (.5,.87) -- (0,1.74) -- (-.5,.87) -- (-1,0);
\draw[fill=white] (-1,0) circle (.1cm) (0,0) circle (.1cm) (1,0) circle (.1cm) (.5,.87) circle (.1cm) (0,1.74) circle (.1cm) (-.5,.87) circle (.1cm);
\draw (-1.2,-.25) node{$\tfrac{1}{2}$} (0,-.26) node{$0$}  (1.2,-.25) node{$\tfrac{1}{2}$} (.83,.9) node{$0$} (0,2.1) node{$\tfrac{1}{2}$} (-.83,.9) node{$0$};
\end{tikzpicture}  
\end{array}
\]
Note that the state on the left dispersion-free, hence, pure. The non-d.f. state on the right is also pure!\footnote{To see this, just note that it's the only state that vanishes at $x, y$ and $z$.} 
\end{example} 

%{\bf Exercise:} 
\begin{exercise} Let $\M$ denote the Firefly box, as described above. 
\begin{mlist} 
\item[(a)] 
Show that the four deterministic states 
and the $\tfrac{1}{2}-\tfrac{1}{2}-\tfrac{1}{2}$ state 
are the only pure states. 
\item[(b)] 
Construct an explanation for 
the non-d.f. state above, in terms of the behavior of the firefly. 
\item[(c)] {\gray {\em (Causalists only!)} What causal structure 
corresponds to the Firefly Box? Does the non-deterministic pure state requires fine-tuning?}
\end{mlist} 
\end{exercise} 

{\em Remark:} In addition to the full model on the firefly test space $\M$, we could consider the model $(\M, \Pr_{\mbox{df}}(\M))$ in which states are 
restricted to the dispersion-free, or deterministic, 
states, and (in order to maintain our standing assumption of a convex state space) 
mixtures of these. This would exclude the non-deterministic pure state.  

\begin{exercise} A probability weight on the Firefly box $\M$ is determined by its 
values on the "corner" outcomes $a, b$ and $c$, so $\Pr(\M)$ is a 3-dimensional polytope. 
Sketch this, and identify the smaller state space consisting of convex combinations of dispersion-free states.  
\end{exercise} 
%we would also need to allow mixtures of these. %Alternatively, we might be interested in a situation in %which, half the time, the firefly is unilluminated, and 
%half the time, in state $\alpha_{1/2}$, in which case ...

%{\bf Exercise:} 

%{\gray 
%\begin{exercise}[for Causalists]\label{exc:causal firefly} What causal structure %
%corresponds to the Firefly Box? Does the 
%non-deterministic pure state requires fine-tuning?
%\end{exercise} 
%}

%{\bf Example 7:} 
\begin{example}\label{ex: doubly-stochastic} Let $\M$ consist of the rows 
and columns of the set $n \times n$. The probability 
weights on $\M$ are exactly the doubly-stochastic 
$n \times n$ matrices. The {\em Birkhoff-von Neumann Theorem}  tells us that the pure probability weights are the permutation matrices, i.e., the d.f. weights. 
\end{example} 

\begin{exercise}
Let $\M_{k,n}$ be the set of rows 
and columns of $k \times n$, with $k < n$. Describe the possible probability weights. 
\end{exercise} 

{\em Remark:} The last few examples, involving test spaces in which $X$ is a finite set, should not make us lose sight of the fact that in general, even a test space having only finite tests will have an infinite outcome-set, often with significant further topological or geometric structure. In particular, both tests and outcomes will often be smoothly parameterized by quantities relating to, e.g., the physical position, orientation, temperature, etc., of some part of a laboratory aparatus relative to other parts; 
the strength of a magnetic field or a current in some part of the apparatus; time as measured by some kind of clock; and so on.

%{\bf Example 8:} 
%\begin{example}\label{ex: ensembles}  
{\gray {\bf Ensembles and Preparations} Given a probabilistic model $(\M,\Omega)$, 
We might want to construct a test space that represents the 
possible ways of {\em preparing} a state in $\Omega$. One,  
not very interesting, way to do this would be to 
define a test space $\{\{\alpha\} | \alpha \in \Omega\}$: each 
test has a single outcome, corresponding to the preparation 
of a particular state. 

Here's a more interesting approach. Given a convex set $K$, let's agree that a finite {\em ensemble} over $K$ is a finite set of pairs $(t_1,\alpha_1),....,(t_n,\alpha_n)$ 
such that $\alpha_1,...,\alpha_n \in K$, and $t_1,...,t_n$ is a list of non-negative real constants summing to $1$ --- that is, a finite probability distribution over $\{1,...,n\}$. 
If $\sum_{i} t_i \alpha_i = \alpha \in K$, we say that $\{(t_i,\alpha_i)\}$ is an 
ensemble {\em for} $\alpha$.  Let $\D(K)$ be the test space of all such finite ensembles for the convex set $K$. We think of $E = \{(t_i, \alpha_i)\}$ as a randomized preparation procedure producing one of the states $\alpha_i$ with prescribed probability $t_i$. 
Note that the outcome space here is $X(K) = (0,1] \times K$. 

I've put the proof of the following in Appendix A:

\begin{theorem}\label{thm: state of the ensemble}  The only probability weight on $\D(K)$ is the weight $\rho((t,\alpha)) = t$.
 \end{theorem} 
 
In other words, {\em the probability to get $\alpha_i$ with probability $t_i$, is $t_i$}.  (Note that this is the {\em conclusion}, not the proof, of the advertised result!)
} 

\tempout{
Let $K$ be a convex set. An {\em ensemble} over $K$ is a collection of pairs $(t_i, \alpha_i)$ such that $\alpha_i \in K$, $t_i > 0$, and $\sum_i t_i = 1$. Let $\D(K)$ be the set of all such ensembles. Note this is irredundant, thus, a test space. Notice that $\rho(t,\alpha) = t$ then defines a probability weight on $\D(K)$. 

In fact, one can show that $\rho$ is the {\em only} probability weight on $\D(K)$. See  Appendix A for details. If $K = \Omega(A)$, we can understand an ensemble $E = \{(t_i, \alpha_i)\}$ as a {\em randomized preparation procedure} that selects one of the states $\alpha_1,...,\alpha_n$ with prescribed probability $t_i$.  The uniqueness result just mentioned tells us that there is no other rational way to assign these probabilities: in other words, necessarily {\em the probability to get $\alpha_i$ with probability $t_i$, is $t_i$}.  (Note that this is the {\em conclusion}, not the proof, of the advertised result!)
\end{example} 
}
%{\bf Local Finiteness} 

\subsection{\bf Events, perspectivity, and algebraic test spaces}

It will be useful to introduce some further language borrowed from classical probability theory. 
An {\em event} for a test space $\M$ is simply an 
event in the usual probabilistic sense for one of the tests in $\M$; that is, an event is a set $a \subseteq E$
\footnote{I'm going to use 
$a,b,...$ for events, since I'll later want to use $A, B, ...$ for system labels.} 
%already taken 
%as system labels.
for some $E \in \M$. We write $\Ev(\M)$ for the set 
of all events of $\M$. If 
$\alpha$ is a probability weight on $\M$, we define 
the probability of an event $a$ in the usual way, that is,  
$\alpha(a) = \sum_{x \in a} \alpha(x)$

%\newpage
\begin{definition}\label{def: orthogonality etc for events} 
Two events $a, b \in \Ev(\M)$ are
\begin{mlist} 
\item[(a)] {\em orthogonal}, written 
$a \perp b$, iff they 
are disjoint and their union is still an event.\footnote{Note carefully that, in general, this has nothing to do with orthogonality in any geometric sense: it's just a term of art 
for a form of mutual exclusivity (though in the case of a frame test space, of course, we can take it literally).}
\item[(b)] {\em complements}, written $a \co b$, iff $a \perp b$ and $a \cup b = E \in \M$ --- that is, 
if $a$ and $b$ partition a test 
\item[(c)] {\em perspective}, written $a \sim b$, 
iff they share a complement, i.e., there exists an event 
$c$ with $a \co c \co b$. 
\end{mlist} 
\end{definition}

Every event has at least one complement, but, 
because tests can overlap, will generally have 
many complements (all of which will then be perspective 
to one another).  

\begin{exercise} Show that $\M$ is semi-classical iff 
every event has a unique complement. 
\end{exercise} 

\begin{exercise} Find an example of a test space $\M$ and events 
$a, a', b, b'$ such that (i) $a$ and $b$ are compatible, (ii) $a \sim a'$, $b \sim b'$ but (iv) $a \cap b \not \sim a' \cap b'$.  
\end{exercise} 

\tempout{
The following is easy to check (for (a), recall that we are assuming $\M(A)$ is irredundant). 

{\bf Lemma 1:} Let $a, b \in \Ev(A)$. Then 
\begin{mlist} 
\item[(a)] $a \sim b \subseteq a \Rightarrow b = a$; 
\item[(b)] $a \sim b \Rightarrow \alpha(a) = \beta(a)$ for 
all probability weights $\alpha$ on $A$. 
\end{mlist} 

{\bf Exercise:} Easily check Lemma 1. 
\tempout{
{\em Proof:}  (a) Let $c$ be complementary to both $a$ and $b$. 
Then $a \cup c = E \in \M$ and $b \cup c = F \in \M$. 
If $b \subseteq a$, then $F \subseteq E$, in which case, 
by irredundance, $F = E$, and we have $b = E \setminus c = a$. 
For (b), just note that if $a \co c$ and $c \co b$, then for any 
state $\alpha$, $\alpha(a) + \alpha(c) = 1 = \alpha(b) + \alpha(c)$. $\Box$
}
}

Notice that if $E, F \in \M$, then $E \sim F$ 
(since both are complementary to $\emptyset$). 
It is also easy to check that if $a \sim b$, then $\alpha(a) = \alpha(b)$ for all probability weights $\alpha$ on $\M$. Also, owing to irredundance, if $a \subseteq b \sim a$, we have $a = b$.

{\bf Algebraic test spaces and orthoalgebras} 
The notion of perspectivity allows one to attatch a kind of "quantum logic" to a large 
class of test spaces.

\begin{definition}\label{def: algebraic} $\M$ is {\em algebraic} iff, 
for all events $a, b, c \in \Ev(\M)$, if 
$a \sim b$ and $b$ is complementary to $c$, then $a$ is also complementary to $c$.
\end{definition} 

\begin{exercise} 
(a) Verify that all of the examples given in Section 1 are algebraic.  (b) Find an example of a 
test space in which there exist events $a, b$ and $c$ with $a \sim b$, $b \perp c$, and $a \cap c \not = \emptyset$. 
\end{exercise} 

{\gray 
\begin{exercise}[$\bigstar$] A test space $\M$ is {\em pre-algebraic} 
iff there exists an algebraic test space $\M'$ with 
$\M \subseteq \M'$. (a) Show that in this case, 
the intersection of all algebraic test spaces containing 
$\M$ is algebraic, and has the same outcome-set as $\M$. 
(b) $\M$ is {\em semi-unital} iff, for every outcome 
$x \in X = \bigcup \M$, there exists a probability 
weight $\alpha$ with $\alpha(x) > 1/2$. Let $\M$ be 
semi-unital and let $\M'$ be the set of all 
subsets of $X$ over which every probability weight 
sums to $1$. Show that $\M'$ is algebraic. 
\end{exercise} 
}

If $\M$ is algebraic, $\sim$ is an equivalence relation on $\Ev(\M)$, with the feature that, 
for all events $a,b$ and $c$, 
\begin{equation} 
a \sim b \ \mbox{and} \ b \perp c \ \ \Rightarrow \ \ a \perp c \ \mbox{and} \ a \cup c \sim b \cup c.
\end{equation}

{\gray 
\begin{exercise}[$\bigstar$] Verify this. 
That is, assuming $\M$ is algebraic,
(i) show that $\sim$ is an equivalence relation, and (ii) verify that the implication (1) holds for all $a,b, c \in \Ev(\M)$. 
\end{exercise}
}

%{\blue 
Let $\Pi = \Pi(\M) = \Ev(\M)/\hspace{-.1in}\sim$, the 
collection of equivalence 
classes of events under perspectivity. 
This carries a natural partial-algebraic 
structure, as follows. 
%setting 
%\[[a] \oplus [b] = [a \cup b] \ \mbox{when} \ %a \perp b,\]
Where $[a]$ stands for the equivalence 
class of event $a$. 
%The structure $(\Pi, \oplus)$ is then an {\em %orthoalgebra}. 
define a relation $\perp$ on $\Pi$ by 
setting $[a] \perp [b]$ iff $a \perp b$. 
Next, if $[a] \perp [b]$, set 
\[[a] \oplus [b] := [a \cup b].\]
Condition (\theequation) above guarantees that these are well-defined. We can also 
define $1 := [E]$ for any $E \in \M$, as tests 
are all equivalent under $\sim$. 
%Of course, we need to check that these are well-defined! 

%{\bf Exercise:} Let $a,a', b \in \Ev(A)$, where $A$ 
%is algebraic. Show 
%that if $a \sim a'$ and $a' \perp b$, then 
%$a \perp b$, and, in this case, 
%$a \cup b \sim a' \cup b$. 
The structures $(\Pi(\M), \perp, \oplus, 1)$ arising in this way can be characterized abstractly:
%}

%{\bf Definition:} 

\begin{definition}\label{def: orthoalgebra}{\em  An {\em orthoalgebra} \cite{FR-TP} is a structure 
$(L,\perp, \oplus, 1)$ where $\perp$ is a symmetric, irreflexive binary relation on $L$, and 
$\oplus : \perp \rightarrow L$ is a partial binary 
operation, defined for pairs $(p,q)$ with $p \perp q$, 
and $1 \in L$, such that $\forall p,q, r \in L$,  
\vspace{-.2in}
\begin{itemize} 
\item[(i)] $p \perp q \Rightarrow p \oplus q = q \oplus p$; 
\item[(ii)] $p \perp (q \oplus r) \Rightarrow 
p \perp q$, $(p \oplus q) \perp r$ and $p \oplus (q \oplus r) = (p \oplus q) \oplus r$; 
\item[(iii)] $\exists 1 \in L \forall p \in L \exists ! p' \in L$ with $p \perp p'$ and 
$p \oplus p' = 1$ \\
\end{itemize} }
\end{definition}

%{\bf Exercise:} 

\begin{exercise} If $A$ is algebraic, $\Pi(A)$ is an orthoalgebra under $1 = [E]$ ($E \in \M(A)$), and $\perp, \oplus$ as defined above. 
\end{exercise}

If $\M$ is algebraic, the orthoalgebra $\Pi(\M)$ is called the {\em logic} of $\M$. It is 
easy to see that every probability weight on $\M$ descends to a finitely-additive probability measure on $\Pi(\M)$. 

Every orthoalgebra arises as $\Pi(\M)$ for some algebraic test space $\M$.  Indeed, there is even a canonical choice for 
$\M$.  In order to define this, we need the concept of 
{\em joint orthogonality}.  Informally, a finite set 
$A = \{a_1,...,a_n\} \subseteq L$ is jointly orthogonal iff the sum $\bigoplus_{i=1}^{n} a_i$ exists. We can express this 
more exactly using the following recursive 

\begin{definition} Let $L$ be an orthoalgebra and let 
$A \subseteq L$ be pairwise orthogonal. We say that $A$ 
is jointly orthogonal, with sum $a = \bigoplus A$, iff 
for every set $B \subseteq A$ and $p \in A \setminus B$, 
$B$ is jointly orthogonal with $\bigoplus B = b \perp a$, 
in which case $\bigoplus A := b \oplus a$. 
\end{definition}

Now given an orthoalgebra, say that $E \subseteq L \setminus \{0\}$ is an {\em orthopartition} or {\em decomposition} of its unit iff 
$E$ is jointly orthogonal with $\bigoplus E = 1$. Let $\D(L)$ denote the set of orthopartitions of the unit. This is an algebraic test space, and one can show that $\Pi(\D(L))$ is isomorphic to $L$. 

{\gray 
\begin{exercise}[$\bigstar$] (a) Write down a reasonable definition of "isomorphism of orthoalgebras". (b) Show that $\D(L)$ is algebraic. (c) Construct an isomorphism 
$L \simeq \Pi(\D(L))$. 
\end{exercise} }

%{\em Remark:} 
%\newpage

{\bf Orthocoherence and orthomodular posets}
 The distinctive feature of non-Boolean orthoalgebras is that pairwise orthogonal elements need not be {\em jointly} orthogonal, meaning they can not necessarily be summed. As an example, in the logic of the triangular Firefly Box test space of Example \ref{ex: Firefly} above, $[a], [b], [c]$ (the "propositions" corresponding to the corners of the triangle) are pairwise orthogonal, but $[a] \oplus [b] = [a \oplus b]$ is not orthogonal to $[c]$, so 
"$([a] \oplus [b]) \oplus [c]$" does not exist ---  reflecting the fact that $\{a,b,c\}$ is not an event. 

%{\bf Definition:} 

\begin{definition}\label{def: orthocoherent} An orthoalgebra is {\em orthocoherent} iff for every pairwise orthogonal triple $p,q,r$, 
$p \perp (q \oplus r)$.
 \end{definition}

Any orthoalgebra carries a natural partial ordering, given by 
\begin{equation} p \leq q \ \Leftrightarrow \ \exists r \ q = p \oplus r.\end{equation}
The resulting poset $(L,\leq)$ has least element $0$ and greatest element $1$. Moreover, the mapping $p \mapsto p'$ is an orthocomplementation on $L$. 

\begin{exercise} (a) Check that (\theequation) really does define a partial ordering. 
{\gray (b) Verify that the operation $'$ an orthocomplementation (looking up the term if necessary!)}
\end{exercise} 

\begin{exercise} Show that if $a,b$ are orthogonal elements in an orthoalgebra $L$, then $a \oplus b$ is a 
{\em minimal} upper bound fo $a, b$ in the ordering defined 
above. (b) Find an example of an orthoalgebra $L$ and a 
pair of orthogonal elements $a, b$ for which $a \oplus b$ 
is not the {\em least} (that is, not the unique minimal) 
upper bound. 
\end{exercise} 

%{\blue 

%\newpage
%{\bf Digression on Quantum Logic}
 In this case, one can show 
 that $(L, \leq, ')$ is what is known as an 
 {\em orthomodular poset} \cite{AW-SEP}. 

\tempout{
In any poset $(L,\leq)$, the least 
upper bound or {\em join} of $p, q \in L$, if this exists, 
is the element $p \vee q$ such that 
\[p, q \leq r \ \Rightarrow \ p \vee q \leq r.\]
The {\em meet} of $p$ and $q$, if it exists, is their greatest lower bound: the element $p \wedge q$ with 
\[r \leq p, q \ \Leftrightarrow \ r \leq p \wedge q.\] 
$(L,\leq)$ is a {\em lattice} iff $p \vee q$ and $p \wedge q$ exist for all pairs $p, q \in L$. 

Let $(L,\leq)$ be a poset with least element $0$. A mapping $p \mapsto p'$ is an {\em orthocomplementation} on $L$ iff 
\begin{itemize} 
\item[$\bullet$] $p \leq q \Rightarrow q' \leq p'$
\item[$\bullet$] $p" = p$
\item[$\bullet$] $p \wedge p' = 0$, where 
\end{itemize} 
The last item simply requires that there is nothing 
in $L$ other than $0$ below both $p$ and $p'$. One 
can check that if we set $1 := 0'$, then $1$ is the greatest element of $L$. In particular, an orthocomplemented poset is {\em bounded}. 

{\bf Definition} An orthocomplemented poset $(L,\leq,')$ 
(a poset with a distinguished orthocomplementation) is {\em orthomodular} iff, for all $p, q \in L$, 
\[p \leq q \ \Rightarrow (q \wedge p') \vee p = q.\]
Note that part of what is asserted here is that the 
meets and joins on the right exist. 

{\bf Exercise:} Let $L$ be an OMP. Define 
$p \perp q$ iff $p \leq q'$, and in that case set 
$p \oplus q = p \vee q$. Show that $(L, \perp, \oplus)$ is 
an orthoalgebra. 
}

\begin{exercise} If $L$ is an orthoalgebra, 
$p \perp q$ iff $p \leq q'$.
\end{exercise} 

{\gray 
\begin{exercise}[$\bigstar$] If $p \perp q$, 
then $p \oplus q$ is a {\em minimal} upper bound 
for $p$ and $q$. Find an example of an orthoalgebra 
in which there exist elements $p \perp q$ in 
which $p, q$ have an additional minimal upper bound, 
besides $p \oplus b$. 
\end{exercise} 

\begin{exercise}[${\bigstar}$] Show that if an orthoalgebra $L$ is orthocoherent, $p \oplus q$ is the least upper bound of $p$ and $q$ whenever $p,q \in L$ with $p \perp q$.  
\end{exercise} 

\begin{exercise}[${\bigstar}$] Look up the definition of an orthomodular 
poset (OMP). Use the result of the previous Exercise to show that every orthocoherent orthoalgebra is one. Also show how to turn every OMP into an orthocoherent orthoalgebra.  
\end{exercise} 

\vspace{-.1in} 
{\em Remark:} Orthomodular lattices and posets were the 
prevailing models of "quantum logics" from the early 1960s through most of the 1980s. %Orthoalgebras represented the first significant generalization of OMPs. 
A natural question: when is $(\Pi(\M), \leq)$ a lattice (hence, an orthomodular one)? A sufficient condition 
is given by the {\em Loop Lemma}, due to R. Greechie and later refined by Foulis, Greechie and 
R\"{u}ttimann. See \cite{FGR} for details. \\
}

{\bf Boolean orthoalgebras and Refinement ideals} Orthoalgebras are a natural generalization of Boolean algebras. In any orthoagebra $L$, call three elements $p, q, r$ {\em jointly orthogonal} iff $p \perp q$ and $p \oplus q \perp r$ (cf. (ii) above). We say that elements $a, b \in L$ are {\em compatible} iff there exists a jointly orthogonal triple $p,q,r$ with $a = p \oplus q$ and $b = q \oplus r$. A Boolean algebra can be defined as an orthoalgebra in which any two elements are compatible. 
%{\blue [Measures? $\sigma$-BAs?]}
%}

\tempout{\blue {\gray {\bf Exercise:} Looking up the definition of a Boolean algebra if necessary, 
show that if $L$ is one, then setting 
$p \perp q \Leftrightarrow p \leq q'$ 
and in this case defining $p \oplus q = p$, 
show that $(L,\perp,\oplus,1)$ is an 
orthoalgebra, in which any two elements are 
compatible. Conversely, show that any 
orthoalgebra with this property, regarded 
as an ordered set, is Boolean.}}

\begin{exercise} Show that a finite, pairwise orthogonl subset 
of an orthoalgebra is jointly orthogonal iff the sub-orthoalgebra 
it generates is Boolean. 
\end{exercise}

If $E, F$ are tests belonging to a 
test space $\M$, we say that $E$ {\em refines}, or {\em is a refinement of},  $F$ iff, for every outcome $y \in F$, 
there exists an event $a \subseteq E$ with $a \sim \{y\}$. A test space is a 
{\em refinement ideal} iff every pair of tests have a common refinement. 

\begin{exercise}[$\bigstar$] Let $\M$ be an algebraic refinement ideal. Show that $\Pi(\M)$ is a 
boolean algebra.
\end{exercise} 

%\end{document}

\begin{exercise}[$\bigstar$] Let $\M$ be algebraic. 
\begin{mlist}
\item[(a)] Show that if $\M_o \subseteq \M$ is a refinement ideal, then so is 
$\langle \M_o \rangle$ 
\item[(b)] Let $B \leq \Pi(\M)$ be a boolean sub-orthoalgebra of $\Pi(\M)$. Show 
that there exists an algebraic refinement ideal $\M_o \subseteq \M$ with 
$B \simeq \Pi(\M_o)$.
\item[(c)] Need the refinement ideal $\M_o$ of part (b) be unique? 
\end{mlist} 
\end{exercise}

%\newpage 
{\bf Coarse-Graining} 
Broadly speaking, a {\em coarse-graining} of an experiment is a second experiment, the outcomes of which are in some relevant way equivalent to non-trivial {\em events} of the given experiment. The simplest 
case is one in which we simply treat a partition, 
$\{a_i\}$, of a given experiment $E$ into non-empty events 
$a_1, a_2,...$ {\em as} an experiment, in which, having 
performed $E$ and obtained, say, outcome $x$, we {\em record} 
only that the event $a_i$ to which it belongs, has occurred. 

It is often very helpful to enrich a test space (or model) by adjoining all coarse-grainings of this kind. 

\begin{definition}\label{def: coarsening} The {\em coarsening} of a test space, $\M$, is the test space $\M^{\#}$ consisting of all 
partitions of tests in $\M$.
\end{definition} 

Thus, an outcome for $\M^{\#}$ is a non-empty 
{\em event} for $\M$, and an event for $\M^{\#}$ 
is a jointly-orthogonal family $\{a_i\}$ of $\M$-events.  There is a natural embedding of 
$\M$ into $\M^{\#}$, namely $x \mapsto \{x\}$. It is 
convenient and harmless to identify $x$ with $\{x\}$ 
so as to regard $\M$ as a subset of $\M^{\#}$. 

\begin{exercise} 
(a) Let $\{a_i\}$ and $\{b_k\}$ be events of $\M^{\#}$, and let $a = \bigcup a_i$ and $b = \bigcup b_k$, understood as events of $\M$. Then $\{a_i\} \perp \{b_k\}$ iff $a \perp b$ and $\{a_i\} \sim \{b_k\}$ iff $a \sim b$. Conclude that $\M$ is algebraic iff $\M^{\#}$ is algebraic. {\gray (b) Show that in this case $\Pi(\M^{\#}) \simeq \Pi(\M)$.}
\end{exercise}

Probability weights on $\M^{\#}$ are essentially the same as those on $\M$: any probability weight on $\M$ already assigns probabilities to events, and this gives a 
probability weight on $\M^{\#}$; conversely, every probability weight on $\M^{\#}$ restricts to one on $\M$, and is determined by this restriction. Thus, it's natural (and harmless) to identify $\Pr(\M^{\#})$ with $\Pr(\M)$. 
Under this convention, we define the coarsening of a model $A$ by $\M(A^{\#}) = \M(A)^{\#}$ and $\Omega(A^{\#}) = 
\Omega(A)$. 

\begin{exercise}  The {\em finite coarsening} of $\M$ 
is the sub-test space $\M^{\#}_{o} \subseteq \M$ consisting 
of finite such partitions. (a) Show that if $\M$ is algebraic, 
 $\M^{\#}_{o}$ and $\M^{\#}$ have canonically isomorphic logics. 
(b) Given an example to show that $\Pr(\M^{\#}_{o})$ can 
be properly larger than $\Pr(\M^{\#}) = \Pr(\M)$. 
\end{exercise}

\subsection{Sequential Measurement and Compounding}

%\end{document}
Suppose $\A$ and $\B$ are two test spaces. If $E \in \A$ and $F : E \rightarrow \M$, 
define a {\em two-stage} test having outcome-set 
\begin{equation}
\bigcup_{x \in E} \{x\} \times F_{x}
\end{equation}
by the following rule: perform $E$; if the outcome 
secured is $x$, perform the test $F_{x}$, and, if 
this yields outcome $y$, record $(x,y)$ as the outcome 
of the two-stage test. 

The collection of all outcome-sets 
having the form (\theequation) is called the 
{\em forward product} of $\A$ and $\B$, and is denoted 
$\for{\A\B}$. A probability 
weight on $\for{\A\B}$ is uniquely determined by an 
{\em initial} weight $\alpha \in \Pr(A)$ and a {\em transition function} $\beta : X(A) \rightarrow \Omega(B)$, by 
the recipe 
\[(\alpha;\beta)(x,y) := \alpha(x)\beta_{x}(y).\]

\begin{exercise} Show that every function of the 
form $(\alpha;\beta)$ is a probability weight on 
$\M(\for{AB})$, and that every probability weight 
on $\M(\for{AB})$ has this form for a unique $\alpha$ and 
some $\beta$, uniquely determined on $\supp(\alpha) = \{x \in X(A) | \alpha(x) > 0\}$, and 
otherwise arbitrary.
\end{exercise} 

\begin{definition}\label{def: forward product} The 
{\em forward product} of probabilistic models $A$ and $B$, denoted $\for{AB}$, has 
test space $\M(\for{AB}) = \for{\M(A)\M(B)}$ and state space $\Omega(\for{AB})$ consisting of all probability weights on $\M(\for{AB})$ of the form 
$(\alpha;\beta)$ with $\alpha \in \Omega(A)$ and 
$\beta \in \Omega(B)^{X(A)}$. 
\end{definition} 

Evidently, if $\omega = (\alpha;\beta)$, then 
\[\sum_{y \in F} \omega(x,y) = \alpha(x)\]
and, if $\alpha(x) \not = 0$, 
\[\frac{\omega(x,y)}{\alpha(x)} = \beta(y).\]
Thus, we also call $\alpha$ the {\em marginal} 
state and $\beta_{x}$, the conditional state given $x$.
If we agree to write $\alpha = \omega_1$ and $\beta_{x} = \omega_{2|x}$, 
then the recipe for $(\alpha;\beta)$ reads 
\[ \omega(x,y) = \omega_{1}(x) \omega_{2|x}(y).\]
Hence, we have 
\begin{equation}
\omega_{2} \ = \ \sum_{x \in E} \omega_{1}(x) \omega_{2|x},
\end{equation} 
where $E$ is any test in $\M(A)$. This is a version of the {\em Law of Total Probability.}

{\bf Interference} 
The forward product offers some insight 
into the idea of interference, often regarded as particularly characteristic of quantum theory.  Let 
$a, b, c \in \Ev(A)$.  It's not hard 
to see that if $a \sim b$, then $c \times a \sim c \times b$ in $\Ev(\for{AA})$. However, in general we do {\em not} have $a \times c \sim b \times c$. This is evident from the fact that 
$\omega_{2|a}(c)$ depends on $\{ \, \beta_{x} \, | \, x \in a \,\}$, $\omega_{2|b}(c)$ depends on $\{ \, \beta_{y} \, | \, y \in b\, \}$, and these sets will generally be very different unless $a = b$.  

As an extreme case of this, suppose there exists some 
{\em outcome} $x \in X(A)$ with $a \sim x_{a}$. Then 
there is no guarantee that $\omega_{2|x} = \omega_{2|a}$. 
Indeed, we generally have 
\[\omega(x_a y) \not = \sum_{x \in a} \omega(xy).\]
We say that $\omega$ exhibits {\em interference} among the 
outcomes $x \in a$. See \cite{Wright} for a more detailed discussion of this point. 

\begin{exercise} Find a concrete example to show that 
$a \sim b$ does not entail $ac \sim bc$ in 
$\Ev(\for{AB})$. 
\end{exercise} 

\begin{exercise}  Show that $a \sim b$ {\em does} entail that 
$ca \sim cb$ for all $c,a,b \in \Ev(\M)$. 
\end{exercise}

We can apply the forward-product construction iteratively 
to obtain a test space that is closed under the formation 
of  branching, sequential tests of arbitrary (finite) length. 

\begin{definition}\label{compounding} If $A$ is a model, define a model 
$A^{c}$, the {\em compounding} of $A$, as follows. 
$X(A^{c})$ is the free monoid, $X(A)^{\ast}$, over $X(A)$. 
This simply means that $X(A^{c})$ consists of finite strings 
of outcomes, regarded as a semigroup under concatenation, and 
with the empty string, which I'll denote by $1$, as a unit. 
We can identify $X(A)$ with the set of one-entry strings, so 
that $X(A) \subseteq X(A)^{\ast}$. 
Call a subset $\G$ of $X(A)^{\ast}$ {\em inductive} iff, 
for all sets $E \in \G$ and any function $F : E \rightarrow \G$, $\bigcup_{x \in E} xF_{x} \in \G$.  The intersection 
of inductive subsets of $X(A)^{\ast}$ is inductive, so 
every collection $\G_o$ of subsets of $X(A)$ is contained 
in a smallest inductive family, 
\[\langle \G_o \rangle = 
\bigcap \{ \G | \G_o \subseteq \G \subseteq \Power(X(A)^{\ast},\ \G \ \mbox{inductive} \}.\]
We now define $\M(A^c)$ to be the inductive family 
generated by $\M(A)$. One can show that a probability 
weight on $\M(A^{c})$ is uniquely determined by a 
weight $\alpha$ on $\M(A)$ and a function 
$\beta : X(A)^{\ast} \to \Pr(\M(A^{c}))$ recursively 
by 
\[\omega(ax) = \omega(a)\beta_{a}(x)\]
with $\beta(1) = \alpha \in \Pr(\M(A))$.  We 
define $\Omega(A^c)$ to consist of those probability 
weights $\omega$ with $\beta(a) \in \Omega(A)$ for all 
strings $a \in X(A)^{\ast}$. 
\end{definition} 

Orthogonality in $A^{c}$ is lexicographic, in the 
sense that if $a, b \in X(A)^{\ast}$, we have 
$a \perp b$ iff $a = u x v$ and $b = u y w$ where 
$u, v, w \in X(A)^{\ast}$ and $x, y \in X(A)$ with 
$x \perp y$.

% It is not quite injective, since 
%$\phi(e,y) = y$ and $\phi(x,e) = x$ for all $x, y \in X^{\ast}$. 
%{\blue [Is this actually true??]}

The following result \cite{RJF} is largely forgotten, but very interesting.  

\begin{theorem}[Randall, Janowitz and Foulis,1973]\label{thm: RJF} Let 
$\M$ be any semi-classical test space having at least two tests with at least two outcomes each, and let $L = \Pi(\M^{c})$.  Then $L$ is a complete, irreducible, atomless, non-Boolean orthomodular poset. Moreover, every interval $[0,a]$ in $L$ is isomorphic to $L^{\kappa}$ for some power $\kappa$. 
\end{theorem}  %{\blue [Check: finite?}}

%{\blue [Do we first need to know that $\M^{c}$ is coherent and %regular?]}

Thus, one obtains very rich non-classical "logics" 
on the basis of elementary operational constructions, 
as long as one starts with two or more incompatible 
experiments. 

%\end{document}

\tempout{
{\bf Coarse-Graining}  A natural operation on a classical 
experiment $E$ is to {\em coarse-grain} it by treating 
certain events as though they were single outcomes. In other 
words, one replaces the outcome-set $E$ with a {\em partition} $\{a_1,...,a_n\}$ of $E$ into non-empty 
subsets.    

{\bf Definition:} The {\em coarsening} of a test space 
$\M$ is the test space $\M^{\#}$ consisting of all partitions 
of tests $E \in \M$. 

Thus, outcomes for $\M^{\#}$ are non-empty {\em events} for $\M$: 
\[X(\M^{\#}) = \Ev(\M) \setminus \{\emptyset\}$. 

It is not hard to show that a probability weight 
on $\M^{\#}$ satisfies $\alpha(a) = \sum_{x \in a} \alpha(\{x\})$ for any nonempty event $a$. Hence, such a weight is 
uniquely determined by a probability weight on $\M$, and 
every probability weight on $\M$ defines one on $\M^{\#}$. 
In effect (that is, ignoring niceties of type), $\Pr(\M) = \Pr(\M^{\#})$.  

{\bf Definition:} The coarsening of a model $A$ 
is $^{\#}$ where $\M(A^\#) = \M(A)^{\#}$ and $\Omega(A^{\#}) 
= \Omega(A)$.  

Notice that $\V(A) \simeq \V(A^{\#})$ and $\E(A) \simeq \E(A^{\#})$. ... [Goes below]

The mapping $x \mapsto \{x\}$ gives us an embedding 
of $A$ in $A^{\#}$. 
}

%\end{document}

%{\blue [Do FR TP here??]}

\subsection{Mappings of Models} 

In order to do anything much with our probabilistic models, 
we need an appropriate notion of a mapping between models. 
There are lots of options, and we'll explore several. But 
the following is good for all-'round purposes: 

\begin{definition}\label{def: morphism} A {\em morphism} from a probabilistic 
model $A$ to a probabilistic model $B$ is a mapping 
$\phi : X(A) \rightarrow X(B)$ 
such that 
\begin{itemize} 
\item[(i)] $x \perp y \Rightarrow \phi(x) \perp \phi(y)$ for all $x, y \in X(A)$; 
\item[(ii)] $a \in \Ev(A) \Rightarrow \phi(a) \in \Ev(B)$
\item[(iii)] $a \sim b \Rightarrow \phi(a) \sim \phi(b)$ 
\item[(iv)] For every state $\beta \in \Omega(B)$, there 
is a state $\alpha \in \Omega(A)$ and a scalar $t \geq 0$ such that for all $x \in X(A)$, 
$\beta(\phi(x)) = t\alpha(x)$. 
\end{itemize} 
\end{definition} 

A morphism from one test space $\M$ to another, $\M'$, 
will be understood to mean a morphism between the corresponding full models. In 
this case, condition (iv) follows automatically from condition (iii).  

Condition (i) can be rephrased as saying that 
$\phi$ is {\em locally injective}, i.e., injective 
on every test.   

%{\bf Exercise:} 
\begin{exercise} Show that if $\phi$ is a morphism 
and $a, b \in \Ev(A)$, $a \perp b \Rightarrow \phi(a) \perp \phi(b)$. 
\end{exercise} 

%It follows that for events 
%$a, b \in \Ev(A)$, if $a \perp b$, then 
%$\phi(a) \cap \phi(b) = \emptyset$ as well.   
%Since $\phi(a) \cup \phi(b) = \phi(a \cup b)$ and 
%the latter is an event, we have $\phi(a) \perp \phi(b)$. 

Condition (ii) does not require the image of a test 
to be a test, but only an event of $B$. However, 
condition (iii) requires that all events of the form 
$\phi(E)$ where $E \in \M(A)$ be {\em equi-probable} in 
all states of $B$: If $\beta \circ \phi = t\alpha$ where $\alpha \in \Omega(A)$, then $\beta(\phi(E)) = \sum_{x \in E} t \alpha(x) = t$. Indeed, 
$\beta \circ \phi$ is a positive weight on $\M(A)$. Condition (iv) requires 
that the corresponding normalized weight belong to $\Omega(A)$. 
%If $t \not = 0$, then we 
%can define a probability weight on $\M(A)$ by setting 
%\[\alpha(x) = t^{-1}\beta(\phi(x)).\]
%Condition (c) further requires that $\alpha$ belong 
%to $\Omega(A)$.  

%We'll say that $\phi$ is {\em test-preserving} iff 
%$\phi(E) \in \M(B)$ for every test $E \in \M(A)$. 
%(Note this makes (ii) redundant). 

%\tempout{
By way of examples: (i) if $(S,\Sigma)$ and $(S', \Sigma')$ are measurable spaces and $f : S \rightarrow S'$ is a measurable surjection, then the preimage map $f^{-1} : \Sigma' \rightarrow \Sigma$ gives us a morphism $\M(S',\Sigma') \rightarrow \M(S,\Sigma)$.  (The surjectivity condition can be relaxed, but at the cost of allowing {\em partial} morphisms.)
%See Section \ref{label:partial morphisms}) 
(ii) Let $U : \H \rightarrow \K$ be an isometry (not necessarily surjective) from a Hilbert space $\H$ to a Hilbert space $\K$, and let $\phi_{U}$ denote the corresponding mapping from the unit sphere of $\H$ to that of $\K$. 

\begin{exercise} Check that $\phi_{U}$ is a morphism from $\F(\H)$ to $\F(\K)$. \end{exercise}

If $g : A \rightarrow B$ and $f : B \rightarrow C$ are morphisms, so is $f \circ g$. For any object $A$, the identity mapping $\id_{A} := \id_{X(A)} : X(A) \rightarrow X(A)$ is a morphism. Thus, probabilistic models and 
morphisms define a concrete {\em category}, which I will 
call $\Prob$.  
%(Later, it will be helpful to consider a 
%larger class of {\em partial morphisms}; these will give 
%us a second category, $\Prob_{o}$. Among other things, 
%this will allow us to lift the surjectivity requirement in %example (i) above. )

We will be interested below in some special classes of morphisms. Specifically, 

%\noindent{\bf Definition:} 
\begin{definition}\label{def: types of morphisms} A morphism $A \rightarrow B$ is 
\begin{mlist} 
%\item[(i)] {\em strong} iff $\phi(E) \sim \phi(F)$ 
%for all $E, F \in \M(A)$, and 
\item[(i)] {\em test-preserving}, or an {\em interpretation}, iff 
$\phi(E) \in \M(B)$ for every test $E \in \M(A)$, and 
\item[(ii)] an {\em embedding} iff test-preserving and 
(globally) injective. 
\item[(iii)] A {\em faithful} embedding iff an embedding 
with $\phi^{\ast} : \Omega(B) \rightarrow \Omega(A)$ surjective.  
%$\phi : X(A) \rightarrow X(B)$ is an injection. %and $\phi^{\ast} : \V(B)_+ \rightarrow \V(A)_+$ is surjective. 
\end{mlist} 
%Note that the last condition above says that if $\alpha \in \Omega(A)$, there exists some 
%state $\beta \in \Omega(B)$ and some scalar $r > 0$ with $\alpha = r \phi^{\ast}(\beta)$. 
%If $\phi$ is both an embedding and test preserving, then $%\phi^{\ast}$ sends normalized 
%states to normalized states; equivalently, $\phi^{\ast %\ast} : \V(A)^{\ast} \rightarrow \V(B)^{\ast}$ takes $u_A$ to $u_B$.  
\end{definition} 

\begin{exercise} 
Let $\phi : X^{\ast} \times X^{\ast} \rightarrow X^{\ast}$ be given by $\phi(x,y) = \phi(xy)$. Check that this defines a test-preserving morphism 
$\for{\M(A^c)\M(A^c)} \rightarrow \M(A^c)$. \end{exercise}

It is straightforward that a $\perp$-preserving and test-preserving mapping $\phi : X(A) \rightarrow X(B)$ automatically preserves events and perspectivity, and hence, is a morphism.  The following observation will also be useful:

\begin{lemma} 
Let $\phi : A \rightarrow B$ be an embedding, and let $\psi : B \rightarrow A$ be a morphism with $\psi \circ \phi = \id_{B}$. Then $\psi$ is test-preserving.
\end{lemma} 

{\em Proof:} If $\alpha \in \Omega(A)$ and $E \in \M(A)$, we have $\phi(E) \in \M(B)$, so 
\[\psi^{\ast}(\alpha)(\phi(E)) = \alpha(E) = 1.\]
It follows that $\psi^{\ast}(\alpha)$ is a probability 
weight in $\Omega(B)$. Thus, if $F \in \M(B)$, we have 
\[\alpha(\psi(F)) = \psi^{\ast}(\alpha)(F) = 1,\]
so (since by our standing assumption, $\Omega(A)$ is 
positive), $\psi(F) \in \M(A)$. $\Box$

\tempout{
{\blue 
\noindent{\em Remark:} One can also consider what we might call {\em partial} morphisms of models, in which 
$\phi(x)$ is undefined for some outcomes $x \in X(A)$. 
What follows can be adapted to this more general notion, 
but we won't pursue this here.} 
\tempout{
In this case, define the {\em support} of $\phi$ 
to consist of those points at which $\phi$ is defined. 
If $\phi$ is strong, $\M_{S} = \{ E \cap S | E \in \M(A)\}$ will be irredundant, $a \sim b in \Ev(\M(A))$ implies 
$a \cap S \sim b \cap S$. We can then define 
a model $A_{S}$ with test space $\M(A_S) = \M_{S}$ 
and state space $\Omega(A_S)$ consisting of 
restrictions of normalized non-zero restrictions of 
states in $\Omega(A)$ to $S$. (Note that these 
restrictions all have the same weight, so normalization 
is possible). Then the obvious partial mapping 
$X \rightarrow S$ defines a partial morphism 
$A \rightarrow A_{S}$. 
$A \rightarrow A_{S}$ given by $x \mapsto x$ for 
$x \in S$ is a strong morphism, and $\phi : A \rightarrow B$ factors through this to give $\phi_{S} : A_{S} \rightarrow B$. {\blue [clean up]}
}
}

{\bf Symmetry}  A {\em symmetry}, or automorphism, of a model $A$ is a bijection 
$\phi : X(A) \rightarrow X(A)$ with $\phi(\M(A)) = \M(A)$ and $\phi^{\ast}(\Omega(A)) = \Omega(A))$. We write $\Aut(A)$ is the group of symmetries of $A$.  One can introduce a notion of dynamics by considering one-parameter groups of symmetries, i.e., 
homomorphisms from the group $(\R,+)$ into $\Aut(A)$. 
One can add to the structure of a model $A$ a preferred symmetry group $G(A)$, constraining the possible dynamics. One can also add topological structure, and then a natural choice is to consider the group of continuous symmetries with continuous inverses. For more on this, see \cite{AW-handbook}

{\bf Observables} 
One plausible way of modeling an {\em observable} on a probabilistic model $A$ is as an interpretation --- a test-preserving morphism --- $\phi : B \rightarrow A$, where $B = (\B_{\sigma}(S,\Sigma),\Omega)$ is the  Borel model associated with some measurable space $(S,\Sigma)$ of "values" (Real or otherwise). Thus, if $b \in \Sigma$ is a measurable subset of $S$, $\phi(s)$ would be a physical event belonging to one of the experiments in $\M(A)$, with the interpretation that obtaining an outcome in $a = \phi(b)$ "means" that the observable has a value in $b$.  There is a potential ambiguity here, but this is resolved by the following 

\begin{exercise} Show that an interpretation $\phi : \B_{\sigma}(S,\Sigma) \rightarrow \M$ is an interpretation from a Borel test space into any test space is injective on non-empty events. 
That is, if $b_1, b_2 \not = \emptyset$, then $\phi(b_1) = \phi(b_2) \Rightarrow b_1 = b_2$. 
\end{exercise} 

If $\phi : B \rightarrow A$ is an $S$-valued observable as above, and $f : S \rightarrow \R$ is a Borel-measurable real-valued random variable on $S$, we obtain an observable $f(\phi) := \phi \circ f^{-1}$ where we understand $f^{-1}$ as an interpretatation 
$\B_{\sigma}(\R, \mbox{Borel field of } \R) \rightarrow \B_{\sigma}(S,\Sigma)$ in the obvious way.  If $\alpha$ is a state of $A$, then $(f(\phi)^{\ast}(\alpha)) =: \lambda$ is a 
Borel measure on $\R$.  When the identity function on $\R$ is integrable with respect to $\lambda$ (e.g., if $f$ is bounded, so that $\lambda$ has bounded support), we can then define the expected value of $f(\phi)$ in state $\alpha$ by $\Exp_{\alpha}(f(\phi)) = \int_{\R} x d\lambda$.  In the case of where $A$ is a projective quantum model, or more generally a von Neumann model, an interpretation $\phi : B_{\sigma}(S,\Sigma) \rightarrow \M(A)$ is 
essentially the same thing as a projection-valued measure, and the construction sketched here reproduces the usual way of handling quantum observables and their expected values. 

{\gray 
%\begin{exercise} 
{\bf Digression: Event-valued morphisms} 
An {\em event-valued} morphism $A \rightarrow B$ is simply a morphism 
$A \rightarrow B^{\#}$.\footnote{For those who know this lingo: $( \cdot )^{\#}$ is a monad in $\Prob$, and event-valued morphisms are morphisms in the corresponding Kleisli category \cite{AW-CCaM}   In the work of Foulis and Randall, "morphisms" were understood to be event-valued by default.} Notice that $\phi(x)$ is permitted to be empty 
for some $x \in X(A)$. This gives us a way of handling 
partial morphisms: they are simply event-valued morphisms 
such that $\phi(x)$ is a singleton, if non-empty. 
%A test-preserving event-valued morphism is called an {\em %interpretation} from $A$ to $B$. 

\begin{exercise} 
Explain how one should compose two event-valued morphisms  $A \rightarrow B$, $B \rightarrow C$ to obtain an event-valued morphism $B \rightarrow C$. Check that this composition rule is associative, and that the composition of two interpretations is an interpretation.\end{exercise}

%Alternatively, one could consider a morphism $\psi$ in the usual sense from a sub-model $A_o$ of $A$ to $B$: this will map  tests $E \in \M(A_o)$ (which we understand as those tests 
%that "measure" the observable) to partitions of some fixed Borel set $b \in \Sigma$ 
%with the interpretation that obtaining an outcome $x \in E$ means that the observable 
%has a value in $b = \phi(x)$. }

\begin{exercise} The {\em support} of an event-valued morphism $\phi : A \rightarrow B^{\#}$, denoted $S_{\phi}$, is the set of outcomes $x \in X(A)$ with $\phi(x) \not = \emptyset$. Show that $\M_{\phi} := \{ E \cap S_{\phi} | E \in \M(A)\}$ is irredundant, hence, a test space in its own right.  Also show that the inclusion mapping $S_{\phi} \rightarrow X(A)$ 
defines a morphism (in the usual sense) from $\M_{\phi}$ to $\M$.  
%Post-selection on success of $\phi$!
\end{exercise} 
}

%{\gray 
%\begin{exercise} An {\em event-valued} morphism $A \rightarrow B$ is a morphism $A \rightarrow B^{\#}$. An outcome-preserving event-valued morphism is called an {\em interpretation} from $A$ to $B$. Explain how one should compose two event-valued morphisms  $A \rightarrow B$, $B \rightarrow C$ to obtain an event-valued morphism. $B \rightarrow C$. Check that this composition rule is associative, and that the composition of two interpretations is an interpretation.\footnote{For those who know this lingo: $( \cdot )^{\#}$ is a monad in $\Prob$, and event-valued morphisms are morphisms in the corresponding Kleisli category.}\end{exercise} }

%{\bf Appendix: }
\tempout{
{\gray {\bf Ensembles and prepartions}

Given a probabilistic model $(\M,\Omega)$, 
We might want to construct a test space that represents the 
possible ways of {\em preparing} a state in $\Omega$. One,  
not very interesting, way to do this would be to 
define a test space $\{\{\alpha\} | \alpha \in \Omega\}$: each 
test has a single outcome, corresponding to the preparation 
of a particular state. 

Here's a more interesting approach. Given a convex set $K$, let's agree that a finite {\em ensemble} over $K$ is a finite set of pairs $(t_1,\alpha_1),....,(t_n,\alpha_n)$ 
such that $\alpha_1,...,\alpha_n \in K$, and $t_1,...,t_n$ is a list of non-negative real constants summing to $1$ --- that is, a finite probability distribution over $\{1,...,n\}$. 
If $\sum_{i} t_i \alpha_i = \alpha \in K$, we say that $\{(t_i,\alpha_i)\}$ is an 
ensemble {\em for} $\alpha$.  Let $\D(K)$ be the test space of all such finite ensembles for the convex set $K$. Note that the outcome space here is $X(K) = (0,1] \times K$. 

I've put the proof of the following in Appendix A:

\begin{theorem}\label{thm: state of the ensemble}  The only probability weight on $\D(K)$ is the weight $\rho((t,\alpha)) = t$.
 \end{theorem} 
}

\tempout{
We also seem to be proving this: if $f : [0,1]$ is $\Q$-linear and non-decreasing (more generally, monotone), it's continuous (Proof: Let $q_n, r_n$ be rationals decreasing, resp. increasing,  to $1$: then $f(q_n x) \rightarrow f(x)^{+}$ and 
$f(q_n x)  = q_n f(x) \rightarrow f(x)$, so $f(x) = f(x)^{+}$; similarly $f(r_n x) \rightarrow f(x)^{-}$ and 
$f(r_n x) = r_n f(x) \rightarrow f(x)$, so $f(x) = f(x)^{-}$. So $f(x)^{+}$ and $f(x)^{-}$ coincide, and 
$f$ is continuous at $x$.)
}
}

%\end{document}
%\newpage
\section{Linearized Models and Effect Algebras
%\footnote{Apologies: this section is still pretty rough!}
} 

{\fontsize{8}{9}\selectfont
\epigraph{\fontsize{8}{9}
%There exists a remarkable contrast between the variety of physical phenomena and
%the homogeneity of a unit sphere. 
It seems strange that all physical situations could be represented by points of [a unit sphere] %such a symmetric structure. 
%A question thus arises whether quantum mechanics should not be generalized by %introducing some more general spaces
}
{\fontsize{7}{9} B. Mielnik \cite{Mielnik}}
}

%The role of convex set theory for quantum mechanics was for acertain time overshadowed by "quantum logic" and lattice theory;it was duly recognized only in a few papers (we mostly refer to [10, 6, 3]) which are, however, still under a dominating influence of lattice theoretical ideas. In this article we take one more step towards abandoning thelattice theoretical approach and we express physics of quanta exclusivelythrough the geometry of convex figures.

The framework sketched thus far can to a large extent be 
"linearized", so that outcomes, events, and states are represented by elements of suitable vector spaces, and morphisms become linear mappings between these.  Moreover, these are {\em ordered} vector spaces of special types: states end up living in what are called {\em base-normed} spaces, and effects, in their duals, which are {\em order-unit spaces}. Morphisms become positive linear mappings taking effects to effects. 

This ordered-linear setup (pioneered in the 1960s and 70s in by Davies and Lewis \cite{DL}, Edwards \cite{Edwards}, \cite{Ludwig}, and Mielnik \cite{Mielnik} among others; see also \cite{FR-ELPS}) is sufficient for many purposes, and has become the standard setting for GPTs. See \cite{Plavala} for a detailed introduction to GPTs in this style, and \cite{BW, AW-handbook} for more on how this articulates with the framework adopted here. 

In this chapter, I'll begin with a brief tutorial on ordered vector spaces 
and their connection with convex sets.  It's tempting to restrict attention to finite-dimensional spaces, but it's hard to bring this off: first, because if even in quantum theory we need infinite dimensionality to allow for continuous observables, and secondly, because simple constructions like compounding ($A \mapsto A^{c}$) take us from finite-dimensional to  infinite-dimensional models. %[Belongs later!] 

%That said, I'm going to simplify life considerably by restricting attention to probabilistic models $A$ such that the state-space $\Omega(A)$ is compact in the product topology on $\R^{X}$.  As mentioned in Section 1, this is automatic if $\Omega(A)$ is closed and $\M(A)$ is locally finite. 
In order to simplify life, in what follows I will assume that $\Omega(A)$ is large enough to {\em separate points} of $X(A)$ --- that is, for distinct outcomes $x, y \in X(A)$, there exists a state $\alpha \in \Omega(A)$ with $\alpha(x) \not = \alpha(y)$.

\subsection{Ordered Linear Spaces} 

It seems prudent to start with a crash-course on ordered vector spaces. A good source of general information on this subject is the book by Aliprantis and Tourky \cite{AT}. The books by Alfsen \cite{Alfsen} and Alfsen and Shulz \cite{AS1, AS} are more advanced (particularly the former), but also  more focused on the material that we'll need. In order to avoid too lengthy a digression, I've consigned some of the details of what follows to Appendix B.

%\newpage
\begin{definition}\label{def: cone} A (convex) {\bf cone} in a real vector space 
$\V$ is a set $K \subseteq \V$ closed under 
addition an multiplication by non-negative scalars: 
\[a, b \in K, t \geq 0 \ \Rightarrow \ ta, a + b \in K.\]
An immediate consequence of this is that $K$ 
is convex.
\end{definition} 

If $K$ is a cone, so is $-K = \{ -x | x \in K\}$. 
One says that $K$ is {\em pointed}, or a {\em proper cone}, iff $K \cap -K = \{0\}$, 
and {\em generating} iff $K$ spans $V$ ---  equivalently, iff 
$\V = K - K := \{ x - y | x, y \in K\}$. 

If $K$ is a pointed, generating cone, define 
\[a \leq_{K} b \ \Leftrightarrow \ b - a \in K.\]
It is easy to check that this is a partial order on $\V$, that  $K = \{ a \in V | a \geq 0\}$, 
and that 
\[a \leq b \Rightarrow a + c \leq b + c \ \mbox{ and } \ 
ta \leq tb\]
for all $a,b,c \in \V$ and all scalars $t \geq 0$. Conversely, any partial ordering satisfying this last pair of conditions determines a cone $K = \{ a | a \geq 0\}$, and then the given order relation $\leq$ coincides with $\leq_{K}$. 

\begin{exercise}[easy!] Verify all this. \end{exercise} 

%(Proof: if $b - a \in K$, then $b - a \geq 0$, so 
%$b \geq 0 + a = a$. So $b \geq_{K} a$ implies $b \geq a$. %Conversely, if $b \geq a$, then $b - a \geq a - a \geq 0$, 

\begin{definition}\label{def: ordered vector space} An {\em ordered vector space} is a pair 
$(\V, \V_{+})$ where $\V_{+}$ is a designated pointed, generating cone, called the {\em positive cone} of $\V$. 
\end{definition} 

\begin{example} The obvious source of examples is function spaces. If $\V \leq \R^{X}$ for some set $X$, the {\em natural} cone for $\V$ is $\V_{+} = \{ f \in \V ~|~ f(x) \geq 0 ~\forall x \in X\}$. We also say that $\R^{X}$, with this cone, is {\em ordered pointwise}.  \end{example} 
%If $\V \leq \R^{X}$ (that is, $\V$ is a subspace of $\R^{X}$, then we say that $\V$ is ordered pointwise iff $\V_+ = \V \cap (\R^{X})_{+}$.}

\begin{example}  If $\H$ is a Hilbert space, and let 
$\L_{s}(\H)$ denote the real vector space of bounded, self-adjoint operators on $\H$. That is, an operator $a$ belongs to $\H$ iff it's defined on all of $\H$, and satisfies $\langle ax, y \rangle = \langle x, ay \rangle$ for all vectors $x,y \in \H$. An operator $a \in \L_s$ is {\em positive} iff $a = b^{\ast} b$ for some operator $b$ on $\H$. One can show that this is equivalent to saying that $\langle ax, x \rangle \geq 0$ for all unit vectors $x \in \H$. Using this, one shows that the set of positive operators form a pointed, generating cone for $\L_{s}(\H)$, so we can, and will, regard the latter as an ordered vector space. 
\end{example} 

{\gray {\em Remark:} This generalizes to any 
$C^{\ast}$-algebra $\mathfrak A$: one says that 
$a = a^{\ast}$ in $\mathfrak A$ is positive iff it has 
the form $a = b^{\ast} b$. Again, the positive elements 
form a cone making ${\mathfrak A}_{s}$ an ordered vector space. 
}

\tempout{\gray Remark Let $S = S(\H)$ be $\H$'s unit sphere. The {\em quadratic form} associated with an operator $a$ n $\H$ is the function $q_{a} : S \rightarrow \C$ given by $q_{a}(x) = \langle ax, x \rangle$. For a complex Hilbert space $\H$, every operator is determined by its quadtratic form; for a real Hilbert space, the same is true for self-adjoint (though not for arbitrary) operators. In either case, $a \mapsto q_{a}$ provides a linear embedding of $\H$ into $\C^{S(\H)}$, taking the positive cone defined above to the pointwise cone. So in this sense, the operator-theoretic ordering is still a pointwise ordering. }

{\bf The Archimedean property} An ordered vector space $\V$ is {\em Archimedean} iff, 
for all $x > 0$ and every $y \in \V$, there is some $n \in \N$ with $y < nx$. Alternatively, 
if $x, y \in \V$ and $nx \leq y$ for all $n \in \N$, then $x \leq 0$. 

The classic example of a non-Archimedean ordered vector space is $\R^2$ the the {\em lexicographic} order: $(x,y) \leq (u,v)$ iff either $x \leq u$ {\em or} $x = u$ and $y \leq v$. The positive cone consists of the right half-plane, excluding the negative real axis.\footnote{\gray Note that this is a {\em linear} order on $\R^{2}$: given $(x,y), (u,v) \in \R^{2}$, either $(x,y) \leq (u,v)$ or $(u,v) \leq (x,y)$. It turns out that {\em no} linear ordering on $\R^{n}$ can be Archimedean for any $n > 1$.}

\begin{exercise}  (a) Show that the positive cone of $\R^{2}$ in the lexicographic order consists of the right half-plane, excluding the negative real axis.  (b) Give an example of points $x, y$ with $x \not \leq 0$ and $nx \leq y$ for all $n$. 
\end{exercise} 

As this might suggest, non-Archimedean ordered linear spaces are somewhat unfriendly. Luckily, most of the spaces one meets in practice are indeed Archimedean:

\begin{lemma} If $\V$ has a Hausdorff linear topology, 
%in which $\V_{+}$ is closed. Then $\V$ is Archimedean. 
then $\V$ is Archimedean if $\V_{+}$ is closed. 
\end{lemma} 

\begin{exercise} Prove this. (Hint: Show that if 
$\V_{+}$ is closed, so is the order relation as a subset of 
$\V \times \V$).\end{exercise} 

\tempout{
{\bf Exercise:} Show that the converse is also true: If 
$\V$ has a Hausdorff linear topology, an Archimedean cone 
in $\V$ is necessarily closed. (Hint:  If $x$ is a limit point of $(\V_{+})$, then for any $y \in \V$ exists some 
$N \in \N$ with $x + (1/n)y \in \V_{+}$ for all $n \geq N$. Hence, 
$-x \leq (1/n)y$, or $n(-x) \leq y$, for any $n \geq N$. Since $\V$ is Archimedean, $-x \leq 0$, so $x \geq 0$, i.e., 
$x \in \V_{+}$. Thus, $\V_{+}$ is closed. 
}

%For the converse, let $\V_+$ be closed. If 
%$nx \leq y$ for all $n$, then $x \leq 1/n y$ for all $n$. Since $\V_+$ is closed, the order relation respects limits, so $x \leq 0$.   $\Box$

%{\blue [Question: Where'd we use Hausdorffness? ]}

{In what follows I'm going to assume 
that \emph{\bf all ordered vector spaces under consideration carry a 
locally convex, Hausdorff linear topology with respect to which 
$\V_+$ is closed.} This is automatic (and the topology 
uniquely determined) if the space is finite-dimensional, and 
is also true for every concrete space we'll be interested in. 
%(almost always a normed space or else a dual spaces in its 
%weak-$\ast$ topology). 
\emph{\bf Thus, ordered vector spaces 
will always be Archimedean from now on.}

%{\bf Example:} Here's a somewhat pathological case. Let 
%$\V = \R^{2}$, with the cone consisting of the {\em open} 
%positive quadrant (that is, the positive quadrant excluding 
%the positive coordinate axes), plus the origin. Then 
%$$

\tempout{
{\em Remark:} In the complex case, the Polarization Identity 
tells us that a bounded operator $a$ is determined by 
its {\em quadratic form}, $\alpha_{a}(x) = \langle ax, x \rangle$. 
The mapping $\L_{s}(\H) \rightarrow \R^{X(\H)}$ is a linear 
embedding, and we see that positivity of $a$ amounts to 
positivity of $\alpha_{a}$. 
}

{\bf Positive linear maps}  If $\V$ and $\W$ are ordered vector spaces, a linear mapping $T : \V \rightarrow \W$ is {\em positive} 
iff $T(\V_{+}) \subseteq \W_{+}$: that is, 
$\forall x \in \V$, if $0 \leq x$, then $0 \leq T(x)$ in $\W$. 
A special case: a linear functional $f : \V \rightarrow \R$ is 
positive iff $f(x) \geq 0$ for all $x \in \V_{+}$. 

\tempout{{\gray For many ordered spaces, positivity implies continuity 
of a linear functional.  This is true, for instance, if .... }
We write $\V^{\ast}_{+}$ for the set of positive linear functionals 
on $\V$.  It's easy to check that $\V_{+}^{\ast}$ is a cone, and 
that $\V^{\ast}_{+} \cap -\V^{\ast}_{+} = \{0\}$. But for infinite-dimensional spaces, it can happen that $\V^{\ast}_{+}$ does not 
span the space $\V'$ of continuous linear functionals. In such cases, one introduces the {\em order dual}, $\V^{\star} = \V^{\ast}_{+} - \V^{\ast}_{+}$, ordered by $\V^{\star}_{+} := 
\V^{\ast}_{+}$. }
}

{\bf Order unit spaces} 
An {\em order unit} in an ordered vector space $\E$ is an element 
$u \in \E_{+}$ with the property that, for every $x \in \E$, 
$-nu \leq x \leq nu$ for some $n \in \N$. In other words, 
\[\E \ = \ \bigcup n[-u,u].\]
If $\E$ is finite-dimensional,  one can show that this is the case if and only if $u$ belongs to the interior of $\E_{+}$. More generally, this is true for ordered Banach spaces 
with closed cones. 
\footnote{The still more general statement is that $u$ is an 
order unit iff it belongs to the {\em algebraic interior} of the cone. See Wikipedia for more on this.} 

%{\bf Definition:} 

\begin{definition}\label{def: OUS} An {\em order-unit space} (OUS) is a pair $(\E,u)$ where $\E$ is an ordered vector space\footnote{with a closed cone, as per our standing assumptions} and $u$ is an order-unit. 
\end{definition} 

\begin{example} For any set $X$, let ${\mathscr B}(X)$ be the space of bounded functions $f : X \rightarrow \R$, ordered point-wise. The constant function $1$ (or any other positive constant, or, indeed, any positive function bounded away from $0$) is an order unit.\end{example} 

{\gray 
\begin{exercise} Show that if $X$ is infinite, the space $\R^{X}$ of all real-valued functions on $X$, ordered point-wise, 
has no order unit.
\end{exercise} 
}

\begin{example}  Let $\L_{s}(\H)$ denote the ordered real vector space of bounded, self-adjoint operators on a Hilbert space $\H$. Then $\1$ is an order unit for $\L_{s}(\H)$. 
%and the order-unit norm is the operator norm. 
\end{example} 
%A locally convex order unit space [check: implies archimedean?] 

An order unit space $(\E,u)$ carries a natural norm, called the {\em order-unit norm}, obtained 
by treating $[-u,u]$ as the closed unit ball: for any 
$x \in \E$, we set 
\[\|x\|_{u} = \inf\{ t \geq 0 | x \in t[-u,u]\}\footnote{\gray This construction works for any convex set $B$ that is balanced 
($x \in B \Rightarrow -x \in B$) and absorbing ($\bigcup_{t > 0} tB = \E$). The quantity 
$p_{B}(x) = \inf\{ t \geq 0 | x \in tB\}$ is called the {\em Minkowski functional} of $B$. It is generally only a semi-norm. See Appendix B for further details.}\]
In the case of ${\mathscr B}(X)$ (cf. Exercise 2.7), this is the sup norm; in the case 
of $\L_{s}(\H)$, it coincides with the operator norm. 

\begin{theorem}[Kadison, 1951]\label{thm: Kadison}  Any (Archimedean) order-unit space, with its order-unit norm, is norm and order isomorphic to a 
subspace of ${\mathscr C}(X)$ for some compact Hausdorff space $X$.
\end{theorem} 

%As an example, $\B_{\sa}(\H)$ is norma and order isomorphic to a subspace of 
%${\mathscr C}(\Omega)$ where $\Omega$ is the set of density operators on $\H$, taken in the weak ($\ast$?) topology.

{\bf Base normed spaces} 
Let $\Omega$ be a convex set in some real vector space $\W$. As long as $0 \not \in \Omega$, we can define a cone 
\[\V_{+}(\Omega) = \{ t\alpha | \alpha \in K \ \& \ t \geq 0\}.\]

\begin{exercise} Using the convexity of $\Omega$, show that 
$\V_{+}(\Omega)$ is indeed a cone. 
\end{exercise} 

We say that $\Omega$ is a {\em base} for its cone, or a 
{\em cone-base}, iff every non-zero element $\upsilon \in \V_{+}(\Omega)$ has a {\em unique} representation $\upsilon = t\alpha$ with $\alpha \in \Omega$. In other words, $\Omega$ is a base iff, for $\alpha, \beta \in \Omega$ and scalars $t, s > 0$, $t\alpha = s \beta \Rightarrow \alpha = \beta$: 
%\begin{frame}%{Probabilistic Models}
%A rough picture: 
{\small 
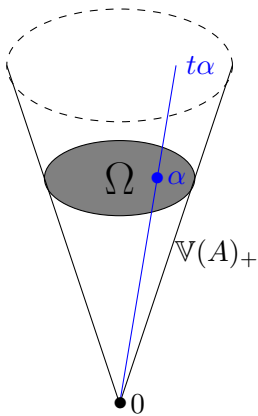
\[\begin{tikzpicture}[scale=.5]
    % Then we draw the rings
%\draw \boundellipse;
\draw[fill=gray] (0,0) ellipse (2 cm and 1cm);
\draw (0,0) node{\Large $\Omega$};
\draw (-3,3) -- (0,-6) -- (3,3);
\draw[dashed] (0,3) ellipse(3 cm and 1.5 cm);
\draw[blue] (0,-6) -- (1,0);
\draw[blue] (1,0) -- (1.5,3);
\draw[blue] (1,0) node{$\bullet$} node[above,right]{$\alpha$};
\draw[blue] (1.5,3) node[above,right]{$t\alpha$};
\draw (1.2,-2) node[right]{$\V(A)_{+}$};
\draw (0.02,-6) node{$\bullet$} node[right]{$0$};
%\draw[->] (1.5,-4) -- (.2,-2);
%\draw (1.5,-4) node[below]{subnormalized states};
\end{tikzpicture}
\]
}
Just as the picture above suggests, this is equivalent to saying that there is a hyperplane 
$H$ with $\Omega = \V(A)_{+} \cap H$ and $0 \not \in H$. 
We can now define an ordered linear space just by considering 
\[\V(\Omega) := \V(\Omega)_{+} - \V(\Omega)_{+}.\]
As constructed, this is a subspace of the ambient space $\W$ 
in which we located $\Omega$, but one can show that $\V(\Omega)$  is essentially independent of this space $\W$, depending only on $\Omega$'s convex structure. 

One can define a semi-norm on $\V(\Omega)$, in something like the way we defined the order-unit norm: declare  $B := \con(\Omega \cup -\Omega)$ (the convex hull of $\Omega$ and $-\Omega$) to be the unit ball. Every $x \in \V$ is a multiple of something in this set (it's "absorbing"), so we can define $\|x\|_{B} = \inf\{t | x \in t B\}$.  

\begin{definition} When $\| ~\cdot~\|_{B}$ is a norm, we call it the {\em base norm} induced by $\Omega$, and refer to the pair $(\V,\Omega)$ as a {\em base-normed space} (BNS). 
\end{definition} 

An important sufficient condition for $\| \cdot \|_{B}$ to be a norm is that 
$\Omega$ be compact in some linear topology, and in this case,  one can show that 
$\V(\Omega)$ is complete in its base-norm.  Importantly, any bounded affine (convex-linear) mapping from $\Omega$ to any normed space $\W$ extends uniquely to a linear (and automatically, positive) mapping $\V \rightarrow \W$: 

\begin{lemma}\label{lem: unique extension lemma} If $(\V,\Omega)$ is a base-normed space, $\W$ any normed space, and 
and $f : \Omega \rightarrow \W$ is a bounded affine mapping, then there exists a unique bounded linear mapping $\hat{f} : \V(\Omega) \rightarrow \W$ with $\hat{f}(\alpha) = f(\alpha)$ for every $\alpha \in \Omega$. 
\end{lemma}

%{\blue \em Proof or Reference!} 

{\em Proof:} The only possible linear extension would be given by 
\[\hat{f}(t\alpha - s\beta) 
= tf(\alpha) - s f(\beta)\] 
for all $\alpha, \beta \in \Omega$ 
and all $s, t \geq 0$.  The trick is to show that this is well-defined. But if $t\alpha - s \beta = t'\alpha' - s\beta'$ we have 
\[\tfrac{t}{t + t'}\alpha + \tfrac{t'}{t + t'}\alpha'
= \tfrac{s}{s + s'}\beta + \tfrac{s'}{s + s'}\beta\]
in $\Omega$. Applying $f$ to both sides and unwinding, we see 
that $tf(\alpha) - sf(\beta) = t'f(\alpha') - s'f(\beta')$, as 
required. It is now routine to check that $\hat{f}$ is bounded 
and linear. $\Box$ 

\begin{exercise} Do so. \end{exercise} 
%\end{document}

%{\gray (This method generally gives a semi-norm, but in cases where $K$ is compact, it's actually a norm, and in that case, $\V(K)$ is actually complete, i.e., a Banach space. }
%[Also, then $V(K)$ Banach?].]

There is a duality between base normed and order-unit spaces. If $(\E,u)$ is an OUS, its (continuous) dual space $\E^{\ast}$, ordered in the usual way, is a {\em complete} BNS with base given by $\Omega = \{ \alpha \in \V^{\ast}_{+} | \alpha(u) = 1\}$. 
Elements of $\Omega$ are called {\em states}, and $\Omega$ is 
$\E$'s {\em state space}.  Conversely, if $(\V,\Omega)$ is a 
base normed space, its dual is a complete OUS, with order unit 
the unique linear functional that is identically $1$ on $\Omega$.  Some further details on these matters are collected in Appendix B; see also \cite{Alfsen, AS}. 

\tempout{
{\bf Lemma:} {\em Any bounded affine mapping $f : K \rightarrow \W$, where $\W$ is a normed linear space, extends uniquely to a positive linear mapping $\overline{f} : \V(K) \rightarrow \W$.} 

{\em Proof:} Let $\Aff(K)$ be the space of bounded affine mappings $f : K \rightarrow \R$. Since $K$ is, {\em ab initio}, a subset of a vector space, say $\mathbb U$, every linear functional on $\U$ defines an affine functional on $K$, so there is no shortage of these. In fact, there are enough to separate points of $K$. Thus, we can embed $K$ in $\Aff(K)^{ast}$ 
(explain). Now $f^{\ast} : \W^{\ast} \rightarrow \Aff(K)$ is 
bounded (say why) and, dualizing again, ....

\begin{example}  For a finite-dimensional quantum system with Hilbert space $\H$, 
$\V = \L_{s}(\H)$, ordered as usual, and  $u = \Tr( \cdot )$. States are density operators. 
Identifying $\V^{\ast}$ with $\V$ (using trace duality), $u$ becomes $\1$ and effects
are positive operators between $\0$ and $\1$.\end{example}

\begin{exercise} Show that if $K$ is a cone-base, 
then (a) $0 \not \in K$, and (b) $\V_{+}(K)$ is pointed. 
\end{exercise} 

We can now define an ordered space $\V(K) = \V_{+}(K) - \V_{+}(K) = \spn(\V_{+}(K)) \leq \W$. 
}

\tempout{
{\bf Lemma:} {\em Any bounded affine mapping $f : K \rightarrow \W$, where $\W$ is a normed linear space, extends uniquely to a positive linear mapping $\overline{f} : \V(K) \rightarrow \W$.} 

{\em Proof:} Let $\Aff(K)$ be the space of bounded affine mappings $f : K \rightarrow \R$. Since $K$ is, {\em ab initio}, a subset of a vector space, say $\mathbb U$, every linear functional on $\U$ defines an affine functional on $K$, so there is no shortage of these. In fact, there are enough to separate points of $K$. Thus, we can embed $K$ in $\Aff(K)^{ast}$ 
(explain). Now $f^{\ast} : \W^{\ast} \rightarrow \Aff(K)$ is 
bounded (say why) and, dualizing again, ....
}

%Results on compactness. 

%Gray results on vector lattices? 

%Order-unit spaces. 

\subsection{Ordered vector spaces and probabilistic models}

As mentioned above, every order unit space $(\E,u)$ gives rise to a probabilistic model in a natural way.  An {\em effect} in $\E$ is any element $a \in \E_{+}$ with $a \leq u$. We write $[0,u]$ for the {\em interval} of all effects. A {\em partition of the unit} in $\E$ is any finite set $E$ of {\em non-zero} effects with $\sum_{a \in E} a = u$.  Let $\D(\E)$ denote the set of all such partitions of unity, and think of this as a test space. Recall that a {\em state} on $\E$ is any positive linear functional $f \in \E'$ with $f(u) = 1$.  Clearly, any state when restricted to $(0,u]$, defines a probability weight on $\D(\E)$. Eliding the distinction between a functional on $\V$ and its restriction to $(0,u]$, we may take $\Omega(\E)$ to be the set of normalized positive functionals, i.e., $\E$'s state-space, and 
this gives us the advertised probabilistic model $(\D(\E), \Omega(\E))$. 

%{\gray {\em Remark:} $\Omega(\E)$ is weak-$\ast$ compact, and 
%$\E$ is isomorphic as an ordered vector space to the 
%space $\Aff_{b}(\Omega(\E))$ of bounded affine functionals 
%on $\Omega(\E)$. }

Conversely, let $A$ be a probabilistic model. If 
$\Omega(A)$ is compact in some linear topology --- in particular, 
if $\M(A)$ is locally finite --- then $\V(A)$ is a complete 
BNS.  But in fact, this is true much more genally. 

\begin{theorem} Let $A$ be a probabilistic model for 
which $\Omega(A)$ is uniformly event-wise closed. Then 
$\V(A)$ is a complete BNS.
\end{theorem} 

It isn't entirely trivial to extract this from the literature, so a proof is given in Appendix C. The case in which 
$A$ is full (i.e,. $\Omega(A)$ consists of all probability 
weights on $\M(A)$) is due to Cook \cite{Cook}. 
%{\green When $A$ is a Borel model, this is essentially 
%the Vitali-Hahn-Saks Theorem [cite] of classical measure theory.} 

As discussed above, $\V(A)^{\ast}$ is an order-unit space. 
%with state-space isomorphic to $\Omega(A)$. 
We have a natural mapping 
\[X(A) \mapsto \V(A)^{\ast}\]
sending $x \in X(A)$ to the evaluation functional $\hat{x} : \alpha \mapsto \alpha(x)$.  This is an {\em effect-valued weight}, in that 
\[\sum_{x \in E} \hat{x} = u\]
for every $E \in \M(A)$. Moreover, if $\Omega(A)$ is large enough 
to separate outcomes in $X(A)$, the mapping $x \mapsto \hat{x}$ is injective. In other words, we have an embedding 
of $A$ into the model $\D(\V^{\ast},\Omega)$. 
%}

%{\blue [Insert dual-pair lingo here. Linearized model is pair $(\E(A),\V(A))$. }

{\bf No-restriction hypotheses} 
The model  $(\D(\V^{\ast}),\Omega)$ has the same states as $A$, but a great many more outcomes and tests. The question arises: are all of these extra outcomes (effects in $[0,u]$ that are not of the form $\hat{a}$ for any event $a \in \Ev(A)$) realizable in practice? 

Certainly some of them are. For example, suppose $E, F$ are two tests in $\M(A)$ of the same size. Let $f : E \rightarrow F$ be a bijection matching up their outcomes in some way. Flip a coin and choose to measure $E$ or $F$, depending on whether you get heads or tails. For any paired outcomes $x$ and $y = f(x)$ in $E$ and $F$, if the system is in state $\alpha$, you'll obtain $x$ with probability $\alpha(x)/2$, and $y$ with probability $\alpha(x)/2$. The effect $\frac{1}{2}(\hat{x} + \hat{f(x)})$ can be realized operationally in this way, and the collection of these, as $x$ ranges over $E$, is a test in 
$\D(\V^{\ast})$. Other convex combinations of outcomes associated with various events of $A$ are achievable in a similar way. 

 In general, however, the convex hull of the event-effects $\hat{a}$, $a \in \Ev(A)$, is not equal to the full effect interval $[0,u]$. The question of which effects in $[0,u]$, and which tests from $\D(\V^{\ast})$, should count as "physical" does not have a very clean answer. It is mathematically convenient to admit them all --- this is called the {\em no restriction} hypothesis --- but as far as I know, no one has ever proposed a good physical or operational justification for doing so.  
 
 {\em Remark:} Even granted that all effects are operationally admissible, it is not obvious that all decompositions of the unit in $\D(\V^{\ast})$ correspond to legitimate experiments.  The assumption  that this is so can be called the {\em no restriction hypothesis for tests}. 

{\bf Sub-normalized States and Processes} A {\em sub-normalized state} of a probabilistic model $A$ is an element of $\V(A)$ 
%function $X(A) \rightarrow [0,1]$
having the form $p \alpha$ where $p \in [0,1]$ and $\alpha \in \Omega(A)$.  Let $\nabla(A)$ denote the set of sub-normalized states of $A$, that is,, 
\[\nabla(A) = \con(\Omega(A) \cup \{0\}).\]
%Evidently, this is the convex hull of $\Omega(A)$ and $0$ in $\V(A)$. %$\R^{X(A)}$.  
We can actually represent $\nabla(A)$ as the space of {\em normalized} states of a model $A_{\ast}$, by adjoining a common "failure outcome" to all the tests in $\M(A)$. 
Formally, let $\ast$ be any symbol not belonging to $X(A)$, and define 
\[\M(A_{\ast}) = \{ E \cup \{\ast\} | E \in \M(A)\}\]
It is easy to see that every probability weight $\beta$  on $\M(A_{\ast})$ has a unique decomposition 
\begin{equation} \beta = p \beta_{1} + (1 - p)\delta_{\ast}\end{equation} 
where $\delta_{\ast}$ is the unique state with $\delta_{\ast}(\ast) = 1$, and $\beta$ is a state with $\beta_1(\ast) = 0$.  Reading $\ast$ as, say, "system failure", the coefficient $p$ in (\theequation) is the probability that the system {\em does not} fail (e.g., is not destroyed). Note here that 
$1-p = \beta(\ast)$, that is, the probability of observing the "failure" outcome.  

{\gray {\em  Remark:} Since the decomposition above is unique, we can consistently interpret this as the probability for the system to be (or to end up) in the failure state, $\delta_{\ast}$. If we adopt this interpretation, then once the system is in state $\delta_{\ast}$, it remains there.  But keep in mind that this is an additional dynamical assumption, not enforced by the formalism.}

To complete the description of $A_{\ast}$, define $\Omega(A_{\ast})$ to be the set of probability weights $\beta$ on $\M(A_{\ast})$ with $\beta_1 \in \Omega(A)$. 
It is clear that $\Omega(A_{\ast})$ is canonically isomorphic to $\nabla(A)$. 

\begin{definition} A {\em process} or {\em channel} from a model $A$ to a model $B$ 
is a positive linear mapping $\Phi : \V(A) \rightarrow \V(B)$ such that 
$u_{B}(\Phi(\alpha)) \leq 1$ for all $\alpha \in \Omega(A)$.  Equivalently, 
the dual mapping $\Phi^{\ast} : \V(B)^{\ast} \rightarrow \V(A)^{\ast}$ takes 
effects to effects. 
\end{definition}

In the case of two quantum models, associated with Hilbert spaces $\H$ and $\K$, a process amounts to a positive, trace-nonincreasing map.  

%Dual process; $\phi^{\ast}(u_B) \leq u_A$. 

The requirement that $u_{B}(\phi(\alpha)) \leq 1$ for all states $\alpha \in \Omega(A)$ 
tells us that $\phi(\Omega(A)) \subseteq \nabla(B) \simeq \Omega(B_{\ast})$. 
 %Note that the unit functional $u_A$ extends uniquely to a functional, which we also denote by $u_A$, on $\Omega(A_{\ast})$, defined by setting $u_{A}(\ast) = 0$. Suppose now that $\phi : \V(A) \rightarrow \V(B)$ is any {\blue process, as defined above [where??]}. Then $\phi$ maps $\Omega(A)$ into $\nabla(B) \simeq \Omega(B_{\ast})$. 
{ For every $\alpha \in \Omega(A)$, 
let  $\beta = \phi(\alpha)$ have the decomposition 
\[\Phi(\alpha) = \beta = p \beta_{1} + (1 - p) \delta_{\ast}\]
as above. Then $u_{B}(\Phi(\alpha)) = p u_{B}(\beta_{1}) = p$.  That is, $u_{B}(\Phi(\alpha))$ is the probability that we will not see the failure outcome 
in the ouput state $\Phi(\alpha)$ --- or, on the interpretation discussed above, 
that the output state {\em is not} the failure state $\delta_{\ast}$. On either 
interpretation, we say that the  
%in the output state $\beta$, 
%the system $B$ does not fail (e.g., is not destroyed). We regard this as the probability that 
 process $\phi$ {\em succeeds}, or {\em occurs}, with probability $p = u_{B}(\Phi(\alpha))$ when the input state is $\alpha$. }
 
 {\gray {\em Remark:} Of course, in quantum theory it is standard to require channels to be 
not merely positive, but completely positive. This concept becomes available once we 
have made the choice of a rule for composing models. We will return to this issue in the
Chapter 4. For the moment, we make the following observation (some assembly required): 

\begin{exercise} Show that if $\Phi : \V(A) \rightarrow \V(B)$ is a channel, then so is 
$\Phi \otimes \1 : \V(A) \maxtensor \V(C) \rightarrow \V(B) \maxtensor \V(C)$ for 
any model $C$. \end{exercise} 
}

\subsection{Effect Algebras} 

The interval $[0,u]$ in an order unit space is the prime example of an structure called an {\em effect algebra} \cite{FB} that significantly generalizes the concept of an orthoalgebra. The axioms are the same, with one exception: 

\begin{definition}\label{def: EA} An {\em effect algebra} is a structure $(L,\perp,\oplus, 0,1)$ where $\perp \subseteq L \times L$ is symmetric binary relation, $\oplus : \perp \rightarrow L$ is a 
partially-defined binary operation on $L$, and $0,1 \in L$, 
such that 
\begin{mlist} 
\item[(i)] $p \perp q \Rightarrow p \oplus q = q \oplus p$; 
\item[(ii)] $p \perp q$ and $(p \oplus q) \perp r$ imply 
$q \perp r$, $p \perp (q \oplus r)$, and 
\[(p \oplus q) \oplus r = p \oplus (q \oplus r)\]
\item[(iii)] For all $p \in L$ $\exists! p' \in L$ with 
$p \oplus p' = 1$
\item[(iv)] $p \perp 1 \Rightarrow p = 0$
\end{mlist} 
Notice that there is no prohibition here against an element 
being self-orthogonal. 
\end{definition} 

In the case of the order interval $[0,u]$ in an OUS, the effect-algebra structure is given buy $a \perp b$ iff $a + b \leq u$, 
in which case $a \oplus b = a + b$. (There are, however, 
effect algebras not of this form.) 

{\gray {\em Remark:}  A particularly important special case is the unit interval $[0,1] \subseteq \R$! Note that if 
$0 \leq t, s$ and $t + s \leq 1$, then for any $a \in [0,u] \subseteq \E$, $ta \perp sa$ and $ta \oplus sa = (t + s)a$. 
Thus, effect algebras of the particular form $[0,u]$ are 
in some sense "modules" over $[0,1]$.}

A {\em state} on an effect algebra is a function 
$\alpha : L \rightarrow \R$ with $\alpha(p) \geq 0$ for 
all $p \in L$, $\alpha(p \oplus q) = \alpha(p) + \alpha(q)$ 
whenever $p \perp q$, and $\alpha(1) = 1$. If 
$L = [0,u_{A}]$ for a probabilistic model $L$, every 
state in $\Omega(A)$ defines a state in this sense, but 
in general, there may be states on $L$ not arising from those in $\Omega(A)$. 

Joint orthogonality is defined exactly as for orthoalgebras, and just as for orthoalgebras, we write $\D(L)$ for the set of 
orthopartitions of the unit in $L$: this is a test space, but 
unless $L$ is an orthoalgebra, it is not algebraic.

\tempout
{\gray 
\subsection{Digression on continuous mixtures}

For many purposes, it's convenient to consider mixtures 
of families of states indexed by some continuous, or more generally, measurable, parameter. Likewise, one often wants to 
consider observables with values in a topological or measureable space. Here, I'll briefly indicate how both can be done. 

\begin{definition}\label{weakly integrable}  Let $\V$ be a vector space with a locally convex Hausdorff topology, and let $(S,\Sigma,\mu)$ be a 
measure space. A function $\phi : S \rightarrow \V$ is 
{\em weakly integrable} with respect to a set of functionals $\F \subseteq \V^{\ast}$ iff for every $f \in \F$, (i) 
$f \circ \phi$ is integrable with respect to $\mu$, and (ii)
there exists a vector 
$v \in \V$ such that 
\[f(v) = \int_{S} f(\phi(s)) d\mu\]
When this is the case, $v$ is unique, and we write 
\[v \ = \ \int_{S} \phi(s) d\mu.\]

The application that most interests us is where 
$\V = \V(A)$ and $\F$ is either the set of functionals $\hat{x}$ associated with outcomes $x \in X(A)$, or 
the set of functionals $\hat{a}$ associated with 
events $a \in \Ev(A)$. If $A$ is locally finite, 
it's enough to assume weak integrability with respect 
to the former. In that case, we have 
\[(\int \phi(s) ds)(x) = \int_{S} \phi(s)(x) d\mu\]
for every $x \in X(A)$.  

If $\mu$ is a probability measure, then we can regard 
$\int \phi(s) d\mu$ as a generalized convex combination 
of the vectors $\phi(s)$. Suppose $\phi(s) \in \Omega(A)$ 
for all $s \in S$. Suppose $\Pr(\M)$ is closed, hence compact, 
in the product topology on $[0,1]^{X}$ (as it is, e.g., when 
$\M$ is locally finite). Then 
if $E \in \M(A)$ we have 
\[\sum_{x \in E} \left ( \int_{S} \phi(s) d\mu \right )(x) 
\ = \ \int_{S} (\sum_{x \in E} \phi(s)(x)) d\mu 
\ = \ \int_{S} 1 d\mu \ = \ 1.\]
It follows that $\int \phi(s) d\mu$ is a probability weight 
on $\M(A)$. If $\Omega(A)$ is closed, which we generally assume, then one can show that $\int_{S} \phi(s) d\mu \in \Omega(A)$. See \cite{Alfsen} for details. 

\tempout{A weakly integrable function $\phi : S \rightarrow \V^{\ast}$ 
is {\em Pettis integrable} iff $\chi_{A} \otimes \phi$ 
is weakly integrable for every $A \in \Sigma$. This means 
that }
}

\subsection{Linearization and Sequential Measurement}

A limitation of the linear framework is that it doesn't play well with sequential tests. As we already know, if $a, b \in \Ev(A)$ with $a \sim b$, $\alpha(a) = \alpha(b)$ 
for all states $\alpha$ of $A$. However, if we afterwards perform 
an experiment on the same or another system $B$, and $c \in \Ev(B)$, then $ac \not \sim bc$, in general. Thus, we can't expect $\V(\for{AB})$ to depend straightforwardly on $\V(A)$ and $\V(B)$. Rather, just as with the logic, the $\V(\for{AB})$ depends 
only on $\V(B)$, but on the detailed test-space structure of 
$A$, that is, on $\M(A)$.  We can linearize $\omega \in \Omega(\for{AB})$ in the second 
argument, since we have $\hat{\omega} : X(A) \rightarrow \V(B)$, but even where 
$\Omega(A)$ separates points of $X(A)$, so that we can effectively replace 
$x \in X(A)$ by $\hat{x}$, the dependence 
of $\hat{\omega}(x) = \alpha(x)\beta_{x}$ on $\hat{x}$ is not linear in general. 
%\[\sum_{x \in E} f \circ \hat{\omega}(x) = (x \otimes f)(\omega)...\]

\tempout{
%{\bf Definition:} 
\begin{definition}\label{def: vector-valued weight} If $\M$ is a test space, a {\em vector-valued weight} on $\M$ with value in a topological vector space $\W$, is a 
function $\omega : X := \bigcup \M \rightarrow \W$ such that
(i) for 
all $E \in \M$, $\sum_{x \in E} \omega(x)$ converges in $\W$, and (ii) for all $E, F \in \M$, 
\[\sum_{x \in E} \omega(x) = \sum_{y \in F} \omega(y).\]
\end{definition} 

\begin{definition} Let $\W$ be a normed linear space. A function $f : X = \bigcup \M \rightarrow \W$ is {\em summable} iff for every test $E \in \M$, $\sum_{x \in E} f(x)$ 
converges in norm to an element of $\W$. 
\end{definition} 

\begin{proposition}\label{prop: weights on TP as vector valeud weights} $\V(\for{AB})_{+}$ is isomorphic to the space of positive summable $\V(B)$-valued functions on $A$, {\blue where $\V(B)$ is understood as a normed space.}
\end{proposition} 

%\begin{proof:} 
%{\em Proof:} 
{\blue For a proof, see Appendix B.} 
%\end{proof} 
}
%Just as with the logic, the $\V$-space of $\for{AB}$ depends only on $\V(B)$, but on the detailed test-space structure of $A$, that is, on $\M(A)$. 

It's natural to wonder %if one should restrict attention to those 
about states on $\for{AB}$ that {\em do} linearize in the first argument. The basic requirement is that such a state should satisfy condition (a) in the following. Recall that 
if $a \in \Ev(A)$, $\hat{a} \in \V(A)^{\ast}$ is defined by $\hat{a}(\alpha) = \alpha(a)$. 

%\newpage
\begin{proposition}\label{prop: linearizing in first argument} Let $\omega$ be a state of $\for{AB}$. 
Consider the following statements:
\begin{mlist} 
\item[(a)] For all $a,b \in \Ev(A)$,  
\[\hat{a} = \hat{b} \ \Rightarrow \ \omega(a,c) = \omega(b,c)\] for all $c \in \Ev(B)$; 
\item[(b)] $\omega(E,y) = \omega(F,y)$ for all $E, F \in \M(A)$ and $y \in X(B)$; 
\item[(c)] $a \sim b$ implies $\omega(a,y) = \omega(b,y)$ for all $a, b \in \Ev(A)$ and $y \in X(B)$.  
\end{mlist} 
Then (a) implies (b) and (c), and the latter two are equivalent. 
\end{proposition} 

{\em Proof:} (a) immediately implies (b), since $\hat{E} = \hat{F}$ for 
any $E, F \in \M(A)$. (b) implies (c), since if 
$a \co c$ and $c \co b$, we can let $E = a \cup c$, $F = c \cup b$, and then 
\[\omega(a,y) + \omega(c,y) 
= \omega(E,y) = \omega(F,y) = \omega(c,y) + \omega(b,y),\]
whence $\omega(a,y) = \omega(b,y)$. 
Reversing this argument shows that (c) implies (b). $\Box$ 

Condition (b) tells us that $\omega$ has a well-defined 
marginal state on $B$, i.e., that the probability of observing 
$y \in X(B)$ is independent of which measurement we make on 
$A$. This is a kind of "no signaling from the past" requirement, 
and is obviously very restrictive. However, when 
$A$ and $B$ are not thought of as causally connected, 
with $A$ "earlier" than $B$, but rather as causally disconnected, 
perhaps spatially widely separated, the idea that there should 
be no signaling between $A$ and $B$ in either direction becomes 
very attractive. We'll explore idea further in the next Section, 
where we'll discuss "non-signaling" composite systems in some detail.

%Define $\back{AB}$ to be $\sigma(\for{BA})$ where $\sigma$ swaps the order of outcomes, i.e., $\sigma(x,y) = (y,x)$. 
%Then the above are all eqiivalent to: 
%$\omega \in \Omega(\for{AB}) \cap \Omega(\back{AB}) = 
%\Omega(\for{AB} \cup \back{AB})$.

%More generally, consider $\V(A^c)$ and $\V(A)$ for $A$ %sequential.... 

\section{Joint Probabilities and Composite Systems} 
\epigraph{\fontsize{8}{9} Our experience with the
everyday world leads us to believe that ... a state of the joint system is just an ordered pair
of states of its parts. .... One of the more shocking discoveries of the twentieth century is that this is {\em \fontsize{11}{12}\selectfont wrong}. }{J. Baez \cite{Baez-QQ}}

By a composite system, I mean a collection of two or more systems, taken together as a single unit.  In quantum theory, composite systems are constructed using the tensor product of Hilbert spaces. In classical probability theory, one uses the tensor product of Boolean algebras --- in the simplest case, this  amounts to the Cartesian product of sets. In this section, we'll discuss how one might model a composite of two arbitrary systems, 
focusing on the idea that states for such a composite model should be, or at any rate should give rise to, joint probabilities for events associated with each of the two systems. 

We will also generally impose a constraint called {\em no-signaling}, which in brief is the principle that the choice of which experiment is performed on one system should have no effect on the {\em probability} of obtaining a given outcome on the other. Both classical and quantum-mechanical composites obey this non-signaling principle. More or less generically, non-signaling composites of non-classical systems turn out to support 
analogues of entangled quantum states, enjoying many of the same properties. It was this observation, more than anything else, that sparked the widespread interest in GPTs, starting with Barrett's paper \cite{Barrett}.\footnote{though it had been pointed out earlier by Kl\"{a}y \cite{Klay} in the early 1980s, prior to the advent of quantum information theory. Needless to say, Kl\"{a}y's paper was largely ignored at the time.}

However, as we'll see, there is generally no one single non-signaling composite of two models. Rather, the choice of such a composite is part of what goes into building a probabilistic theory.

\subsection{Joint probability weights and the no-signaling property}

Suppose $\A$ and $\B$ are test spaces, with outcome-spaces 
$X = \bigcup \A$ and $Y = \bigcup \B$, respectively.  A {\em joint probability weight} on $\A$ and $\B$ is a function 
\[\omega : X \times Y \rightarrow \R\]
such that, for all tests $E \in \A$ and $F \in \B$, 
\[\sum_{(x,y) \in E \times F} \omega(x,y) = 1.\]
In other words, as restricted to $E \times F$, $\omega$ is a joint probability weight in the usual sense.  

We can think of $E \times F$ as the outcome-set for an experiment in which one party, Alice, performs $E$ and another, Bob, performs  $F$, and they later collate their results.  Call this a {\em product experiment} or {\em product test}. The collection of all of these is a test space, denoted (with some abuse of notation) 
$\A \times \B$. In other words, 
\[\A \times \B \ = \ \{ E \times F \ | \ E \in \A, F \in \B\}.\]
{\em Joint probability weights} are simply probability weights 
on $\A \times \B$. 

Given a joint probability weight $\omega$, a test $E \in \A$, and an outcome $y \in Y$, we can define the {\em marginal probability} of $y$ with respect to $\omega$ and $E$ by 
\[\omega_{2|E}(y) = \omega(Ey) = \sum_{x \in E} \omega(x,y).\]
It's easy to see that this must sum to $1$ over every test 
$F \in \B$, so it defines a probability weight $\omega_{2|E}$ on $\B$. Marginals $\omega_{1|F} \in \Pr(\A)$ are defined similarly.
%We can then also define {\em conditional} probability weights 
%\[\omega_{2|E,x}(y) = \frac{\omega(x,y)}{\omega_{1|F}(x)\]

\begin{definition}\label{def: joint state} {\em A joint state on models $A$ and $B$ is a joint state on $\M(A) \times \M(B)$ such that the marginals  $\omega_{1|E}$ and $\omega_{2|F}$ belong to $\Omega(A)$ and $\Omega(B)$, respectively, for all tests $E \in \M(A)$ and $F \in \M(B)$. }
\end{definition} 

In general, the marginals $\omega_{2,E}$ and 
$\omega_{1,F}$ of a joint state will very much depend on the choice of the tests $E \in \A$ and $F \in \B$, so that Alice's {\em choice} of which test to perform will influence the probabilities of Bob's outcomes. In this situation, 
Alice can send (possibly very noisy) {\em signals} to Bob, modulated by her different choices of $E \in \M(A)$.

If Alice and Bob occupy space-like separated locations (that is, if they are constrained to perform their experiments outside of each other's  light cones), then this sort of signaling should not be possible.

\begin{definition}[\cite{FR-TP}]\label{def: no signaling} A probability weight $\omega$ on 
$\A \times \B$ {\em allows no signaling}, or {\em exhibits no influence}, from $\A$ to $\B$ iff $\omega(Ey) = \omega(E'y)$ for all tests $E, E' \in \A$. Similarly, 
$\omega$ exhibits no influence from $B$ to $A$ iff 
$\omega(xF) = \omega(xF')$ for all $F, F' \in \B$. 
If $\omega$ exhibits no influence in either direction, we say that it is {\em influence-free} or 
{\em non-signaling} (NS). 
\end{definition} 

We've seen this behavior before. Recall from 
Section 2 that 
\[\for{\A\B} = \left \{ \bigcup_{x \in E} \{x\} \times F_{x} | E \in \A, F \in \B^{E}\right \}.\]
Define 
\[\back{\A\B} = \sigma(\for{\A,\B})\]
where $\sigma : X \times Y \rightarrow Y \times X$ is the mapping 
$\sigma(x,y) = (y,x)$. 
Finally, let $\bilat{\A\B} = \for{\A\B} \cup \back{\A\B}$. Note that all three of these 
test spaces contain $\A \times \B$, and all three 
have total outcome-set $X \times Y$.  It follows 
that 
\[\Pr(\bilat{\A\B}) = \Pr(\for{\A\B} \cup \back{\A\B}) = \Pr(\for{\A\B}) \cap \Pr(\back{\A\B}).\]
Applying what we learned earlier about probability 
weights on $\for{\A\B}$, we have

%{\bf Lemma 1:} {\em 
\begin{lemma}\label{influence and classical communication}
A probability weight $\omega$ on $\A \times \B$ exhibits no influence from $B$ to $A$ iff $\omega \in \Pr(\for{\A\B})$, no influence from $\A$ to $\B$ iff 
$\omega \in \Pr(\back{\A\B})$, and is no-signaling iff $\omega \in \Pr(\for{\A\B} \cup \back{\A\B}) = \Pr(\bilat{\A\B})$.
\end{lemma} 

{\gray {\em Remark:} The term "signaling" here is potentially confusing. Where we associate $\A$ and $\B$ with two 
parties (Alice and Bob) at remote locations, the performance 
of a two-stage test in $\for{AB}$ generally {\em requires} some 
form of "classical" signaling from Alice to Bob, in order for 
her to communicate her measurement outcome, on the basis of 
which Bob's measurement is to be selected. The possibility 
of this classical communication in the Alice-to-Bob direction 
rules out the possibility of Bob's signaling Alice by 
means of his measurement choices alone. We should probably use a term like "measurement-signaling" to refer to this, but "signaling", without adjectives, is the accepted terminology. %nomenclature
}
%But "signalling" is well established

The non-signaling states --- those exhibiting no influence in either direction --- form a convex subset of $\Pr(\for{\A\B})$. Any non-signaling state $\omega$ has well-defined marginal probability weights, given by 
\[\omega_{1}(x) = \omega(xF) \ \ \mbox{and} \ \ \omega_{2}(x) = \omega(Ey)\]
where $E \in \A$ and $F \in \B$ can be chosen arbitrarily. Using these, we can also define {\em bipartite conditional probability weights} 
\[\omega_{2|x}(y) = \frac{\omega(x,y)}{\omega_{1}(x)} \ \  \mbox{and} \ \ 
\omega_{1|y}(x) = \frac{\omega(x,y)}{\omega_{2}(y)}.\]
It's easy to see that marginal and conditional probability weights are related by 
\begin{equation} 
\omega_{2}(y)  = \sum_{x \in E} \omega_{1}(x) \omega_{2|x}(y) 
\ \ \mbox{and} \ \ \omega_{1}(x) = \sum_{y \in F} \alpha_{2}(y) \omega_{1|y}(x),
\end{equation} 
which are bipartite versions of the {\em law of total probability}. 

\begin{definition}\label{def: joint states}  
A non-signaling joint {\em state } for models $A$ and $B$ is a non-signaling joint probability weight 
$\omega$ with conditional states $\omega_{1|y}$ 
and $\omega_{2|x}$ belonging to $\Omega(B)$ and $\Omega(A)$, respectively, for all $y \in X(B)$ and 
$x \in X(A)$. This implies the marginals 
also live in the correct state spaces, by the bipartite Law of 
Total Probability (equation (\theequation)).   
Write $\Omega_{NS}(A,B)$ for 
the set of all such non-signaling states, and 
let $A \times_{NS} B$  be the model with 
\[\M(A \times_{NS} B) = \M(A) \times \M(B) 
\ \ \mbox{and} \ \ \Omega(A \times_{NS} B) = \Omega_{NS}(A,B).\]
\end{definition} 

The model $A \times_{NS} B$ is the simplest "non-signaling composite" of the models $A$ and $B$, a term I'll define formally in Section 3. First, however, I want to explore some consequences 
of the no-signaling restriction. 

\tempout{

{\bf Exercise:} Suppose $(\alpha;\beta)$ is a probability 
weight on $\for{\A\B}$.  (a) Show that this exhibits no influence from $\B$ to $\A$.  (b) What additional conditions 
are required in order that it also exhibit no influence 
from $\A$ to $\B$? 

{\bf Exercise:} It's briefly tempting to try to 
define bipartite conditional states 
\[\omega_{1|x,F}(y) = \omega(x,y)/\omega_{1|F}(x), 
\ \omega_{2|E,y}(x) = \omega(x,y)/\omega_{2|E}(y)\] 
for a general (not necessarily non-signaling) state on 
$\A \times \B$. What goes wrong? 
}

\subsection{Entanglement} 

A {\em product} state on models $A$ and $B$ is 
one of the form 
\[(\alpha \otimes \beta)(x,y) := \alpha(x)\beta(y)\]
where $\alpha \in \Omega(A)$ and $\beta \in \Omega(B)$. Such a state is always non-signaling, as is any limit of convex combinations (mixtures) of 
such states.  Borrowing lingo from quantum theory:

\begin{definition} A joint state on models $A$ and $B$  is {\em separable} iff it belongs to the closed convex hull of the set of product states. A non-signaling joint state  that is not separable is {\em entangled}. 
\end{definition} 

As we'll see, entangled states exist abundantly in virtually any composite of non-classical models. The basic properties of entangled states in quantum mechanics are actually rather generic features of entangled joint states of non-classical probabilistic models. In particular, we have the following 

%{\bf Lemma 2:} {\em
\begin{lemma}\label{lemma: pure marginals}  Let $\omega$ be any non-signaling joint state on $A \times B$, and let $\alpha \in \Omega(A)$ and $\beta \in \Omega(B)$.  Then 
\begin{mlist} 
\item[(a)]  If $\alpha \otimes \beta$ is pure, then so are $\alpha$ and $\beta$; 
\item[(b)] If either $\omega_1 \in \Omega(A)$ or $\omega_2 \in \Omega(B)$ is pure,  then $\omega = \omega_1 \otimes \omega_2$;
\item[(c)] Hence, if $\omega$ is entangled, the marginals $\omega_1$ and $\omega_2$ are mixed.  
\end{mlist} 
\end{lemma} 

{\em Proof:} (a) {\bf Exercise}, but with these  
hints: (i) for any fixed $y$ with $\beta(y) > 0$, 
$\alpha = (\alpha \otimes \beta)_{1|y}$; 
(ii) for any fixed $y$ with $\omega_{2}(y) > 0$, 
$\omega \mapsto \omega_{1|y}$ is linear. 

(b) Suppose $\omega_2$ is pure. By the bipartite Law of Total Probability (1), we have, for any test $E \in \M(A)$,
\[\omega_{2} = \sum_{x \in E} \omega_{1}(x) \omega_{2|x}. \]
Note that the sum on the right is a convex sum. 
Since $\omega_2$ is pure, for every $x$ we have 
either $\omega_{1}(x) = 0$ or $\omega_{2|x} = \omega_{2}$. In either case, we have 
\[\omega(x,y) = \omega_1(x) \omega_{2}(y)\]
for every $y \in X(B)$. Since this holds 
for every $x \in E$, and $E$ is arbitrary, 
it hold for every $x \in X(A)$.  This 
proves (b), and (c) is an immediate consequence. 
$\Box$

{\gray {\em Historical note:} These points were first noted in this generality, but without any reference to entanglement, in a pioneering paper of Namioka and Phelps \cite{NP} on tensor products of compact convex sets. They were rediscovered, and connected with entanglement, by Kl\"{a}y \cite{Klay}.
} 

Let's agree that a model $A$ is {\em semi-classical} iff 
its test space $\M(A)$ is semi-classical, which, recall, just 
means that distinct tests never overlap. When $\A$ and $\B$ are semi-classical, the test space $\A \times \B$ defined above is again semi-classical. In this situation, we have lots of dispersion-free joint states. However:

\begin{lemma}\label{lemma: semiclassical composites} Let $\A$ and $\B$ be semi-classical test spaces, and let $\omega \in \Pr(\A \times \B)$. If $\omega$ is both non-signaling and dispersion-free, then $\omega = \alpha \otimes \beta$ where $\beta$ and $\\beta$ are dispersion-free.\
\end{lemma} 

{\em Proof:} Suppose $\omega$ is dispersion-free. Since $\omega$ is also non-signaling, it has well-defined marginal states, which must obviously also be dispersion-free, hence, pure. 
But by Lemma \ref{lemma: pure marginals} (b), a non-signaling state with pure marginals is the product of these marginals. $\Box$ 

It follows that any average of non-signaling, dispersion-free states on semi-classical test spaces is separable.

{\em Remarks:}

(1) There exist non-signaling joint 
states on pairs of quantum systems that do not 
correspond to density operators on the composite 
quantum system. A simple example: 
if $\H$ is any complex Hilbert space, 
let $S : \H \otimes \H \rightarrow \H \otimes \H$ 
be given by 
$S(x \otimes y) = y \otimes x$. Then for 
any unit vectors $x$ and $y$ 
$\langle S x \otimes y, x \otimes y \rangle 
= |\langle x, y \rangle|^2$, which clearly 
sums to $1$ over any product basis. The marginals, taken over any orthonormal basis, 
are the maximally mixed state, so this is a 
non-signaling state. It is not, however, a 
quantum state, because $S$ is not even a positive 
operator, much less a density operator. For 
more on this, see \cite{BFRW}. 

{\gray (2) There exist entangled quantum states (certain Werner states \ref{Werner}) that are entangled but nevertheless support a local hidden-variables model. This might seem to be in tension with Lemma \ref{lemma: semiclassical composites}. But note 
that the models in the Lemma are semi-classical, which quantum models are not. }

\subsection{Composites} 

We now try to define a reasonably general notion of 
a composite of two models $A$ and $B$. 
One important consideration is that 
such a composite will generally 
admit outcomes that are {\em not} simply 
ordered pairs $(x,y)$ of outcomes belonging to the two 
models: the outcome-space will need to be bigger than $X(A) \times X(B)$. 
This is clear when we consider the natural 
composite models in classical and quantum probability theory. 

\begin{example} In the case of two (full) Kolmogorovian classical models, $\M(A) = \M(S_A,\Sigma_A)$ and $\M(B) = \M(S_B,\Sigma_B)$, we have a composite given by 
\[\M(AB) = \M(S_A \times \S_B, \Sigma_A \otimes \Sigma_B).\] 
Here 
$\Sigma_A \otimes \Sigma_B$ is the smallest $\sigma$-algebra on $\S_A \times S_B$ containing all 
product sets $a \times b$ with $a \in \Sigma_A, b \in \Sigma_B$. This is vastly larger than 
$\Sigma_A \times \Sigma_B$, 
%which is not itself a boolean algebra in any natural way.  
It is nevertheless true 
that every probability measure $\mu$ on $\Sigma_A \otimes \Sigma_B$ is uniquely determined by its restriction to product sets $a \times b$, and, thus restricted, gives us a joint probability weight on $\M(A) \times \M(B)$ via 
the recipe $a,b \mapsto \mu(a \times b)$. Note that 
this weight is non-signaling: if 
$E \in \M(S_A)$ we have $\sum_{a \in E} \mu(a \times b) 
= \mu((\bigcup_{a \in E} a) \times b) = \mu(S_{A} \times b)$, 
which is independent of $E$, and similarly for $F \in \M(S_B)$. 
\end{example}

\begin{example} If $A$ and $B$ are quantum models, with Hilbert spaces $\H_A$ and $\H_B$, we have a natural composite model 
with Hilbert space $\H_{A} \otimes \H_{B}$. The relevant test space is $\M(AB) = \F(\H_A \otimes \H_B)$, the set of 
all orthonormal bases of $\H_{A} \otimes \H_{B}$. The outcome-space $X(AB)$ is the unit sphere of $\H_A \otimes \H_B$ that is, the set of all unit vectors in $\H_{A} \otimes \H_{B}$, which is vastly larger than the set of vectors of the form $x \otimes y$ 
where $x \in X(A)$ and $y \in X(B)$. Nevertheless, we have a natural mapping $X(A) \times X(B) \rightarrow X(AB)$ sending $(x,y)$ to $x \otimes y$, and this is a test-preserving morphism 
of test spaces. Accordingly, any state in $\Omega(AB) = \Omega(\H_A \otimes \H_B)$ --- that 
is, any state associated with a density operator 
$W$ on $\H_A \otimes \H_B$ --- defines a joint state on $\M(A) \times \M(B)$ by $x,y \mapsto \langle W (x \otimes y), x \otimes y\rangle$. 

\begin{exercise} Check that this state is non-signaling.\end{exercise} 

In the complex case, the joint state above {\em determines} $\rho$, thanks to the polarization identity; but in the real case, it does not.  We'll come back to this point below.)
\end{example} 

{\bf Remark:} Notice that $\dim(\H_{A} \otimes \H_{B}) \geq 4$, so Gleason's Theorem applies: this model is full, even if $\H_A$ and $\H_B$ are qubits.

These examples suggest the following definition:

\begin{definition}\label{def: composites} A {\em non-signaling composite} of models $A$ and $B$ is a model $AB$, together with a test-preserving morphism 
\[\pi : A \times_{NS} B \rightarrow AB\]
such that $\pi^{\ast}(\Omega(AB))$ contains all product 
states $\alpha \otimes \beta$ for $\alpha \in \Omega(A)$ and 
$\beta \in \Omega(B)$.
\end{definition} 

We'll consider some examples below. First, 
let's unpack the definition a bit. The condition that $\pi$ is a test-preserving morphism means, in the first place, that 
$\pi$ is a mapping $X(A) \times X(B) \rightarrow X(AB)$. 
Let's agree to write $\pi(x,y) = xy$, and to call this a {\em product outcome } of $AB$. Then we have 
\[\pi(E \times F) = \{ xy \ | \ x \in E, y \in F\} \in \M(AB)\] 
for all $E \in \M(A)$ and $F \in \M(B)$. Note here that 
$(x,y) \mapsto xy$ will be injective on $E \times F$, 
since morphisms are locally injective. Thus, $\M(AB)$ contains representations of all possible product tests.
%\footnote{
%I haven't asked that $\pi$ be globally injective, because 
%in general, it won't be: see Example 1!}  
The second condition in the definition of a 
composite requires that for any pair of states $\alpha \in \Omega(A)$ and  $\beta \in \Omega(B)$, there must exist {\em at least one} state $\omega \in \Omega(AB)$ with $\omega(xy) = \alpha(x) \beta(y)$. I will say 
a bit more about this below. 

It's time for the promised examples. As mentioned earlier, 
$A \times_{NS} B$ is the simplest non-signaling composite of $A$ and $B$. 

\begin{example} We can extend the definition 
of the bilateral product from test spaces to 
models. Simply define 
\[\M(\bilat{AB}) = \bilat{\M(A)\M(B)}\]
and take $\Omega(\bilat{AB})$ to be the set of 
all probability weights on this that 
are joint states for $A$ and $B$, i.e., 
satisfy $\omega_{2|x} \in \Omega(A)$ 
and $\omega_{1|y} \in \Omega(B)$ for all 
$x \in X(A), y \in X(B)$. Note that if $A$ and $B$ are full, the restrictions on the conditional states are automatic, and $\bilat{AB}$ is thus also full. 
\end{example} 
\vspace{-.1in} 

Like $A \times_{NS} B$, the bilateral product 
is too simple to be of much interest in its own right. But it does have its uses: 

\begin{lemma} If  $\pi : \bilat{AB} \rightarrow C$ is a test-preserving morphism, then the product 
$(C, \pi)$ is a non-signaling composite of $A$ and $B$
\footnote{Or, more exactly, $(C, \pi_{A \times B})$ is a non-signaling composite of $A$ and $B$, where $\pi_{A \times B}$ is $\pi$ understood as a morphism $A \times B \rightarrow AB$}
\end{lemma} 

\begin{exercise} Prove this. \end{exercise} 

The usual composite models from classical probability theory and non-relativistic QM are also examples:

\begin{example}[Example 3.8 revisited]  In the case of two (full) Kolmogorovian classical models, $\M(A) = \M(S_A,\Sigma_A)$ and $\M(B) = \M(S_B,\Sigma_B)$ 
(with their full sets of probability weights as state spaces), It is not hard to see that if $E$ is 
a partition of $S_{A}$ by sets in $\Sigma_A$ and, 
for every $a \in E$, $F_{a}$ is a partition of $S_B$ by sets from $\Sigma_B$, then 
$\bigcup_{a \in E, b \in F_{a}} a \times b$ is a partition 
of $S_{A} \times S_{B}$. Thus, $\pi : a,b \mapsto a \times b$ maps tests in $\for{\M(A)\M(B)}$ to tests in $\M(AB)$, and 
similarly for tests in $\back{\M(A)\M(B)}$. Thus, 
composite classical models are non-signaling. 
\end{example} 

\begin{example}[Example 3.9 revisited] If $A$ and $B$ are quantum models, with Hilbert spaces $\H_A$ and $\H_B$, then its straightforward 
if $E$ is an orthonormal basis for $\H_A$ and, for every 
$x \in E$, $F_{x}$ is an ONB for $\H_B$, then 
$\{ x \otimes y | x \in E, y \in F_{x}\}$ is an 
ONB for $\H_{A} \otimes \H_{B}$, so $\pi : x, y \mapsto x \otimes y$ takes tests in $\for{\M(A)\M(B)}$ to tests in $\M(AB)$. Similarly for tests in $\back{\M(A)\M(B)}$, so composite quantum models are non-signaling. 
\end{example}

{\bf Product States and Strong Composites} 
The condition that $\pi^{\ast}(\Omega(AB))$ contain 
all product states $\alpha \otimes \beta$ means that we can prepare states of $AB$ that look like product 
states, but possibly not in any unique, or even canonical, way. A more restrictive view of a composite system might add the requirement that there be such a canonical choice of product states, and indeed, this is pretty often done in the literature. In order to keep track of these two notions, we'll add an adjective:

\begin{definition}\label{strong composite} A {\em strong} non-signaling composite 
is one equipped with an additional bi-affine mapping 
\[m : \Omega(A) \times \Omega(B) \rightarrow \Omega(AB)\]
such that, for all $\alpha \in \Omega(A), \beta \in \Omega(B)$, 
$x \in X(A)$ and $y \in X(B)$, we have 
\[m(\alpha, \beta) (\pi(x,y)) = \alpha(x)\beta(y).\]
\end{definition} 
\vspace{-.1in}
In this case, I'll generally write $m(\alpha,\beta)$ as $\alpha \otimes \beta$, leaving it to context whether this is to be 
thought of as belonging to $\Omega(A \times_{NS} B)$ or 
$\Omega(AB)$.

\tempout{
{\bf Example 3:} The {\em projective} model associated with a Hilbert space $\H$ has, as outcomes, non-empty closed subspaces of $\H$, or, what comes to the same thing, non-zero (bounded) projection operators in 
$\L(\H)$. The test space, $\M_{\P}(\H)$, consists of all maximal pairwise-orthogonal 
finite collections $\{P_1,...,P_n\}$ of non-zero projections, or, 
equivalently, all finite collections $\{P_i\}$ with $\sum_i P_i = \1$. 
 The state space, $\Omega_{\P}(\H)$, is the set of weights $\alpha_{W} : P \mapsto \Tr(WP)$ for any projection $P$. 
 
Now  let $\H$ and $\K$ be finite-dimensional complex Hilbert spaces 
of dimensions $n$ and $k$, respectively. These can be turned into real Hilbert spaces $\H_{\R}$ and $\K_{\R}$ 
just by restricting scalars to $\R$. These, in turn, can be combined 
using the real tensor product to form $\H_{\R} \otimes \K_{\R}$. 
Note this has dimension $4nk$. This, in turn, can be "complexified": we 
can define a complex scalar multiplication. This gives us 
\[\H \odot \K := (\H_{\R} \otimes \K_{\R})^{\C},\]
a complex Hilbert space of dimension $2nk$.  Set

Recall now that the (basic) quantum model associated with a Hilbert space $\H$ is $(\F(\H),\Omega(\H))$ where $\F(\H)$ is the set of frames (unordered ONBs) of $\H$, and $\Omega(\H)$ is the set of probability weights $\alpha_{W} : x \mapsto \langle Wx,x \rangle$ where $W$ is a density operator on $\H$. This applies equally to real and complex Hilbert spaces. 
}

{\bf Local Tomography} 
There is an important special case in which the distinction between 
ordinary and strong composites vanishes.

\begin{definition}\label{def: LT commposite} A composite $AB$ is {\em locally tomographic} iff 
$\pi^{\ast}$ is injective, i.e., every state 
on $AB$ is determined by the corresponding joint probability 
weight on $A \times B$. 
\end{definition} 

\vspace{-.1in}
If $(AB, \pi)$ is locally tomographic, the mapping $\pi^{\ast} : \Omega(AB) \rightarrow \Omega(A \times_{N} B)$ is injective, and the uniqueness of product states is no longer an issue. In other words, locally tomographic composites are automatically strong. 

\tempout{
Here are two general recipes for constructing locally tomographic non-signaling composites of arbitrary models:
}

\begin{exercise} Let $\H$ and $\K$ be Hilbert spaces, either both 
real or both complex. 

(a) Show that the mapping $\pi : X(\H) \times X(\K) \rightarrow X(\H \otimes \K)$ makes the model associated 
with $\H \otimes \K$ into a non-signaling composite of 
the models associated with $\H$ and $\K$. 

(b) Show that if $\H$ and $\K$ are complex, this composite 
model is locally tomographic. (Hint: use the 
polarization identity twice.)

(c)  Show that if $\H$ and $\K$ are real, the composite 
quantum model is generally {\em not} locally tomographic. 
(Hint: To make life simple, assume $\H = \K$ is  
finite dimensional, and use the decomposition $\L(\H) = \L_{s}(\H) \oplus 
\L_{a}(\H)$ where $L_s(\H)$ and $L_{a}(\H)$ are the 
(real) spaces of symmetric and antisymmetric operators 
on $\H$, and $\L(\H)$ is the space of all operators on $\H$.) 
\tempout{$\H$ and $\K$ be real Hilbert spaces. 
Recall that $\L(\H) = \L_{s}(\H) \oplus \L_{a}(\H)$ where 
$\L_{s}(\H)$ is the space of self-adjoint (symmetric) and 
$\L_{a}(\H)$ is the space of anti-self adjoint (anti-symmetric) operators on $\H$. (a) Using this, show that 
\[\L_{s}(\H \otimes \K) = (\L_{s}(\H) \otimes \L_{s}(\K)) 
\oplus (\L_{a}(\H) \otimes L_{a}(\K)).\]

%(b) Now show that the the composite model based 
%on $\H \otimes \K$ is a {\em non}-locally tomographic 
%composite of the models associated with $\H$ and $\K$. 
}

%(d) Explain why this composite is best thought of as strong. [Which composite? Omit til I remember what I was asking!
\end{exercise} 

% Let $\H$ and $\K$ be real Hilbert spaces. 

{\gray {\bf Digression: Tensor products of quantum logics}  
In \cite{RF-NoGo}, Randall and Foulis showed that 
there is no non-signaling, locally-tomographic 
composite of a certain simple (finite) orthomodular 
lattice $L_5$ with itself that again yields an 
orthomodular lattice (see also \cite{FR-TP, AW-handbook}). Such a tensor product {\em can} be found (as they also showed), but it is a non-orthocoherent  orthoalgebra, hence, 
not even an orthomodular poset. 

There are several possible responses to this. One (see, e.g. \cite{Coecke-LQMII}) is that quantum logic "failed" because it could not accommodate composite systems. But this is far too hasty. A more reasonable take is that the traditional classes of models of quantum logics (orthomodular lattices and posets) are too small, and that the meaning of "quantum logic" should be enlarged to include (at least) orthoalgebras. Another is that the category of orthomodular lattices is {\em too big}. We know there are sub-categories of OMLs that {\em do} have perfectly serviceable composites --- the category of projection lattices, for one!.  This points towards what I think is the correct response (more obvious in hindsight than at the time): that composites of probabilistic models --- including quantum logics as a special case --- are not canonical, and need to be constructed to fit the specific physical theory at hand. }

\subsection{Linearized Composites} 

Naturally, we'd like to know how composites comport with the linearized picture of probabilistic models. As we've seen, the latter can be represented in 
terms of a pair $(\E,\V)$ consisting of an order-unit space $\E$ and a 
base-normed space $\V$, with $\E \leq \V^{\ast}$, and  
%{\blue [Need to say more about 
%this duality. Best done above...]}  
%{\blue dual pair of ordered linear spaces [Never defined!!]} --- specifically, of a base norm and order-unit space. 
conversely, %{\blue 
any such %dual base-norm/order-unit space 
pair gives us a probabilistic model. %c{\blue [Have we actually shown this?] } 

Let's start by considering tensor products of ordered 
vector spaces. For the balance of this section, except for a 
few remarks, I'm going to assume all spaces and models are finite-dimensional. 

If $\V$ and $\W$ are ordered vector spaces, 
a bilinear form $f : \V \times \W \rightarrow \R$ is {\em positive} iff $f(a,b) \geq 0$ for all $(a,b) \in \V_+ \times \W_+$. Write $\B_{+}(\V,\W)$ for the set of positive bilinear forms. 
\tempout{One can show [CHECK!!!] that this {\em is} a 
cone in $\B(\V,\W) 
= \B(\V^{\ast \ast}, \W^{\ast \ast}) = 
\V^{\ast} \otimes \W^{\ast}.$  Also note that there's 
a bilinear form $\V^{\ast} \otimes \W^{\ast} \rightarrow 
\B(\V,\W) = (\V \otimes \W)^{\ast}$ taking $(a,b)$ to the functional 
\[(a \otimes b)(\alpha,\beta) = a(\alpha)b(\beta).\]
}
We can view the tensor product of two finite-dimensional 
vector spaces either as a space dual to the space 
$\B(\V, \W)$, i.e., $\V \otimes \W = \B(\V,\W)^{\ast}$, 
or as a space {\em of} bilinear forms, i.e, 
$\V \otimes \W = \B(\V^{\ast}, \W^{\ast})$. This 
largely boils down to the decision whether to write 
\[(\alpha \otimes \beta)(a \otimes b) \ \ \mbox{vs} \ \  (a \otimes b)(\alpha \otimes \beta)\]
for $a(\alpha) b(\beta)$, where $\alpha \in \V$, $\beta \in \W$, $a \in \V^{\ast}$ and $b \in \W^{\ast}$.  
We end up with canonical isomorphisms 
\begin{equation} \B(\V^{\ast},\W^{\ast}) \simeq \V \otimes \W \simeq  \B(\V,\W)^{\ast} \end{equation}
\vspace{.1in}
so we can use either of the spaces $\B(\V^{\ast},\W^{\ast})$ 
or $\B(\V,\W)^{\ast}$ to represent (or even define) $\V \otimes \W$. 

\begin{exercise} Establish (\theequation) for finite-dimensional $\V$ and $\W$. (Ideally, try to do this by writing down the canonical isomorphisms in question explicitly (and without choosing bases!) rather than just counting dimensions.)
\end{exercise} 

If we are only interested in the linear structure of $\V \otimes \W$, it doesn't matter which of these representations we use. But if we are interested in {\em ordered} vector spaces, they come with different natural cones:

\begin{definition} Let $\V$ and $\W$ be finite-dimensional ordered 
vector spaces.  Their {\em minimal tensor product}, 
$(\V \mintensor \W)$  is $\V \otimes \W$ ordered by 
the dual cone in $\B(\V,\W)^{\ast}$. 
Their {\em maximal tensor product}, $\V \maxtensor \W$, 
is $\V \otimes \W$ ordered by the cone  $\B_{+}(\V^{\ast},\W^{\ast})$.  
\end{definition} 

In more detail:  a tensor $\omega$ belongs to $(\V \mintensor \W)_{+}$ iff $\omega(F) \geq 0$ for all positive bilinear forms $F \in \B(\V,\W)$, and $\omega$ belongs to $(\V \maxtensor \W)_+$ iff 
$\omega(a \otimes b) \geq 0$ for all positive functionals 
$a \in \V^{\ast}$ and $b \in \W^{\ast}$. 
%\end{document} 

%\end{document} 

\begin{lemma}\label{lem: minimal tensor cone} The minimal tensor cone is the cone spanned 
by positive pure tensors. That is, $\omega \in (\V \mintensor \W)_{+}$ iff 
\[\omega = \sum_{i} t_i \alpha_i \otimes \beta_i\]
where $\alpha_i \in \V_{+}$, $\beta_{i} \in \W_{+}$, 
and the coefficients $t_i$ are all non-negative.\footnote{Since $\otimes$ is bilinear, we don't really need the coefficients here: $t(\alpha \otimes \beta) 
= (t\alpha) \otimes \beta$, and if $t$ and $\alpha$ are positive, so is $t\alpha$. }
\end{lemma} 

\begin{exercise} Prove this.  \vspace{-.1in} \end{exercise} 

Thus, in the minimal tensor cone, all normalized states are separable.

\tempout{{\em Proof:} Let $C$ denote the {\blue [closed?]} cone spanned by 
pure tensors $\omega = \alpha \otimes \beta$ with 
$\alpha, \beta$ positive. For any such pair $\alpha, \beta$ 
and any positive bilinear 
form $F$ we have $\omega(F) = F(\alpha,\beta) \geq 0$. It that $C$ is contained 
in $(\V \mintensor \W)_{+}$. Conversely, suppose 
$\tau$ is not in the cone $C$. Find a separating positive 
functional: this amounts to a bilinear form, so not 
in min. cone. $\Box$}

%{\bf Corollary:} {\em 

\begin{corollary}\label{cor: duality for min and max TP} Let $\V$ and $\W$ be finite-dimensional 
ordered vector spaces. Then (up to canonical order-isomorphisms),
\begin{mlist} 
\item[(a)] $(\V \maxtensor \W)^{\ast} = \V^{\ast} \mintensor \W^{\ast}$, 
\item[(b)] $(\V \mintensor \W)^{\ast} = \V^{\ast} \maxtensor \W^{\ast}$, 
and hence 
\item[(c)] $(\V \maxtensor \W) = (\V^{\ast} \mintensor \W^{\ast})^{\ast}$ 
\end{mlist} 
\end{corollary} 
%\end{document} 

\begin{exercise} Prove this one, too. \vspace{-.1in} \end{exercise} 

Returning to probabilistic models, suppose 
$\omega$ is a joint state on $\A \times \B$. 
As we will now see, this is non-signaling 
if and only if it extends to a bilinear form 
on $\V(A) \times \V(B)$. More exactly, 

\begin{theorem}\label{thm: no signaling and bilinearity} A joint probability 
weight $\omega$ on $\M(A) \times \M(B)$ is non-signaling 
iff there exists a bilinear form 
$F_{\omega} : \V(A) \times \V(B) \rightarrow \R$ such that 
\[F_{\omega}(\hat{x},\hat{y}) = \omega(x,y)\]
for all outcomes $x \in X(A)$, $y \in X(B)$.
\end{theorem} 

{\em Proof:} If $\omega$ corresponds to a bilinear form 
$F$ in the indicated way, it's straightforward that $\omega$ is non-signaling; see Exercise 
\ref{exc: no signaling from bilinearity} below. For the converse, 
suppose that $\omega$ is 
non-signaling. Define a mapping 
\[\hat{\omega} : X(A) \rightarrow \V_{+}(B)\]
by setting 
\[\hat{\omega}(x)(y) = \omega(x,y)\]
for every $x \in X(A)$ and all $y \in X(B)$. 
Then $\sum_{x \in E} \hat{\omega}(x) = \omega_{2} \in \Omega(B)$ for every $E \in \M(A)$.  Thus, if 
$b \in \V^{\ast}_{+}$, we have 
\[\sum_{x \in E} b(\hat{\omega}(x)) = b(\omega_{2}) \geq 0,\]
a constant. Thus, $b \circ \hat{\omega} \in \V_{+}(A)$, 
and we have a positive linear mapping $\V^{\ast}(B) \rightarrow \V(A)$ given by $\hat{\omega}^{\ast}(b) = b \circ \hat{\omega}$ for all $b \in \V^{\ast}$. This, in turn, 
defines a bilinear form on $\V^{\ast}(A) \times \V^{\ast}(B)$, 
given by $F_{\omega}(a,b) = a(\hat{\omega}^{\ast}(b))$. $\Box$

\begin{exercise}\label{exc: no signaling from bilinearity} Suppose $F : \V(A)^{\ast} \times \V(B) \rightarrow \R$ 
is a positive bilinear form with $F(u_A, u_B) = 1$, 
and define $\omega(x,y) = F(\hat{x},\hat{y})$ for all 
$x \in X(A), y \in X(B)$. Show that $\omega$ is a non-signaling state on $A$ and $B$. 
\end{exercise} 

\begin{corollary} If $AB$ is a finite-dimensional locally tomographic composite, 
then 
\[(\V(A) \otimes \V(B))_{+} \subseteq \V(AB)_{+} \subseteq (\V(A) \maxtensor \V(B))_{+}.\] 
Hence, as vector spaces (ignoring order), we have \[\V(AB) = \V(A) \otimes \V(B).\]
\end{corollary} 

\begin{corollary}[Corollary to Corollary] A finite-dimensional composite $AB$ is locally tomographic iff 
$\dim(\V(AB)) = \dim(\V(A))\cdot\dim(\V(B))$. 
\end{corollary}

Using this, it's not hard to prove the following \cite{AW-TP}

\begin{theorem}\label{thm: State space of bilateral}  If $A$ and $B$ are finite-dimensional probabilistic models, then 
\[\V(\bilat{AB}) = \V(A) \hat{\otimes}_{\mbox{max}} \V(B).\]
\end{theorem}

{\gray Theorem (\thetheorem) was first proved in the finite-dimensional case by Kl\"{a}y Foulis and Randall \cite{KFR}, using a dimension-counting argument; 
the approach sketched here, which also works (with suitable 
modifications) in infinite dimensions, is from \cite{Wilce92}. This was based on earlier work on tensor products of compact convex sets \cite{NP} and ordered linear spaces \cite{Wittstock} %by Namioka and Phelps \cite{NP}, amongst others. 
It is not hard to show that if $\Omega(A)$ or $\Omega(B)$ is a simplex, then $\V(A) \mintensor \V(B) = \V(A) \maxtensor \V(B)$. A question raised in \cite{NP} %Namioka/Phelp
was whether the converse is true. This was only settled recently, 
in the affirmative, by Aubrun, Lami, Palazuelos and Pl\'{a}vala in \cite{ALPP}. }

\tempout{
{\em Remark:} At this point, we can condense our definition of a non-signaling composite in a nice way. The bilinear 
mapping $m : \V(A) \times \V(B) \rightarrow \V(AB)$ 
extends to a linear mapping $\V(A) \otimes \V(B) \rightarrow \V(AB)$, and since product states go to states, this must 
be positive with respect to the minimal tensor cone. 
We also have an outcome-preserving morphism $\pi : X(A) \times X(B) \rightarrow X(AB)$. This takes (non-signaling) states 
in $\Omega(AB)$ back to non-signaling joint states 
on $A \times B$, giving us a positive linear mapping $\V(AB) \rightarrow \V(A) \maxtensor \V(B)$. 
}

\subsection{Non-signaling states and effects as Mappings} 

In finite dimensions, a bilinear form $F : \V(A)^{\ast} \times \V(B)^{\ast} \rightarrow \R$ is effectively the same thing as a linear mapping $\phi : \V(A)^{\ast} \rightarrow \V(B) = \V(B)^{\ast \ast}$: for 
$a \in \V(A)^{\ast}$,  just define $\phi(a)(b) = F(a,b)$ for all 
%$ to be the unique vector in $\V(B)$ with $b(f(a)) = F(a,b)$ for all 
$b \in \V(B)^{\ast}$.  
%This simple observation, which more or less emerged during the proof of Theorem ..., 
%is very useful. 
Similarly, an element $\omega$ of $\V \otimes \W$ defines a positive linear mapping 
\[\hat{\omega}: \V^{\ast} \rightarrow \W \simeq \W^{\ast \ast}\]
where, for all $a \in \V(A)^{\ast}$ and $b \in \W^{\ast \ast}$, 
\[\hat{\omega}(a)(b) = (a \otimes b)(\omega).\]
Dually, if $f \in (\V \otimes \W)^{\ast} = 
\V^{\ast} \otimes \W^{\ast}$, we 
have an associated linear mapping 
\[\hat{f} : \V \rightarrow \W^{\ast}\] 
given by 
\[\hat{f}(\alpha)(\beta) = f(\alpha \otimes \beta).\]
If $\omega$ belongs to the maximal tensor cone 
of $\V$ and $\W$, then $\hat{\omega}$ is positive, and similarly if $f$ belongs to the maximal 
cone of $\V^{\ast} \otimes \W^{\ast}$,  then 
$\hat{f}$ is positive. 

Suppose now that 
$\V = \V(A)$ and $\W = \V(B)$ for some 
probabilistic models $A$ and $B$, and that 
$AB$ is a locally tomographic composite of these models. If $\omega \in \Omega(AB)$, 
we can represent $\omega$ as an element 
of $\V(A) \otimes \V(B)$. In this case, for 
any effect $a \in \V(A)^{\ast}$, 
$\hat{\omega}(a) = \omega(a,~ \cdot~)$ 
is $\omega_{1}(a) \omega_{2|a} \in \V(B)_+$  --- what 
is sometimes called the {\em un-normalized} 
conditional state of $\omega$ given $a$. For 
this reason, $\hat{\omega}$ is called the 
{\em conditioning map} associated with $\omega$.   
If $f$ is a bipartite effect, we call the mapping $\hat{f}$ the {\em co-conditioning map} 
associated with $f$. 
%Since $\hat{\omega}(a)$ is sub-normalized when $a$ is an effect, we see 
%that $\hat{\omega}$ is a channel. 

\begin{exercise} Show that, similarly, if 
$f$ is an effect in $\V(AB)^{\ast} \simeq \V(A)^{\ast} \otimes \V(B)^{\ast}$, then 
$\hat{f}(u_{A}) \leq u_{B}$.  
\end{exercise} 

%{\bf Channels} 
If $\phi : A \rightarrow B$ is a morphism of models, then 
$\phi^{\ast} : \V(B) \rightarrow \V(A)$ is a positive linear mapping with the property 
that $u_{B}(\phi^{\ast}(\beta)) \leq 1 \ \ \forall \beta \in \Omega(A)$, that is, 
a channel, in the language of Definition 2.9. 

\tempout{ 
 More generally, 
we refer to any positive mapping $\Phi : \V(B) \rightarrow \V(A)$ satisfying (\theequation) as a {\em channel} from $A$ to $B$, and interpret this as some kind of (perhaps purely hypothetical) physical process 
whereby "input" states from $B$ are transformed into "output" states of $A$.  
As an example, any positive linear mapping $\V(\H) \rightarrow \V(K)$ with 
$\Tr(\Phi(a)) \leq \Tr(a)$ for all positive trace-class operators $a$ on $\H$, is a channel 
in this sense. 

{\gray {\em Remark:} Of course, in quantum theory it is standard to require channels to be 
not merely positive, but completely positive. This concept becomes available once we 
have made the choice of a rule for composing models. We will return to this issue in the next chapter. For the moment, we make the following observation (some assembly required): 

\begin{exercise} Show that if $\Phi : \V(A) \rightarrow \V(B)$ is a channel, then so is 
$\Phi \otimes \1 : \V(A) \maxtensor \V(C) \rightarrow \V(B) \maxtensor \V(C)$ for 
any model $C$. \end{exercise} 
}
}

Given a channel $\Phi : \V(B) \rightarrow \V(A)$, we can dualize to obtain a positive 
mapping $\Phi : \V(A)^{\ast} \rightarrow \V(B)^{\ast}$ where 
$\Phi^{\ast}(a)(\beta) = a(\Phi(\beta))$. Note that $\Phi^{\ast}(u_A) \leq u_B$. 
We refer to $\Phi^{\ast}$ as the {\em dual channel} associated with $\Phi$. 

\begin{lemma} Let $\omega$ be a non-signaling state on $A \times B$, and 
let $f$ be an effect in $\V(C)^{\ast} \maxtensor \V(A)^{\ast}$. Then 
$\hat{\omega} \circ \hat{f}$ is a channel from $C$ to $B$, and 
$\hat{f} \circ \hat{\omega}$ is a dual channel from $B$ to $C$. 
\end{lemma} 

\begin{proof} Let $\gamma$ be a state of $C$. Then $\hat{f}(\gamma)$ is an 
effect on $A$, so $\hat{\omega}(\hat{f}(\gamma))$ is a sub-normalized 
state on $B$.  The second statement is proved similarly. $\Box$ 
\end{proof} 

%For any model $A$, let $D(A) = (\D([0,u_A]),\Omega(A))$ be the model associated 
%with the OUS $\V(A)^{\ast}$. 

%\begin{Lemma:} Let $\phi : \V(A)^{\ast} \rightarrow \V(B)^{\ast}$ be any positive 
%linear mapping with $\phi(u_A) \leq u_B$. Then, restrcted to $[0,u_A]$, 
%$\Phi$ is a partial morphism $D(A) \rightarrow D(B)$, and the corresponding dual 
%channel is $\phi$. 
%\end{lemma} 

{\bf Remote evaluation and teleportation} 
One of the most striking applications of quantum information theory is the possibility of using 
an entangled state to construct channel through which 
the state of a system at one location can be 
"teleported" to a second, remote location, 
the original state being destroyed in the process. As it turns out, this possibility is not specifically quantum-mechanical, being available in a wide range of "post-quantum" GPTs. 

This is really an application of the following 
simple observation from linear algebra.  
If $\omega \in \V \otimes \W$, define a 
linear mapping
$\hat{\omega} : \V^{\ast} \rightarrow \W$ by 
$\hat{\omega}(a)(b) = (a \otimes b)(\omega)$. 
Similarly, if $f \in (\V \otimes \W)^{\ast}$, 
define $\hat{f} : \V \rightarrow \W^{\ast}$ 
by $\hat{f}(\alpha)(\beta) = f(\alpha \otimes \beta)$. Of course, in our setting, 
these are the conditioning and co-conditioning maps associated with a bipartite 
state and effect. But the following is independent of this interpretation:

\begin{lemma}[Remote Evaluation] \label{lemma: remote eval} Let $\U, \V, \W$ be any three finite-dimensional vector spaces, and 
let $\alpha \in \U$,  $\omega \in \V \otimes \W$,  $f \in (\U \otimes \V)^{\ast}$, and 
$e \in \W^{\ast}$. Then 
\[(f \otimes e)(\alpha \otimes \omega) 
= e(\hat{\omega} \circ \hat{f})(\alpha).\]
\end{lemma} 

Note: $\hat{f}(\alpha) \in \V^{\ast}$, and 
$\hat{\omega}(\hat{\alpha}) \in \W$, so 
this type-checks. 

\begin{exercise} Prove lemma \ref{lemma: remote eval} with the following hint: verify the equation when 
$\omega$ and $f$ are pure tensors (say, 
$\omega = \beta \otimes \gamma$ and 
$f = a \otimes b$). Then extend by linearity. 
\end{exercise} 

In the particular case of three probabilistic models $A$, $B$, and $C$, representing three physical systems, suppose Alice controls a composite system $AB$ and Clovis controls $C$.
Suppose, further, that we have composites $(AB)C$ And $A(BC)$, 
and that these are isomorphic, so that effects on one can be evaluated at states on the other\footnote{This condition, which holds both classically and quantumly, will always 
be satisfied when our models sit inside a monoidal probabilistic theory, a topic we take up in the next chapter.} 
System $A$ in an unknown state $\alpha$, while $B$ and $C$ share a known state $\omega$. If Alice performs a measurement on 
$AB$ and obtains an outcome represented by an effect $f$, then Clovis's state, conditional on Alice's having obtained this outcome, is $\hat{\omega}(\hat{f}(\alpha))$. Let us call this {\em remote evaluation} of the function $\hat{\omega} \circ \hat{f}$. 

\begin{definition} $\omega$ is an {\em isomorphism} state iff $\hat{\omega}$ is a positive isomorphism $\V(B)^{\ast} \rightarrow \V(C)$. An {\em isomorphism effect} is defined dually. The inverse of an isomorphism state is an isomorphism effect, and vice versa. 
\end{definition}

Suppose now that $C = A$, that is, Clovis's system is a copy of Alice's system $A$. Let $f$ be an isomorphism effect  on $AB$ and let $\omega$ be an isomorphism state of $BA$, and let $g = \hat{\omega} \circ \hat{f} : A \simeq A$ (noting that this is invertible). Then if Alice performs a measurement of $f$ when the shared system is in state $\alpha \otimes \omega$, Bob's conditional state $(\alpha \otimes \omega)_{B|f}$ is $g\alpha$. 
Assuming that Alice knows the effect $f$ and the state $\omega$ 
(perhaps having engineered both), she also knows $g$. Telephoning Bob, she asks him to implement the symmetry $g^{-1}$ on 
his copy of $A$. The result is that Bob's state is 
now $\alpha$. Alice has (conclusively) teleported 
$\alpha$ from her copy of system $A$ to Bob's. 

A stronger form of teleportation requires us to find an entire test's worth of isomorphism effects on $AB$. Suppose $\{f_i\}$ is a collection of isomorphism effects with $\sum_{i} f_i = u_{AB}$, and that this corresponds to a measurement that Alice can make: upon obtaining outcome $f_i$, she instructs Bob to implement 
the symmetry corresponding symmetry $g_{i}^{-1}$, 
so as to leave his system in state $\alpha$. This is called {\em deterministic teleportation}, since in this scenario it is {\em certain} that the state will be successfully teleported. 
The existence of a partition of $u_{AB}$ by isomorphism effects is a stronger constraint on $AB$, but one can still find examples that are neither classical nor quantum in which 
this is possible. See \cite{BBLW-Teleportation} for the details.

{\gray {\em Gratuitous Remark:} For what it's worth, I found it very hard to understand what I read about teleportation protocols (which always seemed to place heavy emphasis on state collapse) until I realized it all comes down to Lemma \ref{lemma: remote eval}. }
%For more on teleportation in GPTs, see 
%the consequences of this simple 
%observation, see 
%This is the core of the analysis of teleportation 
%in GPTs from 
%\cite{BBLW-Teleportation}.}

\section{Probabilistic Theories}
 %, Some gentle remarks... 2014
{\fontsize{8}{9}\selectfont
\epigraph{\fontsize{9}{11}\selectfont A model is a mathematical structure in the same sense that the Mona Lisa is a painted piece of wood.}{\textit{\fontsize{9}{11}\selectfont B. van Fraassen \cite{van Fraassen}}}
}

Having described probabilistic models and composites thereof, we are in a position finally to say what we mean by a {\em probabilistic theory.} Presumably, this should involve some specified collection of models, representing the kinds of systems one wants to study, but also a specification of certain mappings --- let us say, {\em processes} --- connecting these models. This immediately suggests that a probabilistic theory might be thought of as a {\em category} of probabilistic models. While this is not {\em quite} the picture I will ultimately advocate, it is a good first approximation. In any case, it will be helpful to start with a short review of basic category-theoretic ideas.  A good general reference for this material  is the book by Emily Riehl \cite{Riehl}. For a lighter but very nice overview, see \cite{FS}. 

\subsection{Categorical Fundamentals} 

A {\em category} consists of a class $\Cat$ of 
{\em objects} and, for every pair $(A,B)$ of objects,  a set $\Cat(A,B)$ of {\em morphisms}, or {\em arrows}, from $A$ to $B$.\footnote{\green The words {\em class} and {\em set} are to be taken literally here: the objects 
may form a proper class, but the morphisms between two objects 
form a set. Some authors refer to this as a {\em locally small} category, allowing for more general structures in which 
$\Cat(A,B)$ can also be a proper class.} 
There is also specified
\begin{mlist} 
\item[(a)] A composition rule 
\[\circ : \Cat(A,B) \times \Cat(B,C) \rightarrow \Cat(A,C)\]
such that $f \circ (g \circ h) = (f \circ g) \circ h$ whenever one side is (and hence, both are) defined; 
\item[(b)] For each object $A$, an {\em identity morphism} 
$\id_{A} \in \Cat(A,A)$ such that $\id_{A} \circ f = f$ 
for any $f \in \Cat(B,A)$ and 
$f \circ \id_{A} = f$ for any $f \in \Cat(A,B)$. 
\end{mlist} 

Familiar examples include the category $\RVec$ of real vector 
spaces and linear mappings; the category $\Set$ of sets 
and mappings;  the category $\Grp$ of groups and 
group homomorphisms. Also, any poset $(L,\leq)$ 
can be regarded as a category having elements of $L$ 
as objects, with a unique morphism in $\Cat(a,b)$ if $a \leq b$, and $\Cat(a,b) = \emptyset$ otherwise.  A monoid (a semigroup with identity) is essentially the same thing as a one-object category. 

{\em Notation and Terminology:}  I will write $A \in \Cat$ to mean that $A$ is an object of $\Cat$ (though technically this is 
an abuse of notation). Sets of the form $\Cat(A,B)$ are usually called {\em Hom-sets}, reflecting the once more common notation $\Hom_{\Cat}(A,B)$ (itself an echo of the set of homomorphisms between two algebraic structures). 

{\bf The Category $\Prob$} In Section 1.4, we defined a morphism from a probabilistic model $A$ to a probabilistic model $B$ to be a mapping $\phi : X(A) \rightarrow X(B)$ such that $\phi(x) \perp \phi(y)$ for all $x, y \in X(A)$ with $x \perp y$, $\phi(a) \in \Ev(B)$ for all $a \in \Ev(A)$, and $\beta \circ \phi \in \Omega(A)$ for every state $\beta \in \Omega(B)$. It's clear that the composition of morphisms $A \rightarrow B$ and $C \rightarrow D$ is still a morphism, and that the identity mapping $X(A) \rightarrow X(A)$ is a morphism. Hence, we have a pretty general category of probabilistic models and morphisms.  We'll denote this 
by $\Prob$. 

It's tempting at this point to define a probabilistic theory to be a sub-category of $\Prob$. There are two reasons to resist this temptation. The first is that we'd like our probabilistic theories to be equipped with a compositional structure, and usually one that admits both entangled states and entangled effects.  The "native" compositional structures on 
$\Prob$, $\times_{NS}$ and $\bilat{\,  \, }$, allow for the former but not the latter, and are therefore not usually suitable; so we end up having to define our composition rule in a theory-specific way. At best, then, a probabilistic theory would be a sub-category of $\Prob$ with extra structure: a chosen compositional rule, specific to that sub-category. The second reason is that we might sometimes want to consider non process-tomographic theories: those in which there exist distinct {\em physical} processes $A \rightarrow B$ that are represented by the same morphism  between certain 
probabilistic {\em models of} $A$ and $B$.  This suggests 
that probabilistic theories should assign models in $\Prob$ 
to "systems" in a process theory $\Cat$. In other words, 
it should be a {\em functor} --- a term I'll now review.

{\bf Functors} It's often important to be able to shift from one category to another in a structure-preserving way. A {\em functor} from a category $\Cat$ to another, $\Dat$, is an assignment of an object $FA \in \Dat$ to 
every object $A \in \Cat$, and of a morphism 
$Ff \in \Cat(FA, FB)$ for every $f \in \Cat(A,B)$, so 
that 
\begin{itemize} 
\vspace{-.1in}
\item[(a)] $F(\id_{A}) = \id_{FA}$;
\item[(b)] $F(f \circ g) = Ff \circ Fg$ 
for all morphisms $f$, $g$ with $f \circ g$ defined. 
\end{itemize} 
\vspace{-.1in}
\tempout{
\[
\begin{tikzcd}
A \arrow[dd, "f" left]  
& &  
F(A)\arrow[dd, "Ff"] \\
 &  \Longrightarrow & \\
B  & & F(B)
\end{tikzcd}
\]
}
A {\em contravariant} functor $\Cat \rightarrow \Dat$ is defined 
in the same way, except that it reverses the order of composition: $F(f \circ g) = Fg \circ Ff$. To emphasize the distinction, functors as defined above are often called {\em covariant} functors. \footnote{\gray An equivalent
way to put things is to say that a contravariant 
functor $\Cat \rightarrow \Dat$ is a covariant functor $\Cat^{\op} \rightarrow \Dat$, where $\Cat^{\op}$ denotes the {\em opposite category} of $\Cat$. This has the same objects, but morphisms in $\Cat^{\op}(A,B)$ are in fact morphisms in $\Cat(B,A)$, and composition is reversed: $f \circ_{\op} g = g \circ f$.}

\begin{example} The power set construction gives us two functors, one covariant and one contravariant, on $\Set$. In both, each set $X$ is taken to its power set,  $\Pow(X)$. In the covariant version, a mapping $f : X \rightarrow Y$ is taken to the set mapping $\Pow(X) \rightarrow \Pow(Y)$ given by $a \mapsto f(a)$ 
%$ = \{ f(x) | x \in a\}$. 
For the contravariant version,  $f$ is taken to the mapping 
$\Pow(Y) \rightarrow \Pow(X)$ given by $a \mapsto f^{-1}(a)$.
\end{example} 

\begin{example}  There is a canonical contravariant "linearization" functor $\Set \rightarrow \Vec$ 
given by 
\[X \mapsto \R^{X} \ \ \mbox{and} \ \ f \in Y^{X} \ \mapsto \ f^{\ast} : \R^{Y} \rightarrow \R^{X}\]
where $f^{\ast}(\beta) = \beta \circ f$. There is also a covariant functor, defined as follows. Let $\R^{[X]}$ be the set of all {\em finitely supported} functions (functions taking the value $0$ at all but finitely many points) $\phi : X \rightarrow \R$. Any such function has a unique expression $\phi = \sum_{x \in F} \phi(x) \delta_{x}$ where $F$ is a finite subset of $X$, $\phi(x) \not = 0$ for all $x \in F$, and $\delta_{x}$ is the point-mass at $x$. Given a mapping $f : X \rightarrow Y$, we define $f_{\ast}(\phi) = \sum_{x \in F} \phi(x) \delta_{\phi(x)}$. It's easy to check that $f_{\ast}$ is linear, and that $(f \circ g)_{\ast} = f_{\ast} \circ g_{\ast}$. 
\end{example} 

\begin{exercise} Show that $\V(~\cdot~)$ defines a contravariant functor $\Prob \rightarrow \Vec$, while $\V^{\ast}(~ \cdot ~)$ is a covariant such functor. 
\end{exercise} 

%If $\Cat$ and $\Dat$ are categories, we can construct a category $\Cat \times \Dat$ in which objects are pairs of objects $(A,B)$ with $A \in \Cat$ and $B \in \Dat$, and a morphism $(A,B) \rightarrow (A',B')$ is  a pair of morphisms $(f,g)$ with $f \in \Cat(A,A')$ and $g \in \Cat(B,B')$, with the obvious composition rule. 

In general, the image of a category $F : \Cat \rightarrow \Dat$ --- that is, the collection of objects $F(A)$ for $A \in \Cat$, and of morphisms $Ff$ for $f \in \Cat(A,B)$ --- is not a sub-category of $\Cat$, because one can have a situation in which $F(A) = F(B) =: D$, so that for morphisms $f : C \rightarrow A$ and $g : B \rightarrow C'$, $Ff : F(C) \rightarrow D$, $Fg : D \rightarrow F(C')$, but $g \circ f$ is not defined, and $Fg \circ Ff$ is not of the form $F(h)$ for any $h : C \rightarrow C'$. 

\begin{exercise} Find an example illustrating this possibility. (Hint: think small.) \end{exercise} 

\begin{exercise} Show that if $F : \Cat \rightarrow \Dat$ is injective on objects, then $F(\Cat)$ is a subcategory of $\Dat$.
%Explain what goes wrong 
%if $F$ is {\em not} injective on objects. 
\end{exercise} 

{\bf Natural Transformations} 
The single most important idea in Category theory is the following:

\begin{definition} A {\em natural transformation} 
from a functor $F : \Cat \rightarrow \Dat$ to a functor $G : \Cat \rightarrow \Dat$  is an assignment, 
to all pairs of objects $a \in A$, of 
a morphism $\phi_{a} : Fa \rightarrow Ga$ such that for 
every $f \in \Cat(A,B)$, 
\[
\begin{tikzcd}
%A \arrow[dd, "f"]  
%& &  
FA \arrow[r,"\phi_A"] \arrow[d,"Ff"] & GA \arrow[d, "Gf"]\\
% &  \Longrightarrow & \\
%B  & & 
F(B) \arrow[r,"\phi_B"] & GB
\end{tikzcd}
\]
commutes.  The definition for contravariant functors is the same.  The morphism $\phi_{A}$ is called the {\em component} of the natural transformation {\em at} $A$.  
\end{definition}

\begin{example} Suppose $\Cat$ is a category with a single object, $1$. Then the collection $\Cat(1,1) =: \Cat(1)$ is a monoid (semigroup with identity). A functor $F : \Cat \rightarrow \Set$ picks out a set $S := F(1)$, and assigns to every $g \in \Cat(1)$, a mapping $Fg : S \rightarrow S$, such that $F(gh) = F(g) \circ F(h)$. This is just to say that $F$ specifies a set,  and an action of the monoid $\Cat(1)$ on this set. Given two such 
functors, say $F$ and $G$, with $F(1) = S$ and $G(1) = T$, 
we see that a natural transformation $\phi  : F \rightarrow G$ 
is just an mapping $\phi : S \rightarrow T$ such that 
$\phi(gx) = g\phi(x)$ for every $x \in s$. In other words, 
$\phi$ is an equivariant mapping. 
\end{example} 

An {\em isomorphism} in a category $\Cat$ is a morphism 
$\phi : A \rightarrow B$ in $\Cat$ having an inverse: that 
is for which there exists a morphism $\psi : B \rightarrow A$ 
with $\psi \circ \phi = \id_{A}$ and $\phi \circ \psi = \id_{B}$. 
It is easy to check that such an inverse, if it exists, 
is unique, so we can write $\psi = \phi^{-1}$. 
If a natural transformation $\phi : F \rightarrow G$ is 
such that for every object $A \in \Cat$, $\phi_{A}$ is 
an isomorphism in $\Cat$, then we say that $\phi$ is a 
{\em natural isomorphism} from $F$ to $G$, writing 
$\phi : F \simeq G$. In this case we also have a natural isomorphism $\phi^{-1} : G \simeq F$ with components $(\phi^{-1})_{A} = (\phi_{A})^{-1}$.

\begin{exercise}  Let $\Expo$ and $\Pow$ be the contravariant functors $\Set \rightarrow \Set$ given on objects by 
$\Expo(X) = 2^{X}$ and $\Pow(X)$ = the power set of $X$, respectively, and on morphisms (mappings) by $\Expo(f) = f^{\ast} : \alpha \mapsto \alpha \circ f$, and $\Pow(f) : a \mapsto f^{-1}(a)$, respectively. Find a natural isomorphism $\Expo \rightarrow \Pow$, and carefully check that it actually is one. 
\end{exercise}

{\gray
Suppose now that we have functors $F : \Cat \rightarrow \Dat$ 
and $G : \Dat \rightarrow \Cat$. Then we have two composite functors, $G \circ F : \Cat \rightarrow \Cat$ and $F \circ G : \Dat \rightarrow \Dat$. If these are respectively the 
identity functors on $\Cat$ and $\Dat$, then we say that 
$F$ and $G$ are mutually inverse {\em isomorphisms of categories}. But a weaker notion is available that is often 
more useful: we say that $F$ and $G$ define an {\em equivalence} between $\Cat$ and $\Dat$ iff there exist natural isomorphisms 
$\phi : G \circ F \simeq \id_{Cat}$ and $\psi : F \circ G \simeq \id_{\Dat}$. The standard example of categories that are 
equivalent is as follows: 

\begin{exercise} Let $\Cat$ be the category having natural 
numbers as objects, with a morphism from $n$ to $m$ being 
an $m \times n$ matrix, with matrix multiplication as 
composition. Let $\Dat$ be the category of finite-dimensional 
vector spaces having a preferred ordered basis. 
Let $F(n) = \R^{n}$ and for an $m \times n$ matrix $A$, 
let $F(A) : \R^{n} \rightarrow \R^{m}$ be multiplication by $A$. 
(a) Find a functor $G : \Dat \rightarrow \Cat$, and 
(b) find natural isomorphisms $G \circ F \simeq \id_{\Cat}$ 
and $F \circ G \simeq \id_{\Dat}$. Finally, (c), explain why these categories can't be isomorpic. (Hint: how many objects 
does $\Cat$ have?) 
\end{exercise} 
}

{\gray {\em Remark:} Let $\Func(\Cat,\Dat)$ denote the class of all functors $\Cat \rightarrow \Dat$: we can think of this as a category with Hom-sets $\Nat(F,G)$.  Note this will seldom be locally small!}

\subsection{Monoidal Categories and Process Theories}

Given a category $\Cat$, we can construct a new 
category $\Cat \times \Cat$, in which objects are 
pairs $(A,B)$ of objects in $\Cat$, and a morphism 
$(A,B) \rightarrow (C,D)$ is a pair $(f,g)$ of morphisms $f \in \Cat(A,C)$ and $g \in \Cat(C,D)$.    A functor 
$\Cat \times \Cat \rightarrow \D$ is called a {\em bi-functor} on $\Cat$. 

A {\em symmetric monoidal category} (SMC) is a category $\Cat$, 
equipped with a bifunctor $\Box : \Cat \times \Cat \rightarrow \Cat$, 
plus, for all objects $A, B \in \Cat$, a {\em swap}(or {\em symmetrizer}) morphism 
$\sigma_{A,B} : A \Box B \rightarrow B \Box A$, for all 
triples of objects $A,B,C \in \Cat$, an {\em associator} morphism 
$\alpha_{A;B,C} : A \Box (B \Box C) \rightarrow (A \Box B) \Box C$, 
and, finally, a {\em unit object} $I \in \Cat$ and, for 
all objects $A$, morphisms $\lambda_{A} : I \Box A \rightarrow A$ 
and $\rho_{A} : A \Box I \rightarrow A$, called left and right 
{\em unitors}, such that various diagrams --- called 
"coherences" --- all commute.  One of these is the 
{\em associator coherence}
%One of them is the following {\em hexagon coherence}, relating the %associator and the swap
(here, I've suppressed the subscripts, which can be pencilled in from context):
\[
\begin{tikzcd}
& (A \Box B) \Box (C \Box D) \arrow[rd, "\alpha"] \arrow[ld,"\alpha"] & \\
 ((A \Box B) \Box C) \Box D \arrow[d, "\alpha"] & & 
 A \Box (B \Box (C \Box D)) \arrow[d, "\alpha"] \\
(A \Box (B \Box C)) \Box D \arrow[rr,"\alpha"]  & &    A \Box ((B \Box C) \Box D)  
\end{tikzcd}
\]
\tempout{Another is the {\em hexagon coherence}, which relates $\alpha$ and 
$\sigma$ (again, I suppress indices):
\[
\begin{tikzcd}
(A \Box B) \Box C \arrow[rr, "\alpha"] \arrow[d,"\sigma \otimes C"] & & A \Box (B \Box C) \arrow[d, "\sigma_{A, B \Box C}"]\\
(B \Box A) \Box C  \arrow[d,"\alpha"]  & &    (B \Box C) \Box A 
\arrow[d,"\alpha"] \\
B \Box (A \Box C) \arrow[rr,"B \otimes \sigma"]   & & B \Box (C \Box A)
\end{tikzcd}
\]
}
%Another is the {\em Pentagon} (or {\em associator}) coherence: 
Effectively, this says that $\Box$ is associative 
up to a natural isomorphism, so that for all objects $A, B, C$, 
we a canonical isomorphism 
$A \Box (B \Box C) \simeq (A \Box B) \Box C$. Two further coherences (which I have not written down) 
guarantee that $I$ behaves like a unit: $I \Box A \simeq A \Box I \simeq A$, and that $\Box$ is commutative, $A \Box B \simeq B \Box A$, again up to the given natural isomorphisms. 

One says that a SMC is {\em strict} if the associators and left and right unitors are identities, so that 
$A \Box (B \Box C) = (A \Box B) \Box C$ and 
$I \Box A = A \Box I = A$ for all objects $A, B, C$.  (In general, even when this is so, the swap morphisms will be non-trivial.)  Every SMC $\Cat$ has a canonically equivalent "strictification". 
See \cite{FS, Riehl} for further details.

\begin{example} $\Set$, with $A \Box B = A \times B$ and 
$\Vec$ with $A \Box B = A \otimes B$ are SMCs. So 
is any join-semilattice, regarded as a category, 
under $a \Box b = a \vee b$. 
In particular, the collection of open subsets of a topological space is a SMC in this way.
\end{example} 

{\gray 
\begin{exercise}[Dull Exercise in Bookkeeping] Pencil in the missing subscripts on all the $\alpha$-s in the diagram above.
\end{exercise} 

\begin{exercise}[for Obsessives] (a) Look up the remaining coherences in the definition of a SMC. (b) Figure out what the associator, swap, and unitors are for the category $\Set$ with Cartesian product, and tediously check that the coherences are all satisfied. (Or perhaps this isn't so much tedious as relaxing, somewhat like playing solitaire.)
\end{exercise}
}
 An elementary but very important consequence of the definition of an SMC is that if we have morphisms $f \in \Cat(A,B)$, $h \in \Cat(B,C)$ 
and $g \in \Cat(A',B')$, $k \in \Cat(B',C')$, then 
\begin{equation} 
(f \Box g) \circ (h \Box k) = (f \circ h) \Box (g \circ k).
\end{equation} 

{\bf Scalars} An important consequence of (\theequation) concerns the set $\Cat(I,I)$. As with any object 
in the form $\Cat(A,A)$ in any category, this is a monoid under $\circ$. However, 
if $(\Cat, \Box)$ is a strict SMC, then (\theequation) implies that, for all  $s, t \in \Cat(I,I)$, 
$s \circ t = s \Box t$. Using the symmetry of $\Box$, one can show that 
$s \circ t = t \circ s$. In other words, $\Cat(I,I)$ is a {\em commutative} monoid. 
See \cite{AC} or \cite{CK} for the details. It is usual to refer to elements of $\Cat(I,I)$ as {\em scalars}. 

%{\blue [Discuss scalars here?]} 

{\em Remark:} It is very common to see the monoidal product in an abstract SMC denoted by $\otimes$. When the category is one in which objects are finite-dimensional vector spaces, however, I will always reserve $\otimes$ for the usual tensor product, using a different symbol if I want to discuss a different monoidal product. 

{\bf Process Theories} 
In parts of the GPT-adjacent quantum-foundational 
literature, SMCs are referred to as {\em process theories}. %[Explain] 
The idea is that objects are physical systems, and morphisms 
are physical processes having these systems as inputs and outputs. 
One understands $I$ as "nothing", i.e., the absence of a system. 
Thus, a morphism $\alpha : I \rightarrow A$ is a morphism that 
produces something from nothing; this is usually interpreted as the preparation of a "state" of $A$. Similarly, a morphism $a : A \rightarrow I$ is a process with no output-system; this is usually understood as a (destructive) measurement outcome or "effect". We will need 
to be careful with this language, however as such "states" and 
effects do not always correspond exactly to states and effects as defined earlier. I will come back to this point below.

 Finally, 
$A \Box B$ is understood as a composite system obtained by setting $A$ and $B$ "side by side", and $f \Box g : A \Box B \rightarrow C \Box D$ represents the processes $f : A \rightarrow C$ and $g : B \rightarrow D$ operating "in parallel". Process theories become probabilistic if we are given a rule for assigning probabilities to "circuits", a term I'll explain 
%{\blue Do so!} 
presently.  For more on the connection between this point of view and the GPT framework, see \cite{AW-shortcut}

{\bf Graphical Language} If $(\Cat, \Box, \mbox{etc})$ is a {\em strict} SMC, one can 
represent expressions involving the compositional and monoidal structure 
in terms of certain diagrams in a visually appealing way. 
The convention is that systems (objects) are represented by 
"wires" (lines or other curves), and processes of various sorts, by "boxes" of various shapes. Composite systems are represented 
by parallel wires, and boxes can have any number of input or output wires; e.g,. a box representing a morphism $A \Box B \rightarrow C \Box  D \Box E$ will have two input wires, labeled $A$ and $B$, and three output wires, labeled $C, D$ and $E$ (see Figure (a) below). 
Composition of processes is represented by a sequential hooking together of boxes via wires, and the flow of 
"time" (the order of composition) is upwards, from the bottom of the page towards the top. Identity morphisms are not drawn, but it's handy to think of a wire labeled by, say, $A$, as standing equally for the object {\em and} its identity morphism. States are usually drawn as triangular boxes, "pointing down", and 
effects, as triangular boxes "pointing up", as in figure (b):
%in sch a way that some elementary and 
%visually intuitive diagram-modification rules correspond to permissible 
%manipulations using the coherences above. 
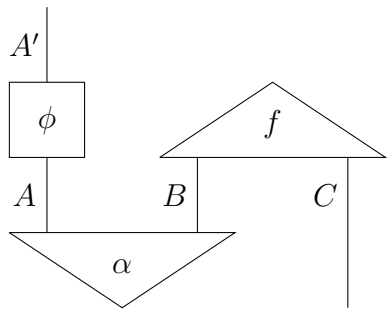
\[
\begin{array}{ccc}
\begin{tikzpicture} 
\draw (0,0) rectangle (2,1);
\draw (.5,-1) -- (.5,0);
\draw (1.5,-1) -- (1.5,0);
\draw (.3,1) -- (.3,2);
\draw (1,1) -- (1,2);
\draw (1.7,1) -- (1.7,2); 
\draw (1,.5) node {$\phi$};
\draw (.3,-.5) node {$A$}; 
\draw (1.3,-.5) node {$B$};
\draw (.5,1.5) node {$C$}; 
\draw (1.2,1.5) node {$D$};
\draw (1.9,1.5) node {$E$};
\end{tikzpicture}
& \hspace{.5in} & 
\begin{tikzpicture} 
\draw (0,1) -- (3,1) -- (1.5,0) -- (0,1);
\draw (1.5, .55) node {$\alpha$};
\draw (0,2) rectangle (1,3);
\draw (.5,1) -- (.5,2);
\draw (.5,3) -- (.5,4);
\draw (2.5,1) -- (2.5,2);
\draw (2,2)  -- (5,2) -- (3.5,3) -- (2,2);
\draw (4.5,2) -- (4.5,0);
\draw (3.5,2.45) node {$f$};
\draw (.5,2.5) node {$\phi$};
\draw (.2,1.5) node {$A$}; 
\draw (.2,3.5) node {$A'$};
\draw (2.2,1.5) node {$B$}; 
\draw (4.2,1.5) node {$C$};
\end{tikzpicture} \\
\mbox{(a)} & & \mbox{(b)} \\
%& \mbox{Figure} & 
\end{array} 
\]
%\end{document}
This illustrates a system in which we have a 
state $\alpha : I \rightarrow A \Box B$, an effect 
$f : B \Box C \rightarrow I$, and a process $\phi : A \rightarrow A'$, 
combined to form 
\begin{equation} \phi \circ (\id_{A} \Box f) \circ (\alpha \Box \id_{C}) : C \longrightarrow A'. 
\end{equation} 
(This expression makes sense because $C = I \Box C$ and
$A' = A' \Box I$.) 
Whether the diagram or the one-line expression above is easier on the eye will depend on the eye in question. However, when the category $(\Cat, \Box, I)$ has some extra structure --- specifically, if it is compact closed or, even better, dagger-compact \cite{AC} --- this graphical notation, and various modifications of it, support a powerful graphical {\em calculus} in which the identities defining the compact structure are replaced by simple graph-rewriting rules. I will not discuss that here, but refer you to the paper \cite{AC} of Abramksy and Coecke, or the book \cite{CK} by Coecke and Kissinger. 
%\end{document}
%The swap morphism is represented by a crossing of wires. [Illustrate]

%\tempout
{%\blue 
{\bf Circuits} Let us say that a diagram like this is a {\em circuit} iff it has only the trivial input and output system, $I$.  For example, the diagram in figure (b) above becomes a circuit if we add an effect $a' : A' \rightarrow I$ for $A'$ and an initial state $\gamma : I \rightarrow C$ for $C$: 
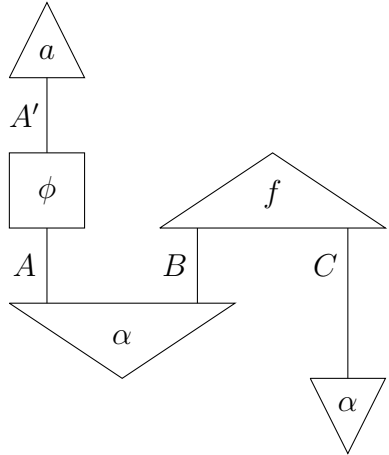
\[\begin{tikzpicture} 
\draw (0,1) -- (3,1) -- (1.5,0) -- (0,1);
\draw (1.5, .55) node {$\alpha$};
\draw (0,2) rectangle (1,3);
\draw (.5,1) -- (.5,2);
\draw (.5,3) -- (.5,4);
\draw (2.5,1) -- (2.5,2);
\draw (2,2)  -- (5,2) -- (3.5,3) -- (2,2);
\draw (4.5,2) -- (4.5,0);
\draw (3.5,2.45) node {$f$};
\draw (.5,2.5) node {$\phi$};
\draw (.2,1.5) node {$A$}; 
\draw (.2,3.5) node {$A'$};
\draw (2.2,1.5) node {$B$}; 
\draw (4.2,1.5) node {$C$};
\draw (4,0) -- (5,0) -- (4.5,-1) -- (4,0);
\draw (0,4) -- (1,4) -- (0.5,5) -- (0,4);
\draw (.5,4.34) node {$a$};
\draw (4.5,-.35) node {$\alpha$};
\end{tikzpicture} \]
If $\Phi : C \rightarrow A'$ is the mapping defined in (\theequation), then the diagram above represents 
$a \circ \Phi \circ \alpha : I \rightarrow I$.  In other words, a circuit is a collection of processes that compose (using $\otimes$ and $\circ$) to yield a process in $\Cat(I,I)$, i.e., 
a scalar. 
%  One can show that in any SMC, sequential and parallel composition on $\Cat(I,I)$ coincide, and make the latter a commutative monoid. 
%If we pick a monoid homomorphism $p: \Cat(I,I) \rightarrow \R_{+}$, this gives us 
%a way of attaching a real number to any circuit. If we choose a special effect 
%$u_A$ for every object $A$, we can use this to pick out a preferred set of 
%normalized states on $A$, and on this basis, construct 
%a probabilistic model associated with any object $A \in \Cat$. {\blue [Normalization!!]} 
}
%See \cite{Wilce-recipe} for 
%details. 
%In the special case in which $\Cat$ is the category of Hilbert spaces and arbitrary 
%linear mappings, $\Cat(I,I)$ is essentially the monoid of complex numbers under %multiplication. If we choose $p : \C \rightarrow \R_{+}$ to be $p(z) = |z|^2$, we 
%come close to recovering 

\subsection{Probabilistic Theories}  

At this point, we can make official the proposal from the end of Section 4.1:

%one might define\footnote{as, earlier, 
%we provisionally {\em did} define} a {\em probabilistic theory} to be a category of %probabilistic models, for instance, a subcategory of $\Prob$.  However, a much more %flexible approach is the following: 
%One problem with this approach is that the probabilistic theories we care most about are those equipped with a non-signaling compositional structure, ideally one turning our theory into a symmetric monoidal category, and the large ambient category $\Prob$ does not support any universally serviceable such compositional structure. Rather, the monoidal structures we encounter are theory-specific. 
% Another issue is that in practice, physical theories begin with dynamical models of some sort, and attach probabilistic structure to these after the fact. 
%Therefore, 
%rather than defining probabilistic theories to be categories {\em of} probabilistic models, 
%we will adopt the following more flexible approach:

\begin{definition}\label{def: probabilistic theories} A {\em probabilistic theory} is a 
functor $F : \Cat \rightarrow \Prob$ where $\Cat$ is a category, understood as a 
 theory of physical systems and processes.   For purposes of these notes, we will also require that $F$ be injective on objects.   
\end{definition} 

The idea is that objects in $\Cat$ are physical systems, or 
perhaps mathematical proxies for these (sites in a spin lattice, regions of spacetime, etc.) and that for a given system $A$, $F(A)$ is the probabilistic model assigned to that system by the theory.  In supposing that $F$ is injective on objects, we are assuming that different physical systems are to be represented by distinct probabilistic models. This is a weak requirement, since we can simply label models by the names of the systems they are to represent.  The payoff is that we then have an image category, $F(\Cat)$.  If $\Cat$ is an SMC --- a 
"process theory" in the usual sense --- then we can impose the further requirement that $F$ map $\Cat$'s monoidal composition rule to a reasonable non-signaling compositional rule on $F(\Cat)$. I will return to this below. First, however, let's consider a few examples of probabilistic theories without worrying yet about monoidal structure.

\begin{example}\label{ex: Mackey functor} Suppose $\Cat$ is a category in which objects are (say, complex) Hilbert spaces and morphisms are isometries, that is, not-necessarily surjective linear mappings preserving inner products. Then $\H \mapsto (\F(\H), \Omega(\H))$ --- with its obvious action on isometries --- is a probabilistic theory. I will call this the {\em Mackey functor}. This is one (very simple) version of quantum theory, which we might call {\em unitary QM} (since the symmetries, i.e., the invertible processes, are given by unitaries).  We can also consider {\em projective unitary QM}, in which we map $\H$ to $(\F_{\P}(\H),\Omega_{\P}(\H))$ where $\F_{\P}(\H)$ is the collection of maximal families of rank-one projections, 
$\Omega_{\P}(\H)$ is the set of probability weights 
of the form $\alpha(p) = \Tr(Wp)$ where $p$ is a rank-one projection, and with an isometry $U : \H \rightarrow \H'$ going to the morphism $\phi_{U}(p) = UpU^{\ast}$. Call this the {\em projective Mackey functor}. 
\end{example} 

\tempout{
\begin{example}\label{ex: absolutely continuous example} Let $\Meas$ denote the category having, as objects, finite measure spaces, that is, pairs $(S,\mu)$ where $S$ is a measurable 
space and $\mu$ is a fixed, finite, positive measure on $S$.  Recall that a measure $\nu$ on $S$ is {\em absolutely continuous} with respect to $\mu$, written $\nu \ll \mu$, iff $\mu(a) = 0$ implies $\nu(a) = 0$ for all measurable sets $a \subseteq S$. A morphism $f : (S,\mu) \rightarrow (T,\nu)$ will be  a measurable function $f : S \rightarrow S'$ such that $f^{\ast}(\nu) \ll \mu$. Let $\M(S)$ be the collection of countable measurable partitions of $S$, as usual, and let 
$\Omega(S,\mu)$ be the collection of probability measures on $S$ absolutely continuous with respect to $\mu$, interpreted as probability weights on $\M(S)$, and let 
$M(S,\mu) = (\M(S),\Omega(S,\mu))$. For each morphism $f : (S,\mu) \rightarrow (T,\nu)$, let $M(f) : \M(T) \rightarrow \M(S)$ be the usual test-space morphism of 
{\blue Example [..].}
\begin{mlist}  
\item[(a)] Show that $M$ is a probabilistic theory. 
\item[(b)] Is it monoidal? 
\end{mlist} 
\end{example} 
}

\begin{example}\label{ex: quantization functor} Suppose $\Cat$ is a category of finite measure spaces $(A,\mu_A)$ and measure-preserving 
mappings (perhaps the category of configuration-spaces of some classical mechanical systems, each with its Liouville measure). For each $A \in \Cat$, let $\H(A) = L^2(A,\mu_{A})$.  This is a functor from $\Cat$ to the category $\Hilb$ of Hilbert spaces and isometries. Composing this with either of the functors $F$ from Example \ref{ex: Mackey functor}, we have a simple version of "quantization". \end{example} 

\begin{example}\label{ex: line example} Let $\Cat = (\R, \leq)$, the linearly ordered set of real numbers, thought of as representing time --- and 
thought of as a category. A probabilistic theory over $\Cat$ would assign a model $(\M_{t}, \Omega_{t})$ to each $t \in \R$. This would give us a picture of some system evolving over time. 
\end{example} 

%{\bf Example:} Let $\Cat$ be any process theory .. [Recipe]. Maybe later. 

{\gray 
\begin{example}[\cite{AW-SC}] Let $G : \FinSet \rightarrow \Grp$ be 
a functor from the category of finite sets and mappings to that of groups and 
group-homomorphisms. For each finite set $E$, let $\sigma : S(E) \rightarrow G(E)$ be an embedding of the 
symmetric group on $E$ into $G(E)$ (that is, suppose such an 
embedding exists for each finite set $A$, and that one such 
embedding has been selected.) We can build a probabilistic 
theory this way: think of each $E$ as a reference experiment. 
Choose a reference outcome $x_o \in E$ (in any way you like), 
and define $K(E) = G(E \setminus x_o)$. This is embedded 
in $G(E)$.  Let $X(E) = G(E)/K(E)$, and let 
${\mathscr G}(E) = \{ [gx_o] | g \in G(E) \}$, where 
$[gx_o] = g K(E) \in G(E)/K(E)$. Let $\Omega(E) = \Pr({\mathscr G}(E))$. 
Then $E \mapsto (\mathscr{G}(E),\Omega(E))$ defines a 
probabilistic theory in which every model is highly symmetric: 
$G(E)$ acts transitively on the set $X(E)$ of outcomes, in such a way as 
to act transitively also on the set $\G(E)$ of tests. Moreover, 
the stabilizer of any test acts transitively on the outcome-set of that 
test. \end{example}

\begin{example}[Example of Previous Example] Let $G(E) = \Unitary(\C^{E})$, with the obvious behavior on mappings. The above construction returns a version of finite-dimensional quantum theory in which every quantum system has a preferred (say, computational) basis. \end{example} 
}

{%\blue 
{\bf Process Tomography} We've required that a probabilistic theory be injective on objects, so that distinct systems have 
distinct models.  However, we've not imposed any corresponding injectivity condition on morphisms.  This allows for a situation in which two physically distinct processes $f, g : A \rightarrow B$  in $\Cat$ may give rise to probabilistically (or operationally) identical  morphisms $F(f) = F(g)$ between 
$F(A)$ and $F(B)$.

\begin{definition}\label{def: process-tomography} A probabilistic theory $F$ is {\em process-tomographic} iff $F$ is injective on morphisms. 
\end{definition} 

Unitary complex QM is clearly process tomographic. 
A prime example of a non-process tomographic theory is 
projective unitary QM. With notation as in Example 
\ref{ex: Mackey functor}, if 
$U : \H \rightarrow \K$ is an isometry and $\phi_{U}$ 
is given by  
$\phi_{U} : p \mapsto UpU^{\ast}$ for each rank-one projection 
$p$, then for any complex number $z \in \C$ with $|z| = 1$, we have $\phi_{zU} = \phi_{U}$: our functor is not injective 
on morphisms. }

{\gray {\em Remarks:} When a probabilistic theory 
$F : \Cat \rightarrow \Prob$ is process-tomographic, 
$\Cat$ is isomorphic to the image category $F(\Cat)$, and hence there is no harm in taking the theory 
to {\em be} the latter: a category of probabilistic models and morphisms between these. 
When $F : \Cat \rightarrow \Prob$ is not process-tomographic, it may be helpful to think in terms of the slightly more concrete category 
$\Cat_{F}$ having objects $F(A)$ where $A \in \Cat$ 
and morphisms $f \in \Cat(A,B)$. This will 
give us a picture of our theory as one in which 
objects "are" probabilistic models, but morphisms 
are not {\em just} morphisms of these models. }

{\bf Monoidal Probabilistic Theories}   As noted earlier the large ambient category $\Prob$ carries no universally serviceable monoidal structure. Rather, the monoidal structures that 
arise in practice are theory-specific.  If $\Cat$ is a 
symmetric monoidal category and $F : \Cat \rightarrow \Prob$ is a probabilistic theory, 
then as indicated above, we can use the fact that $F$ is injective on objects to carry the monoidal product on $\Cat$ 
across to $F(\Cat)$, simply defining, for any $A, B \in \Cat$, 
\begin{equation} 
F(A)F(B) := F(A \otimes B)
\end{equation} 
If $F$ is process-tomographic, that is, injective on morphisms as well, 
 we can also define $F(f)\otimes F(g) = F(f \otimes g)$; but in general, this will not be well-defined.  Thus, we need to add a condition, namely, that 
 if $f, f' \in \Cat(A,C)$ and $g, g' \in \Cat(C,D)$, 
\begin{equation} 
F(f) = F(f'), F(g) = F(g') \ \Rightarrow \ 
F(f \otimes g) = F(f' \otimes g')
\end{equation} 
When this holds, we say that $F$ is {\em monoidal}. 
Equations (10) and (11) then define a symmetric monoidal 
product on $F(\Cat)$, making $F$ a strict monoidal functor, as you can check. 

\begin{exercise} Do in fact check this. \end{exercise} 

\begin{exercise} Show that projective unitary real QM is 
monoidal. \end{exercise} 

Of course, we want a bit more: we want 
$F(A) F(B)$ to be a {\em non-signaling composite} of $F(A)$ and $F(B)$. That is, for every pair of objects $A, B \in \Cat$, 
we want a morphism $\pi_{A,B} : F(A) \times_{N} F(B) \rightarrow F(AB)$ satisfying Definition \ref{def: composites} 
But we can ask for even a bit more than this, and, in so doing, simplify things. Both 
$(A,B) \mapsto F(A) \times_{NS} F(B)$ and 
$(A,B) \mapsto F(A \otimes B)$ are actually 
bifunctors on $\Cat \times \Cat$ --- 
indeed, the former is exactly $\times_{N} \circ (F \times F)$ 
and the latter, $F \circ \otimes$. 

\begin{exercise}  Check that $\times_{NS}$ really 
is a bifunctor on $\Prob$.
\end{exercise} 
\vspace{-.1in} 

This observation makes it natural to adopt the following

\begin{definition}\label{def: NS theory} A {\em  Non-Signaling} probabilistic theory based 
on a SMC $\Cat$ is a pair $(F,\pi)$ where $F : \Cat \rightarrow \Prob$ is a monoidal probabilistic theory  and $\pi$ is a natural transformation 
\[ \times_{NS} \circ (F \times F)  \longrightarrow F \circ \otimes\]
such that $(F(A)F(B),\pi_{A,B})$ is a NS composite of 
$F(A)$ and $F(B)$, for all objects $A, B$ in $\Cat$. 
\end{definition} 

Note that the last line in the definition is required 
{\em only} to guarantee that product states are 
implementable in $F(A)F(B)$: the remaining conditions in the definition of a 
non-signaling composite follow automatically. 

\begin{exercise} Show that each of the theories described in Examples \ref{ex: Mackey functor}, \ref{ex: quantization functor}, and \ref{ex: line example} is monoidal, and describe its monoidal product. 
\end{exercise} 

%That is, $\pi$ takes the bi-functor 
%$F( - ) \times_{NS} F( - )$ to the bifunctor 
%$F( - \otimes - )$. 

%Show this makes $F(A \otimes B)$ a NS composite. Well, will be a composite IF have ran $\pi^{\ast}(\Omega) \supseteq \Omega_{\sep}$. This can just be requirement. 

{\em Remark:} As noted earlier, the terms "state" and "effect" for morphism of the form $I \rightarrow A$ and $A \rightarrow I$, respectively, in a SMC, do not always agree 
with the notions of state and effect as we've defined them. 
To illustrate this, consider the SMC $(\Prob, \times_{NS})$: 
the tensor unit is the trivial model $I$ with $\M(I) = \{\{\bullet\}\}$ and $\Omega(I) = \{1\}$, where $1$ is mapping 
$\bullet \mapsto 1 \in \R$. For an arbitrary model $A$, 
there are lots of morphisms $\phi : I \rightarrow A$, but 
these amount to selections of outcomes $x = \phi(\bullet)$, 
not to states. And, in general, there are no morphisms 
$A \rightarrow I$ at all.  

%Linearization helps to address this issue. However, 
%since the functor $\V$ is contravariant and $\E$, covariant, 
%we find that if $\phi : I \rightarrow A$ in $\Cat$, then we 
%have $F(\phi) : \V(There are two issues here.  One is that, in most concrete probabilistic theories, states and effects are simply not processes (morphisms of models). The other is that, the functor $\V$ is contravariant, 

%Asking that processes 
%$I \rightarrow A$ and $A \rightarrow I$ be taken to 
%states and effects in the GPT sense is an additional 
%constraint on a probabilistic theory. 

\tempout{If we want $F(A)F(B)$ to be a strong composite of 
$A$ and $B$, then for every $A, B \in \Cat$, we need 
a bi-affine mapping 
\[m_{A,B} : \Omega(F(A)) \times \Omega(F(B)) \rightarrow \Omega(F(A \otimes B))\]
This should be the dual of a morphism 
\[F(A \otimes B) \rightarrow A \times_{\sep} B\]
Hm. So: we want ...
a canonical choice of $\alpha \otimes \beta$ --- the stronger notion --- then we need more. Idea:  Add a morphism 
\[m^{\ast} : AB \rightarrow \D(A \times_{S} B)\]
where $A \times_{S} B$ has test space $\M(A) \times \M(B)$, but states restricted to separable ones. Here, $\D(A)$ is 
$\D(\Omega(A))$, test space of partitions of unit in 
$[0,u_{A}]$. 
}

\subsection{Other Frameworks}

%%The simplest examples of probabilistic theories are those that are simply sub-categories of $\Prob$. 
Let's now take a look at several well-known frameworks for GPTs: the very simple one known variously as {\em Boxworld} and the {\em device-independent} framework, the approach based on taking a compact convex set as an abstract state space, the {\em circuit framework} due to Hardy \cite{Hardy-Reconstructing}, and the closely related framework of {\em operational probabilistic theories} as developed by Chiribella, D'Ariano and Perinotti in \cite{CDP-Purification, CDP}, %\cite{CDP2010, CDP}. 

{\bf Boxworld} 
The best-known example of a "post-quantum" (non-classical but also non-quantum) GPT considers agents 
--- Alice, Bob, Clovis, {\em et alia} ---  each equipped with a black box having a display --- say, a pair of lights, one red and one green --- 
and a switch with two settings, plus a start/reset button. When the button is pressed, one and only one of the lights flashes. 
We can understand this as a test space: the switch can be in one of two positions, say up or down.  This gives us a 
semi-classical test space 
\[\B = \{E_{u}, E_{d}\}\]
where $E_{u} = \{(r,u), (g,u)\}$ and $E_{d} = \{(r,d), (g,d)\}$ 
Letting $\Omega = \Pr(\M)$ (geometrically, a square), 
we have a probabilistic model. 

It's sometimes helpful to encode both the outcomes and the switch as bits, writing, e.g., $(0|0)$ for the outcome of seeing the red light when the setting is down, $(0|1)$ for the red light when the setting is up, and so on. Then our test space has outcomes  
$X = \{(i|j) | i, j \in \{0,1\}\}$, and 
tests $E_{0} := \{(i|0) | i \in \{0,1\}\}$ and $E_{1} = \{(i|1) | i \in \{0,1\}\}$. Each of $E_0$ and $E_1$ can be regarded 
as a classical one-bit measurement, this two-bit example is arguably the simplest imaginable non-classical generalization of a 
classical bit. Accordingly, this model --- or more generally, any model of this form (two tests, two-outcomes each, with 
all probability weights allowed) is usually called a {\em gbit}. 

We can combine two gbits $\B_1$ and $\B_2$ --- both written in binary form, as above --- as follows: construct $\B_1 \times \B_2$, and re-arrange and regroup the entries of each outcome according the scheme 
\[((i|j),(k|l)) \mapsto (i,k | j, l)\]
We'll write $E_{j,l}$ for the image of $E_{j} \times E_l$ under this scheme, i.e., 
$E_{j,l} := \{ (i,k|j,l) | i,k \in \{0,1\}\}$.  Let $\B_1 \tensor \B_2 = \{ E_{j,l} | j, l \in \{0,1\}\}$. Thus, we have 
four tests, each with four outcomes. Note that this is just an isomorphic copy of $\B_1 \times \B_2$, and the need to reshuffle the indices is just a consequence of the way we've decided to label outcomes. In particular, $\B_1 \otimes \B_2$ is still semi-classical.

Continuing in this way, we can build up larger test spaces of the form 
\[\M = \B_1 \otimes \cdots \otimes \B_n\]
($n$ times) with $2^{n}$ tests, each with $2^{n}$ outcomes.  We 
form models by allowing all probability weights to count as states, and finally, a category by allowing all possible morphisms between models of this form to count as processes.  This is $\Boxworld$.  If we linearize, we find that, for models $A$ and $B$ in $\Boxworld$, 
\[\V(A \otimes B) = \V(A) \maxtensor \V(B)\]
for all $A, B \in \Boxworld$, and hence, $\V^{\ast}(AB) = \V^{\ast}(A)\mintensor \V^{\ast}(B)$. In other words, 
Boxworld is locally tomographic, and allows arbitrarily strong correlations, but permits no entangled effects. It follows \cite{GMCD} that it is not possible to carry out protocols like teleportation and entanglement-swapping in this theory. 

\tempout
{\gray {\bf Semi-classical Models}  The test spaces arising in Boxworld are all {\em semi-classical}, meaning that 
distinct tests don't overlap. This has a strong corollary. 

\begin{lemma} Let $\M$ be semi-classical. Then every probability weight on $\M$ is a weighted (integral) average of dispersion-free states.
\end{lemma} 

This is easy to see if $\M$ consists of finitely many finite tests, so that $X = \bigcup \M$ is finite (as in boxworld). In this case, the integral can be replaced by a sum, and every state is a convex combination of dispersion-free states. For  a proof in the general case, see [ref]. 

Semi-classical test spaces, and models based on these, while in some respects very close to being classical, are still pretty interesting.  %[Discuss]
}

{\bf Convex Operational Theories}   
As we've already discussed, any order-unit space $(A,u)$ is associated with a probabilistic 
model: one takes $\M(A,u)$ to be the collection of all sets $E \subseteq (0,u]$ with $\sum_{a \in E} a = u$; 
states are restrictions to $(0,u]  = \bigcup\M(A,u)$ of positive linear functionals $f \in A^{\ast}$ with 
$f(u) = 1$.   Let $\OUS$ stand for the category of order-unit spaces and positive linear mappings 
$\phi : A \rightarrow B$ with $\phi(u_{A}) \leq u_{B}$. Such a mapping restricts (and co-restricts) to a 
mapping $(0,u_{A}] \rightarrow (0,u_B]$, and it's easy to check that this is a morphism of models.  
In other words, we've constructed a functor $\M  : \OUS \rightarrow \Prob$. 

A bit more generally, if $F : \Cat \rightarrow \OUS$ is any (covariant) functor, we obtain a probabilistic theory 
$\M \circ F$. 

An important special case: let $\Cat$ be the category of compact convex sets. For any such set $K$, 
the space $\Aff(K)$ of bounded affine functionals $f : K \rightarrow \R$ is an order-unit space, with order 
taken pointwise on $K$ and the order unit the constant functional with value $1$. Any bounded 
affine mapping $\phi : K \rightarrow K'$ defines a bounded linear mapping $\phi^{\ast} : \Aff(K') \rightarrow \Aff(K)$ 
namely $a' \mapsto a' \circ \phi$.  This gives us a {\em contravariant} functor, which I'll just call $\Aff$, from 
the cagegory $\Conv$ of compact convex sets and affine mappings, to $\OUS$. Composing this 
with $\M$ above gives us a contravariant functor $\Conv \rightarrow \Prob$.   If we 
have any SMC $\Cat$ and a {\em contravariant} functor $F : \Cat \rightarrow \Conv$, we obtain 
a covariant probabilistic theory $\M \circ \Aff \circ F : \Cat \rightarrow \Prob$.

\tempout{\blue We also have a natural (covariant) functor $\Prob \rightarrow \OUS$ taking $A = (\M,\Omega)$ to $\V^{\ast}(A)$.  
If $\phi : A \rightarrow \M(E,u)$ is a morphism of models, then 
$\phi : X(A) \rightarrow E$ is a vector-valued state on $A$, and thus 
gives us $\phi^{\ast} : E^{\ast} \rightarrow \V(\Pr(\M(A)))$. Since $\phi$ is 
a model-morphism, however, $\phi^{\ast} : E^{\ast} \rightarrow \V(A)$. Dualizing again, 
we have $\phi^{\ast \ast} : \V(A)^{\ast} \rightarrow E^{\ast \ast}$. 
Is the image weak-$\ast$ continuous? $\alpha_n (a) \rightarrow \alpha(a)$ for 
$\alpha_n, \alpha \in E^{\ast}$ and $a \in E$ gives us 
\[\phi^{\ast \ast}(f)(\alpha_n) =  f(\phi^{\ast}(\alpha_n)) = 
f(\alpha_n \circ \phi^{\ast})\]
where $f \in \V(A)^{\ast}$. 

Combining this with $\M$, we have an endofunctor 
$\Prob \rightarrow \Prob$ taking $A$ to $(\M(\V(A)^{\ast}, u_A), \Omega(\V(A)^{\ast},u))$.   }

%[Discuss: adjoint?] }

%{\blue 
%Monoidality: Can use max or min tensor product; etc.  }
%\end{document} 
{\gray {\bf Convex Operational Theories from SMCs} Here is a useful special case of the above construction.  
As discussed earlier, in any SMC category $\Cat$, the 
monoid $S := \Cat(I,I)$ of scalars is is commutative.  
Suppose now that we are given a monoid homomorphism 
$p : S \rightarrow [0,1]$, where we regard $[0,1]$ as a monoid under multiplication. That is, 
for all scalars $s, t$,  $p(s \circ t) = p(s)p(t)$.  Given this one piece of data, we can now construct 
an entire probabilistic theory based on $\Cat$, as follows: given any object $A$ and any 
$\alpha \in \Cat(I,A)$ and $a \in \Cat(A,I)$, let 
$\hat{a}(\alpha) := p(a \circ \alpha)$. 
For each $\alpha$, then, we have a mapping $\hat{\alpha} \in [0,1]^{\Cat(A,I)}$. 
Let $\Omega(A)$ be the closed convex hull of these mappings $\hat{\alpha}$:  
Note that this is compact, since $[0,1]^{\Cat(A,I)}$ is compact by Tychonoff's Theorem. Each $a \in \Cat(A,I)$ defines an effect (a bounded affine functional on $\Omega(A)$) by evaluation: $\hat{a}(\hat{\alpha}) = \hat{\alpha}(a)$.  %{\blue [Functoriality!]} 
One can now proceed as above to obtain a probilistic theory.  In particular, one can show that this is monoidal. %{\blue [Check!]}. 
However, whether it is a non-signaling theory, seems to be a bit delicate. This is true if the resulting monoidal theory is locally tomographic.  See \cite{AW-shortcut} for some details.\\
}

%\end{document} 

%\newpage
{\bf Operational Theories %\emph{nei modi di Pavia}} %
\emph{\`{a} la Pavia}} 
The framework developed by the Pavia school (G. M. D'Ariano and his students, G. Chiribella and P. Perinotti) around 2010 has been particularly popular and influential.\footnote{A similar {\em circuit framework} was proposed by Lucien Hardy \cite{Hardy-Foilable, Hardy-Reconstructing} at about the same time. I limit the discussion here to the Pavia approach, with which I am more familar.} This begins with a notion of {\em test} that is similar to the one we've been using, but with an added bit of structure: first, each test has an {\em input} and {\em output} system. Secondly, they distinguish between 
the {\em outcomes} of the test and the physical event corresponding to it. Pavia represent such a thing with a diagram like this: 
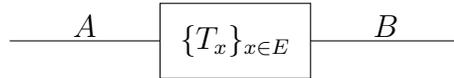
\[
\begin{tikzpicture} 
\draw (0,1) -- (2,1)    (4,1) -- (6,1) ; 
\draw (1,1.2) node{$A$}   (3,1) node{$\{T_x\}_{x \in E}$}   (5,1.2) node{$B$};
\draw (2,.5) rectangle (4,1.5);
\end{tikzpicture} 
\]
Here, $A$ and $B$ are the input and output system, respectively, $E$ is a test in our sense, and for each outcome $x \in E$, $T_{x}$ is the corresponding physical event or process.  The idea is that when the test is performed and outcome $x$ is secured, the experimenter knows that the process $T_x$ has taken place.  Regarding all of this, they say 

\begin{quote} Each test represents one use of a physical
device, like a Stern-Gerlach magnet, a beamsplitter, or
a photon counter. [...] When the physical device is used, it
produces an outcome ... e.g. the outcome
could be a sequence of digits appearing on a display, a
light, or a sound emitted by the device. The outcome
produced by the device heralds the fact that some event
has occurred.
\end{quote} 
They also note that the input and output labels essentially serve to dictate which tests can be composed sequentially. Regarding this, they posit that a test $\{T_{x}\}$ with output space $B$ and a test $\{S_{y} \}$ with input space $B$ {\em can} 
be composed (in that order) to yield a test that they write as $\{T_{x} \circ T_{y}\}$.  Tests are also allowed to compose 
in parallel: given any tests $\{T_{x}\}$ from $A$ to $B$ and $\{T_{y}\}$ from $C$ to $D$, there is a test $\{T_{x} \otimes T_{y}\}$ from a {\em composite system} $AC$ to a composite system $BD$. They proceed to enforce enough structure 
on this machinery to guarantee that the set of systems and "events" between them form a strict symmetric monoidal category. 
In particular, there is a unique {\em trivial system} $I$ such that $IA = AI = I$. 
Finally, they call tests of the form $I \rightarrow \{T_x\} \rightarrow A$, {\em preparation tests} and tests of the 
form $A \rightarrow \{R_{x}\} \rightarrow I$, {\em observation tests}. Composing these will give a test 
$I \rightarrow \{S_{y} \circ T_{x}\} \rightarrow I$. Each $S_{y} \circ T_{x}$ is then an "event" from 
$I$ to $I$ --- a {\em scalar}, in the SMC jargon.  CDP assume that here, 
each scalar {\em is} a probability, i.e., every process or event of type $I \rightarrow I$ 
{\em is} a real number in $[0,1]$, and also assume that $p \otimes q = pq$ for any two such events.  
A {\em circuit} is a collection of tests that compose, sequentially or in parallel, to yield a test with input and output $I$: 
any such circuit now has a defined probability. 

Now let's see if we can paraphrase, and perhaps also slightly generalize, this set-up in our language. Effectively, the Pavia school  %CDP 
%Pavia have 
has a symmetric monoidal category $\Cat$ (which they take to be strict, but let's not), along with an assignment of 
a test space $\M(A,B)$ to every pair of objects $A, B \in \Cat$.  We are also given a mapping $X(A,B) := \bigcup \M(A,B) \rightarrow \Cat(A,B)$ assigning a process $T_x : A \rightarrow B$ to every outcome $x \in X(A,B)$.  
%The rough idea is that if a test $E \in \M(A,B)$ is performed and an outcome 
%$x \in E$ is obtained, this {\em means} that the process (or, in %their terms, "event") $T_{x}$ has taken place.
 It's also required that we have two test-space morphisms  
\[\M(A,B) \times X(B,C) \rightarrow \M(A,C) \ \ \ (x,y) \mapsto xy\]
and 
\[\M(A,B) \times \M(C,D) \rightarrow \M(A \otimes C, B \otimes D) \ \ \ (x,y) \mapsto x \otimes y\]
such that 
 \[T_{x,y} = T_{y} \circ T_{x} \ \ \mbox{and} \ \ T_{x \otimes y} = T_{x} \otimes T_{y}. \]
Finally, we want a mapping $p : \Cat(I,I) \rightarrow [0,1]$ such that 
\vspace{.1in}
\begin{mlist} 
\item[(i)] $p(s \circ t) = p(s)p(t)$ for all $\alpha, \beta \in \Cat(I,I)$, and $p(\id_{I}) = 1$;
\item[(ii)] $s \mapsto p(s)$ is a probability weight 
on $\M(I,I)$. 
\end{mlist} 

A triple $(\Cat, \M, p)$ satisfying these 
requirements is a very general schema for an operational 
probabilistic theory, in something close to Pavia's sense. 
But it's reasonable to impose some further restrictions. 

%Let's agree to write $(\Cat, \M, p)$ for an operational %probabilistic theory of this type. To distinguish 
%such a thing from an OPT as presented in \cite{CDP}, we'll refer 
%to it as a {\em Pavian theory}. 

%\footnote{\gray It is always true in a SMC that $s \circ t = s %\otimes t$ for scalars $s, t \in \Cat(I,I)$.} 

For one thing, it's reasonable to take the mapping $X(A,B) \mapsto \Cat(A,B)$ to be surjective, on the argument that if there are 
processes in $\Cat$ that correspond to no outcome at all, these are in some sense unobservable, and can be elided. 
The requirement that $T_{x,y} = T_{y} \circ T_{x}$ makes the set of "outcomed" processes  closed under 
composition, so we still end up with a perfectly good category after such an elision.  Hence, we'll assume going forward 
that $x \mapsto T_{x}$ is surjective. 

This leaves open the possibility that the set $X(A,B) := \bigcup \M(A,B)$ may be quite a bit larger than $\Cat(A,B)$: there may, in other words, be many different outcomes $x$ that map to the same test $T_{x}$. However, in \cite{CDP} and elsewhere, it seems that the authors are assuming that the map $x \mapsto T_{x}$ is injective, in which case we may as 
well simply take $X(A,B)$ to {\em be} $\Cat(A,B)$.   We will save ourselves time, and also some trouble, if we adopt 
this point of view. So let's do that. \emph{\bf From now on, $\M(A,B)$ consists of sets of morphisms, and $\bigcup \M(A,B) = \Cat(A,B)$. } Accordingly, we'll suppress the mapping $T$, writing $x$ rather than $T_{x}$ for a generic 
morphism-{\em qua}-outcome in $\Cat(A,B)$.  Let's further simplify notation a bit further by writing $\M(A)$ for $\M(A,A)$ for all $A \in \Cat$.

We now want to ask: how is a Pavian theory $(\Cat, \M, p)$  a probabilistic theory in our sense? The answer is that if 
$\Cat$ is a symmetric monoidal category, so is $\Cat^{\op} \times \Cat$ (more on this shortly), and $\Cat( -, -)$ is a functor 
$\Cat^{\op} \times \Cat \rightarrow \Set$. Every pair of morphisms $u : A' \rightarrow A$ and 
$v : B \rightarrow B'$ define a morphism $(A,B) \rightarrow (A',B')$ in $\Cat^{\op} \times \Cat$, 
and these determine a morphism of test spaces $\M(A,B) \rightarrow \M(A',B')$ given by 
\[\phi(x) = v \circ x \circ u. \] 
So we can regard $\M$ as a functor from $\Cat^{\op} \times \Cat$ to test spaces and morphisms. To obtain a functor into $\Prob$, 
we need to assign a state-space to each pair $(A,B)$ in $\Cat^{\op} \times \Cat$.   There are natural candidates for 
states on $\M(A,B)$: for each $\alpha \in \Cat(I,A)$ and $b \in \Cat(B,I)$, we could consider 
\[p_{\alpha,b}(x) := p(b \circ x \circ \alpha).\]
This will assign a probability to each $x \in \Cat(A,B)$. However, in general these probabilities won't sum 
correctly --- that is, they won't generally sum to the same value --- over the various tests in $\M(A,B)$. 

\begin{definition}[\cite{CDP}] The theory $(\Cat, \M, p)$ is {\em causal} iff for every $\alpha : I \rightarrow A$ and every 
pair of tests $E, F \in \M(A,I)$, we have $\sum_{x \in E} p(x  \circ \alpha)   = \sum_{y \in F} p (y \circ \alpha)$. 
\end{definition} 

%Let us write $\|\alpha\|$ for the value of the sum in the above. If $\|\alpha\| = 1$, then we say that $\alpha$ is {\em normalized}.  If $\|\alpha\| = 0$ --- in other words, if $p(y \circ \alpha) = 0$ for all $y \in \Cat(A,I)$ ---  we say that $\alpha$ is null. 

\tempout{
The following is straightforward:

\begin{lemma} The following are equivalent: 
\begin{mlist} 
\item[(a)] $(\Cat, \M, p)$ is causal; 
\item[(b)] For all one-outcome tests $\{e\}, \{f\}$ in $\M(A,I)$, and for all $\alpha \in \M(I,A)$, 
$\pr(e \circ \alpha) = \pr(f \circ \alpha)$; 
\end{mlist} 
\end{lemma}
}

There is an important sufficient condition for $(\Cat, \M, p)$ to be causal. The assumptions made thus far tell us that $\M(A,B) \times \M(B,C) \subseteq \M(A,C)$, but it would be natural 
to allow branching sequential measurements as well (Pavia call these "conditional measurements"). In other words, 
we'd like to have $\for{\M(A,B)\M(B,C)} \subseteq \M(A,C)$.  Let us say that $(\Cat, \M)$ {\em allows branching measurements} 
when this is so for all $A,B,C \in \Cat$. 

The following is \cite[Lemma 7]{CDP-Purification},  
%appears as Lemma 7 in \cite{CDP-Purification} %{\blue [check ref!]}, 
but we can give a shorter proof. 
 
\begin{lemma}\label{lemma: branching implies causal} If $(\Cat, \M)$ allows branching measurements, it's causal.\end{lemma} 

\vspace{-.1in}
{\em Proof:} $\{\alpha\} \times E \sim \{\alpha\} \times F$ for any $E, F \in \M(A,I)$ and $\alpha \in \M(I,A)$, 
so as $s \mapsto p(s)$ is a probability weight on $\M(I,I)$ we have $\sum_{x \in E} p(x \circ \alpha) 
= \sum_{y \in F} p(y \circ \alpha)$. $\Box$ 

%{\blue [Why is $p$ a prob. weight on $\M(I)$?]}

{\bf \emph {From now on, let's assume that $(\Cat, \M, p)$ allows branching measurements (and hence, is causal).}}

\begin{definition} Let $\M$ be any test space. A one-outcome test $\{u\}$ is a {\em unit test}, and its single outcome 
is a {\em unit outcome}. Note that the probability of a unit outcome is $1$ for every state.\footnote{
%\cite{CDP2010, CDP} 
\cite{CDP} calls such a thing a {\em deterministic} "event", since it occurs with certainty. But the term "deterministic" has so many other connotations that I think it's best to avoid it here. Much earlier, Foulis and Randall called one-outcome tests "transformations", but this, too, is a freighted word.  I think "unit test" and "unit outcome" are preferable.} 
\end{definition}

In the context of a Pavian theory $(\Cat, \M, p)$, say that $u \in \Cat(A,B)$ is a unit iff $u$ is a unit outcome of $\M(A,B)$. 

{\bf{\em Further Assumption:}} \emph{\bf \emph{ In what follows, every test space $\M(A,B)$ contains at least one unit.}}
%[NB: this will also follow from closure %under coarse-graining!] 

\begin{lemma}  For every $\alpha \in \Cat(I,A)$ and every unit $u \in \Cat(B,I)$, $p_{\alpha,u}$ is a sub-normalized 
state on $\M(A,B)$.  Moreover, $p_{\alpha,u} = p_{\alpha,u'}$ for any two units $u, u' \in \Cat(B,I)$. 
\end{lemma} 

{\em Proof:} Let $E, F \in \M(A,B)$. We need to show that $p_{\alpha,u}$ sums to the same value over both.  
But $E \times \{u\}$ is a test in $\M(A,B) \times \M(A,I) \subseteq \M(A,I)$, so this follows from Lemma \ref{lemma: branching implies causal} 
and the definition of causality.  For the second claim, note that for every $x \in X(A,B) = \Cat(A,B)$, 
$x \circ \alpha \in \Cat(I,B)$, so $p(u \circ (x \circ \alpha)) = p(u' \circ (x \circ \alpha))$ by 
the definition of causality, and the fact that $\{u\}$ and $\{u'\}$ are tests. $\Box$ 

It follows that we can write $p_{\alpha,u}$ as $p_{\alpha}$, 
which we now do.

Let us say that $\alpha \in \Cat(I,A)$ is {\em null} iff $p(x \circ \alpha) = 0$ for all 
$x \in \Cat(A,I)$. Equivalently, $\alpha$ is null iff $p_{\alpha}$ is identically zero on $\Cat(A,B)$ for 
every $B$. Thus, if $\alpha$ is non-null, we can normalize $p_{\alpha}$ it by setting 
\[\hat{p}_{\alpha}(x) = \frac{p(x \circ \alpha)}{p(u \circ \alpha)}.\]
% More generally, $x \in \Cat(A,B)$ is null iff for all $\alpha \in \Cat(I,A)$ and all 
%$y  \in \Cat(B,I)$, $p(y \circ x \circ \alpha) = 0$.  
Write $\Cat_{+}(A,I)$ for the set of 
non-null processes in $\Cat(A,I)$. We now define, for every pair $(A,B)$, a state-space 
\[\Omega(A,B) = \{ \hat{p}_{\alpha} | \alpha \in \Cat_{+}(I,A) \}\]
where $p_{\alpha}(x) = p(u \circ x \circ \alpha)$, $u$ any unit in $\Cat(B,I)$. 

At this point, subject to the conditions imposed above, we 
have a probabilistic theory, albeit not generally convex. 
Replacing $\Omega(A,B)$ if necessary by its convex hull, we can linearize as usual by applying the $\V$ and $\V^{\ast}$ functors, 
or the functors $\V$ and $\E$ if we prefer (the latter is the Pavia approach). 

In earlier chapters, we made it a standing assumption that 
probabilistic models have separating, positive sets of states. 
Imposing these conditions constrains the structure of the test spaces $\M(A,B)$, and of 
the SMC $\Cat$.  

\begin{lemma} Suppose $(\Cat, \M, p)$ is a causal Pavian theory 
in which, for all objects $A, B$, $\Omega(A,B)$ is positive 
and separating for $\M(A,B)$. Then 
\begin{mlist} 
\item[(a)] For every pair of objects $A, B$, $\M(A,B)$ 
contains a unique unit outcome.
\item[(b)] The semigroup of scalars of $\Cat$ has a unique idempotent (namely, the identity); 

\end{mlist} 
\end{lemma}

{\em Proof:} (a) Let $u, v \in \Cat(A,B)$ be units. Since 
the theory is causal, $p_{\alpha}(u) = p_{\alpha}(v)$ for 
all $\alpha$, whence, since $\Omega$ is separating, 
$u = v$.

(b) If $s^2 = s$ in $S := \Cat(I,I)$, then 
$p(s^2) = p(s)^2$, so either $p(s) = 0$ or $p(s) = 1$.  
If the former holds, then  for 
every $\alpha \in \Cat(I,I)$ we have $p(s \circ \alpha) 
= p(s) p(\alpha) = 0$, so $p_{\alpha}(s) = 0$. This is 
impossible if the set of states is positive. Hence, 
$p(s) = 1$. It follows that for every $\alpha \in \Cat(I,I)$, $p_{\alpha}(s) = p(s) p(\alpha) = p(\alpha) = p_{\alpha}(1)$, whence, by separation (or by part (a)), $s = 1$. $\Box$ 

To finish this story, we need to say something about monoidality. 
As mentioned above, if $\Cat$ is a symmetric monoidal category, then $\Cat^{\op} \times \Cat$ inherits this structure: define 
\[(A,B) \otimes (C,D) := (A \otimes C, B \otimes D)\]
and, for $a : A' \rightarrow A$, $b : B \rightarrow B'$, $c : C' \rightarrow C$ and $d : D \rightarrow D'$, let 
\[(a,b) \otimes (c,d) : (A \otimes C, B \otimes D) \rightarrow (A' \otimes C', B' \otimes D')\]
be given by 
\[(a,b) \otimes %\Hom
(c,d) = (a \otimes c, b \otimes d).\]

Scalars in $\Cat^{\op} \times \Cat$ are pairs $(s, t) \in \Cat^{\op}(I,I) \times \Cat(I,I)$, and these compose as 
\[(s,t) \circ (s',t') = (s's,tt') = (ss',tt') =  (s,t) \otimes (s',t')\]
(notice in the penultimate expression we use the fact that the monoid $\Cat(I,I)$ is commutative; see, e.g., 
\cite{HV} for this). 
We have a mapping $S(\Cat^{\op} \times \Cat) \rightarrow S(\Cat)$ given by $(s,t) \mapsto st$. Composing this with 
the given function $p : S(\Cat) \rightarrow [0,1]$, we have a canonical probability assignment 
for $\Cat^{\op} \times \Cat$.  

{\bf Conjecture:} $(\M(A \otimes B, C \otimes D), \Omega(A \otimes B, C \otimes D))$ is a (strong) non-signaling 
composite of $(\M(A,C), \Omega(A,C))$ and $(\M(B,D), \Omega(B,D))$. 

I will be surprised if this is not true, but I have not yet sat down to do the necessary book-keeping. The reader should feel free to give it a try -- and please let me know either way! 

%\end{document} 

\begin{appendix}

{\gray \section{State of the Ensembles}

Recall that if $K$ is a convex set, a finite {\em ensemble} for $K$ is a finitely-supported probability weight (or distribution) on $K$, which we can represent as a set of pairs 
$(t_i, \alpha_i)$, $i = 1,...,n$, where $t_i \in (0,1]$ with $\sum_{i} t_i = 1$ and $\alpha_i \in K$.  The set $\D(K)$ of all finite ensembles for $K$ is irredundant, and can therefore 
be understood as a test space with outcome-set $X(K) = \bigcup \D(K) = (0,1] \times K$.  We wish to prove that 
%\begin{theorem} 
the only probability weight on $\D(K)$ is the weight $\rho((t,\alpha)) = t$.
% \end{theorem} 

%{\em Proof:} 
Let $f : (0,1] \times K \rightarrow [0,1]$ be a probability weight on $\D(K)$. 
%Note we can interpret this also as a function $f : K \rightarrow %[0,1]^{(0,1]}$. 
%I claim that $f$ is independent of $\alpha \in K$. 
For any $t \in (0,1]$ and $\alpha \not = \beta$, we have
$\{(t,\alpha),(1-t,\beta)\}$ and $\{(1-t,\beta),(t,\beta)\} \in \D(K)$, so 
$(t,\alpha) \sim (t,\beta)$. Thus, $f(t,\alpha)$ is independent 
of $\alpha$, and we can write  $f(t) := f(t,\alpha)$ with $\alpha$ free to vary. 

Note that as $\{(1,\alpha)\} \in \D(K)$ for all $\alpha \in K$, $f(1) = 1$.  We can extend $f$ to $[0,1]$ by setting $f(0) = 0$.  Now note that if $0 < t \leq 1$, we have $\{(t,\alpha), (1 - t,\alpha)\} \in \D(K)$, so $f(1 - t) = 1 - f(t)$. The extension to $0$ above makes this work for all $t \in [0,1]$.  Moreover, if $r,s,t \in (0,1)$ are distinct and $t + s + r = 1$, then both $\{(t + s,\alpha), (r,\alpha)\}$ and $\{(t,\alpha), (s,\alpha), (r,\alpha)\} \in \D(K)$, so $f(t + s) + f(r) = 1 = f(t) + f(s) + f(r)$, so we have 
\[f(t + s) = f(t) + f(s)\]
for all distinct $s,t \in (0,1]$ with $s + t < 1$. But if $s + t = 1$, we already have this, and it's trivially true if $s$ or $t$ is $0$; so it works in general: $f$ is additive on distinct pairs of values in $[0,1]$. It follows that $f$ is increasing. 

Let $s, t \in (0,1]$ with $s + t = 1$. Then for 
any choice of $\alpha, \beta \in K$ with $\alpha \not = \beta$, we have 
\[f(s) + f(t) = f(s,\alpha) + f(t,\beta) = 1,\]
so $f(s) = 1 - f(t)$. If $n \in \N$, choose 
$n+1$ distinct elements $\alpha_1,...,\alpha_{n}, \beta \in K$, 
and form the ensemble $\{(\tfrac{1}{n}t,\alpha_1),...,(\tfrac{1}{n}t,\alpha_n), s \beta\}$: we have 
\[n f(\tfrac{1}{n}) = f(\tfrac{1}{n},\alpha_1) + \cdots + 
f(\tfrac{1}{n}, \alpha_n) + f(s) = 1\]
so $nf(\tfrac{1}{n}) = 1 - f(s) = f(t)$, whence 
$f(\frac{1}{n} t) = \tfrac{1}{n} f(t)$. This extends to $t \in [0,1]$ since $f(0) = 0$. 
%We also 
%have $f(x) = f(\frac{1}{n}(x) + (1 - \frac{1}{n}x) = 
%\tfrac{1}{n}f(x) + f((1 - \tfrac{1}{n}x)$, so 
%$f((1 - \tfrac{1}{n})x) = (1 - \tfrac{1}{n})f(x)$. 

Suppose now that $k < n$: then for distinct 
$\alpha_1,...,\alpha_k, \beta$, we can construct 
ensembles 
$\{(\frac{1}{n},\alpha_1),...,(\frac{1}{n},\alpha_k), ((1 - \frac{k}{n}),\beta)\}$ 
and $\{(\frac{k}{n},\alpha), (1 - \frac{k}{n}, \beta)\}$: both belong to $\D(K)$, so 
\[f\left (\tfrac{k}{n} \right ) = k f \left (\tfrac{1}{n} \right ) = \tfrac{k}{n}.\]
So $f$ is the identity on rational points in $[0,1]$. Since $f$ is increasing, it follows that $f$ is the identity at all points of $[0,1]$. (For suppose $f(t) \not = t$ for some some $t \in [0,1]$. Choose a rational $q$ lying 
strictly between $t$ and $f(t)$: if $t < f(t)$, we have 
$f(t) \leq q$, a contradiction, and similarly if 
$f(t) < t$.) $\Box$ 

%then choose a rational point $t < q < f(t)$, and similarly 
%if $f(t) < t$.))
%the left-hand limit $f(x^{-})$ exists for every $x \in [0,1]$. We also have $f((1 - 1/n)x) = (1 - 1/n)f(x)$. The left-hand side approaches $f(x^{-})$, and the right-hand side approaches $f(x)$ as $n \rightarrow \infty$. Hence, $f$ is left-continuous at every $x \in [0,1]$, and the result follows. $\Box$ 

\tempout{
Now extend $f$ to $\R$ by setting $f(n + t) = n + f(t)$ for any $n \in \Z$ and any $t \in (0,1]$. Then $f$ is the identity on all rational points. Since $f$ is increasing, it's got at most countably many jump discontinuities. Suppose $x$ is one such.  We have 
\[f((1 - 1/n)x) = (1 - 1/n)f(x) \leq f(x)\]
The left-hand side approaches $\lim_{t \rightarrow x^{-}} f(t)$, while the right hand side approaches $f(x)$. Since 
the two sides are equal, so are the two limits, a contradiction.  Conclusion: $f$ is left-continuous.  Hence, $f$ is the identity.  
We have proved the following
}

{\em Remark:} We have not  used the convex structure of $K$ at all here. It could just be a set! So this is really a result about the free simplex on an infinite set $K$.  
%The point is that the mapping $\mu \mapsto \hat{\mu} : \Delta(K) \rightarrow K$ taking a finitely supported probability weight on $K$ to its barycenter extends to a test-preserving morphism from $\D(\Delta(K))$ to $\Delta(K)$. Any probability weight on the latter pulls back to the unique probability weight on the former, so there is but one probability weight on $\Delta(K)$. 

\tempout{
We also seem to be proving this: if $f : [0,1]$ is $\Q$-linear and non-decreasing (more generally, monotone), it's continuous (Proof: Let $q_n, r_n$ be rationals decreasing, resp. increasing,  to $1$: then $f(q_n x) \rightarrow f(x)^{+}$ and 
$f(q_n x)  = q_n f(x) \rightarrow f(x)$, so $f(x) = f(x)^{+}$; similarly $f(r_n x) \rightarrow f(x)^{-}$ and 
$f(r_n x) = r_n f(x) \rightarrow f(x)$, so $f(x) = f(x)^{-}$. So $f(x)^{+}$ and $f(x)^{-}$ coincide, and 
$f$ is continuous at $x$.)
}
}

\section{Base-normed and order-unit spaces} 

This appendix collects some basic facts about base-normed and order-unit spaces. The book \cite{Alfsen} by Alfsen is a standard source for this material, but goes into far more detail than we need, and assumes far more background than most readers will have. 
% in these notes. 

%\newpage
{\bf Conebase Spaces} 
A {\em conebase} in a vector space $\V$ is a convex set $K \subseteq \V$ such that (i) $K$ spans $\V$, and (ii) $K$ 
is sparated from $0$ by a hyperplane; equivalently, there exists is a linear 
functional $u$ on $\V$, which we call the {\em unit functional}, with $K \subseteq u^{-1}(1)$. 
An example is the set of density operators, as embedded 
in the space of trace-class self-adjoint operators 
on a Hilbert space. In this case, the functional $u$ is 
the trace.

The {\em cone generated} by $K$ is $\V_{+} = \R_{+}K$. 
It is straightforward to show that this is indeed a 
convex, generating, pointed cone, so $(\V,\V_+)$ is an 
ordered vector space. Note that if $ta \in \V(K)_{+}$ with $t \geq 0, a \in K$, then $t = u(ta)$. Also note that 
every vector in $\V$ has a decomposition of the form 
$sa - tb$ where $s, t$ are non-negative reals and $a, b \in K$

%{\bf Definition:} 
\begin{definition} A {\em conebase space} is a pair $(\V,K)$ where $\V$ is an ordered vector space, $K$ is a conebase, and $\V_{+} = \R_{+} K$.  \end{definition} 

{\gray {\em Remark:} An equivalent notion %\cite{BW-Information} is 
is a pair $(\V,u)$ where $\V$ is an ordered vector space 
and $u$ is a positive linear functional with the property 
that $u(a) = 0$ implies $a = 0$ for $a \in \V_{+}$. 
This is called an {\em abstract state space} in \cite{BW-Information}. 
Clearly, every conebase space is associated with an abstract state space; conversely, given an abstract state space $(\V,u)$, let $K = u^{-1}(1) \cap \V_{+}$; then 
$(\V,K)$ is a conebase space. %\cite{BW-Info}
}

In what follows, $(\V,K)$ is a conebase space, and $u$ is 
its unit functional.  

\begin{proposition}\label{prop: extension} Let 
$\phi : K \rightarrow \W$ be an affine 
maping from $K$ into a vector space $\W$. Then $\phi$ has a unique extension to a linear mapping $\tilde{\phi} : \V(K) \rightarrow \W$. 
\end{proposition} 

{\em Proof:} The only candidate is 
\[\tilde{\phi}(sa - tb) = s\phi(a) - t \phi(b)\]
where $s, t \in \R_{+}$ and $a, b \in K$. To see that 
this is well-defined, let 
$sa - tb = s'a' - t'b' = v \in \V$, where 
$s,t,s', t' \in \R_{+}$ and $a,b,a',b' \in K$. Then 
\[sa + t'b = s'a' + tb =: w \in \V(K)_{+}.\]
We wish to show that 
\begin{equation} 
s\phi(a) + t'\phi(b) = s'\phi(a') + t\phi(b).
\end{equation} 
Applying the functional $u$, we have $s + t' = s' + t =: r \geq 0$. If $r = 0$, $s, t, s'$ and $t'$ are all zero, and 
there is nothing to prove. Otherwise, we have 
\[\frac{1}{r}w = \frac{s}{r} a + \frac{t'}{r}b' = \frac{s'}{r} a' + \frac{t}{r} b.\]
Since the combinations on the right are convex, they belong 
to $K$, and we can apply $\phi$ to obtain 
\[\frac{s}{r} \phi(a) + \frac{t'}{r} \phi(b') = 
\frac{s'}{r} \phi(a') + \frac{t}{r} \phi(b)\]
which immediately yields (\theequation). It is now 
straightforward to check that $\tilde{\phi}$ is linear. $\Box$

If $\W$ is an ordered vector space and $\phi : K \rightarrow \W$ 
is an affine mapping with $\phi(K) \subseteq \W_{+}$, then the 
linear mapping $\tilde{\phi}$ is positive. 
The following shows (what is more or less obvious) that the ordered linear structure of $(\V,K)$ is entirely determined by the convex structure of $K$. 

\begin{corollary}\label{cor: regular embedding} With notation as above, suppose $\W$ is 
an ordered vector space, that $\phi : K \simeq \phi(K)$ is injective, and that $\phi(K) \subseteq \W_{+}$ is a base for the 
cone of $\W_{+}$. Then $\tilde{\phi}$ is an order-isomorphism. 
\end{corollary} 

%\begin{exercise} Prove this. \end{exercise} 

{\em Proof:} By the remark above, $\tilde{\phi}$ is 
positive. It is surjective because $\phi(K)$ spans $\W$. 
To see that it's injective, suppose $\tilde{\phi}(sa - tb) = s\phi(a) - t \phi(b) = 0$. 
Let $w$ be the functional on $\W$ with $w(\phi(K)) = 1$: 
applying this, we see that $s - t = 0$, i.e., $s = t$. 
Now $\tilde{\phi}(sa - tb) = s(\phi(a) - \phi(b)) = 0$, 
so $\phi(a) - \phi(b) = 0$, so $\phi(a) = \phi(b)$. 
But $\phi$ is injective, so $a = b$. Finally, note 
that since $\tilde{\phi}^{-1}$ takes $\phi(K)$ to $K$, 
it is positive. $\Box$ 

More generally, if $\phi : K \rightarrow \W$ is an affine 
injection and $\phi(K)$ is a base for the cone 
$\W_{+} \cap \spn(\phi(K))$, then $\spn(\phi(K))$, ordered 
by this cone, is linearly and order isomorphic to $\V(K)$. 
%$\phi : K \rightarrow \U$ and $\psi : K \rightarrow \W$ are two embeddings of $K$ into ordered 
%vector spaces $\U$ and $\W$ such that $\phi(K)$ and 
%$\psi(K)$ are bases for the cones  then the spans of 
%$\phi(K)$ and $\psi(K)$ are isomorphic as conebase spaces, 
%and we can regard both as isomorphic copies of $\V$. 
We will now establish a canonical representation for 
$\V$. 

Let $\Aff(K)$ denote the space of affine functionals 
$f : K \rightarrow \R$, ordered pointwise on $K$.  
For any vector space $\V$, let $\V'$ denote its algebraic dual space. 

\begin{corollary}\label{cor: dual of V(K)} $\V(K)' \simeq \Aff(K)$. \end{corollary}

{\em Proof:}  By Proposition \ref{prop: extension} above, every $f \in \Aff(K)$ 
extends uniquely to a linear functional $\tilde{f} \in \V'$. 
Conversely, if $\phi \in \V(A)'$, then 
$\phi|_{K} \in \Aff(K)$. We have $\tilde{\phi|_K} = \phi$ 
and $\tilde{f}|_{K} = f$, so $f \mapsto \tilde{f}$ 
defines a linear isomorphism $\Aff(K) \rightarrow \V(A)'$. 
This is positive, since if $f \geq 0$ on $K$, then 
$\tilde{f}(ta) = tf(a) \geq 0$ for all $ta \in \V(A)_{+}$. 
The inverse mapping $\V(A)' \rightarrow \Aff(K)$ sending 
$\phi$ to $\phi|_{K}$ is clearly positive. $\Box$. 

We now have an embedding 
\[\V(K) \leq \V(K)" \simeq \Aff(K)'\]
taking $a \in K$ to $\hat{a} \in \Aff(K)'$, namely, 
$\hat{a}(f) = f(a)$ for all $f \in \Aff(K)$. This is a regular embedding: the constant function $1$ on $K$ 
defines a linear functional $u$ in $\Aff(K)"$ by 
$u(\phi) = \phi(1)$, and $u(a) = \hat{a}(1) = 1(a) = 1$ 
for all $a \in K$. Identifying $a$ with $\hat{a}$, we can 
tret $K$ as a subset of $\Aff(K)'$, and can now identify $\V$ with its span in $\Aff(K)'$.

{\bf Seminorms and Minkowski functionals} Our aim now is to 
put a norm on a conebase space $(\V,K)$. Before turning to 
this, it will be helpful to start with a bit of background on 
the way in which, by specifying a suitable convex neighborhood of the origin as a "unit ball", we can construct a seminorm 
on any vector space. 

Suppose $B$ is a convex subset of a vector space $\V$. For every real number $r \geq 0$, 
let $rB = \{ ra \ | \ a \in B\}$. We say that $B$ is {\em absorbing} iff $\bigcup_{r \geq 0} rB = \V$, and {\em balanced} iff $B = -B$. 
%$a \in B \Rightarrow -a \in B$. 
Note that then $0 \in B$, by convexity. If $0 \in B$, we say 
that $B$ is {\em radially bounded} iff for every 
$x \in B$, the set $\{ r \in \R | rx \in B\}$ is bounded, 
and {\em radially compact} iff this set is also closed (thus, compact). 

\begin{exercise} Show that if $B$ is convex and absorbing, 
it spans $\V$. \end{exercise} 

%In what follows, assume that $B \subseteq \V$ is convex, %balanced and absorbing. 
%by 
%taking 
%\[B := \con(K \cup -K)\]
%to be the unit ball.   

\begin{definition} Let $B \subseteq \V$ be absorbing. The {\em Minkowski functional} of $B$ 
is the mapping $\| ~\cdot~ \|_{B} : \V \rightarrow \R$ 
defined by 
\[\|a\|_{B} \ = \ \inf \{\ r \ | \ \exists a_o \in B \ \alpha = r a
\ \} = \inf \{ \ r \geq 0 \ | \ a  \in rB \ \}.\]
\end{definition} 

%\[\|\alpha\| = \sup\{ r \geq 0 | r\alpha \in B\} = \inf \{ r \geq 0 | \alpha \in rB\}.\]
%It is not hard to show that this is always a seminorm.

%Further, since $K$ is compact, $B$ is also compact %(Exercise!); and from this, it follows that $\| \cdot \|$ is a norm, called the {\em base norm} associated with $K$, and that $\V$ is complete (a Banach space) in this norm.    Here are the details. 
\vspace{.1in}

\begin{lemma}\label{lemma: Minkowski} Let $B$ be convex, balanced, and absorbing. 
%in $\V$. Then 
\begin{mlist} 
\item[(a)] 
$\| \cdot \|_{B}$ is a seminorm. 
\item[(b)] $\| \cdot \|_{B}$ is a norm iff $B$ is 
radially compact. 
\end{mlist} 
\end{lemma} 
%\vspace{-.1in} 

{\em Proof:} 

(a) Let $a, b \in \V$ and set $s = \|a\|, t = \|b\|$, so that 
$a = sa_o$ and $b = t b_o$ for $a_o, b_o \in B$.  Let 
\[c_o = \frac{s}{s + t} c_o + \frac{t}{s + t}c_o \in B\]
and note that $(s + t)c_o = a + b$. Hence, $\|a + b\| \leq s + t = \|a\| + \|b\|$. 
Since $a \in B$ iff $-a \in B$, we have $\|a\| = \|-a\|$, and clearly 
$\|ra\| = r\|a\|$ for any $r \geq 0$, $\| r a \| = |r| \|a\|$ for any $r \in \R$, 
and $\| \, \cdot \, \|$ is a seminorm. %$\Box$ 

%{\em Remark:} Notice that we only use three facts about $B$ in this argument: first, that it's {\em absorving}, meaning 
%that $\bigcup_{r \geq 0} r B = \V$, secondly, that it's %convex, and finally, that it's {\em balanced}, i.e., 
%$a \in B \Rightarrow -a \in B$. Thus, {\em any} 
%convex subset of a vector space $\V$ having these three 
%properties defines a seminorm $\| \, \cdot \, \|_{B}$ --- %also 
%called the {\em Minkowski functional} of $B$ ---  on $\V$ in %this way, and $B$ is then the unit ball with respect to 
%this seminorm. 

%\begin{lemma} 
%{\bf Lemma:} {\em If $B$ is radially compact,  $\| \, \cdot %\, \|$ is a norm.}
% \end{lemma} 

%\vspace{-.1in} 
%{\em Proof:} 
%\begin{proof} 
(b) Suppose $a \not = 0$ and $\|x\| = 0$. Then $a \in tB$ for any $t > 0$, whence (setting $r = 1/t)$, $ra \in B$ for any $r > 0$. Hence, the ray $\R_{+}a \subseteq B$, and 
thus $B$ is not radially compact.  I leave the converse 
as 
%\end{proof} 
%$\Box$ 

%\begin{exercise} 
{\bf Exercise:} Show that if $B$ is not radially compact, 
$\| \cdot \|_{B}$ is not a norm. $\Box$ 

\tempout{

\begin{lemma} $\| \cdot \|_{B}$ is a seminorm. \end{lemma} 
\vspace{-.1in} 

{\em Proof:} Let $a, b \in \V$ and set $s = \|x\|, t = \|y\|$, so that 
$a = sa_o$ and $b = t b_o$ for $a_o, b_o \in B$.  Let 
\[c_o = \frac{s}{s + t} c_o + \frac{t}{s + t}c_o \in B\]
and note that $(s + t)c_o = a + b$. Hence, $\|a + b\| \leq s + t = \|a\| + \|b\|$. 
Since $a \in B$ iff $-a \in B$, we have $\|a\| = \|-a\|$, and clearly 
$\|ra\| = r\|a\|$ for any $r \geq 0$, $\| r a \| = |r| \|a\|$ for any $r \in \R$, 
and $\| \, \cdot \, \|$ is a seminorm. $\Box$

%{\em Remark:} Notice that we only use three facts about $B$ in this argument: first, that it's {\em absorving}, meaning 
%that $\bigcup_{r \geq 0} r B = \V$, secondly, that it's %convex, and finally, that it's {\em balanced}, i.e., 
%$a \in B \Rightarrow -a \in B$. Thus, {\em any} 
%convex subset of a vector space $\V$ having these three 
%properties defines a seminorm $\| \, \cdot \, \|_{B}$ --- %also 
%called the {\em Minkowski functional} of $B$ ---  on $\V$ in %this way, and $B$ is then the unit ball with respect to 
%this seminorm. 

A convex set $B$ containing the origin is {\em radially compact} iff for every one-dimensional subspace $L$ of $\V$, $L \cap B$ is compact. 

\begin{lemma}\label{lemma: B radially compact implies norm} If $B$ is radially compact,  $\| \, \cdot \, \|$ is a norm. \end{lemma} 

\vspace{-.1in} 
{\em Proof:} 
%\begin{proof} 
Suppose $a \not = 0$ and $\|x\| = 0$. Then $a \in tB$ for any $t > 0$, whence (setting $r = 1/t)$, $ra \in B$ for any $r > 0$. Hence, the ray $\R_{+}a \subseteq B$, and 
thus $B$ is not radially compact. 
%\end{proof} 
$\Box$ 

\begin{exercise} Show that the converse to the above lemma also holds. 
\end{exercise} 
}

A topology $\tau$ on $\V$ is {\em linear} iff 
it renders addition and scalar multiplication continuous 
as mappings $\V \times \V \rightarrow \V$ and 
$\R \times \V \rightarrow \V$, respectively.  
It can be shown \cite[Theorem 7.3]{Kelly-Namioka} that 
every finite-dimensional subspace of $\V$ is closed in every 
linear topology on $\V$. 

%\begin{exercise} Let $\tau$ be a Hausdorff linear topology on %$\V$. Any one-dimensional linear subspace of $\V$ is closed. 
%\end{exercise} 

%Proof: Let $L$ be a one-dimensional subspace of 
% $\V$. For any non-zero $x \in L$, scalar multiplication gives us a continuous bijection $\R \times \{x\} \rightarrow L$. Since the relative product topology on $\R \times \{x\}$ is Hausdorff, this is a homeomorphism 

\begin{proposition}\label{prop: compact implies radially bounded} Suppose that $B$ is compact in 
some Hausdorff linear topology on $\V$. Then $B$ is radially compact and 
% $K$ is compact in $\R^{X}$, then $\V$ is 
$\V$ is complete in $\| \, \cdot \, \|_{B}$. 
\end{proposition} 

{\gray 
The proof is DIY, with hints:

\begin{exercise} Let $(X,d)$ be a metric space. For 
every $r > 0$ and $a \in X$, let $B_{r}(a) = \{ x \in X | d(x,a) \leq r\}$. This is the {\em closed ball} of radius $r$ 
at $a$. Suppose $\tau$ is a topology on $X$ in which every closed ball is compact.  Show that $(X,d)$ is complete. 
(Hints: Let $(x_n)$ be a Cauchy sequence in 
$X$. 
\begin{mlist} 
\item[(a)] Show that $(x_n)$ is bounded, hence, contained in some closed ball $B$. 
\item[(b)] Show that $(x_n)$ has a $\tau$-limit point $x$ in $B$. 
\item[(c)] Show that for every $\epsilon > 0$ there is a 
closed ball of radius $\delta < \epsilon/2$ and a natural 
number $N$ with 
$\overline{B}_{\delta}(x_N) \subseteq B$
\item[(c)] Conclude that the limit point $x$ also belongs 
to $B_{\delta}(x_N)$. Conclude that $d(x_n,x) < \epsilon$ for all $n \geq N$, hence, $(x_n)$ converges 
to $x$ in the metric $d$.) 
\end{mlist} 
\end{exercise} 

\begin{exercise} Show that if the unit ball in a normed space $\V$ is compact in some linear topology $\tau$, then $\V$ is complete. \end{exercise} 
}

{\bf Base-normed spaces} Suppose $(\V,K)$ is a conebase space. is a compact, convex subset of a topological vector space. Define 
\[B := \con(K \cup -K):\]
It's easy to see that $B$ is convex, balanced, and aborbing in $\V$, so we can define a seminorm $\| \, \cdot \, \|_{B}$ as above. If $B$ is radially compact, then its Minkowski functional defines a norm on $\V$, called the {\em base norm}, and in this case, we say that $(\V,K)$ is a {\em base-normed space}. We say that $\V$ is a {\em complete} base-normed space iff it is complete, i.e., a Banach space, in its base-norm.  Proposition B.5 supplies a useful sufficient condition for this:

\begin{lemma}\label{lemma: compact base} Let $(\V,K)$ be a conebase space and let 
$K$ be compact in some linear topology on $\V$. Then 
$\V$ is a complete base-normed space. 
\end{lemma} 

This becomes a corollary to Proposition \ref{prop: compact implies radially bounded}
once we establish that the compactness of $K$ 
implies that of $B$.  I'll leave this as 

\begin{exercise} Show that if $K$ is $\tau$-compact for a 
linear topology $\tau$ on $\V$, then so is $B = \con(K \cup -K)$. \end{exercise} 

Lemma \ref{lemma: compact base} covers a lot of ground. For instance, the state space of any 
von Neumann or Jordan model is compact, and thus, the corresponding space 
$\V(A)$ is complete. However, some of our most important infinite-dimensional 
state spaces are {\em not} compact. In particular, the set of density operators 
on a finite-dimensional Hilbert space, or, more generally, the space of normal 
states on a von Neumann algebra, are not compact in any useful topology. 
Luckily, a stronger completeness theorem is available that covers these cases. 
This is discussed in Appendix C.

Let $\phi : \V \rightarrow \W$, where $\W$ is a normed space, 
and suppose $\phi$ is bounded on $K$, say with $\sup_{a \in K} \|\phi(a)\| = M$. 
Any point $v \in B$ has the form 
$v = sa + (1-s)b \in B$ where $0 \leq s \leq 1$ 
and $a, b \in K$, so 
\[\|\phi(sa - (1 - s)b)\| \leq s\|\phi(a)\| + (1 - s) \|\phi(b)\| 
\leq M.\]
Thus, $\phi$ is bounded, with $\|\phi\| \leq M$. Since 
{\em a priori} $M \leq \|\phi\|$, we have $M = \|\phi\|$. 
%\[\|\phi(v)\| \leq M\]
%for all $v \in B$, whence, 
%$\|\phi\| \leq M$. 
%$\|\phi(a)\| > M - \epsilon$, then as $a \in B$, 
%$\|\phi\| > M - \epsilon$ as well.  As $\epsilon$ can be 
%chosen to be arbitrarily small, we have $M = \|\phi\|$. 
We now have the following bounded version of \ref{prop: extension}:

\begin{lemma}\label{lem: unique extension bis} If 
$(\V,K)$ is a conebase space and $\W$ is a normed vector 
space, any bounded affine mapping $\phi : K \rightarrow \W$ 
extends uniquely to a bounded linear mapping $\tilde{\phi} : \V \rightarrow \W$ with $\|\tilde{\phi}\| = \sup_{a \in K} \|\phi(a)\|$. 
\end{lemma} 

{\em Proof:} The existence of a unique linear extension 
was established in Proposition \ref{prop: extension}, so 
we need only show that this is bounded. But 
\[\sup_{v \in B} \|\tilde{\phi}(v)\|_{\W} 
= \sup_{a \in K} \|\tilde{\phi}(a)\|_{\W} = 
\sup_{a \in K} \|\phi(a)\|_{\W}\]
and by assumption, this last is finite. $\Box$ 

%{\gray {\em Remark:} More generally, a {\em base-normed space} is a pair $(\V,K)$ where $K$ is a conebase such that 
%$B = \con(K \cup -K)$ is radially compact. But the 
%case where $K$ is compact (and $\V$ therfore complete) is %sufficient for our purposes.}

\begin{exercise}  Boundedness of linear extensions. 
Banach dual of $\V$ is  $\simeq \Aff_{b}(K)$. Additivity of 
base-norms. 
\end{exercise} 

{\bf Order-Unit normed spaces} An {\em order unit} in an ordered vector space $\E$ is an element $u \in E^{+}$ with the property that, for every $a \in E$, $a \leq nu$ for some $n\in\mathbb{N}$. 

The unit effect $u$ in $\E(A)$, for any probabilistic model $A$, is clearly an order unit. If $\E$ is finite-dimensional, one can show that any $u$ belonging to the interior of the positive cone $\E_{+}$ is an order unit. This is also true for ordered Banach spaces with closed cones. But not every ordered vector space has 
an order unit. For example, if $X$ is infinite, $\R^{X}$, ordered pointwise, has no order unit. 

\begin{exercise} Show that the constant function $1$ is 
an order-unit for the space $B(X)$ of bounded linear functionals. Show that the identity operator on a Hilbert 
space $\H$ is an order unit for $\L_{\sa}(\H)$.
\end{exercise}

Suppose $(\E,u)$ is an order unit space. For any $a \in \E$, 
we have natural numbers $m,k$ with $a \leq  m u$ and $-a \leq ku$, whence, $-mu \leq a \leq nu$. Taking $n = \max(k,m)$, 
it follows that $-n u \leq a \leq n u$. Thus, $E \ = \ \bigcup_{n} ~[-nu,nu]$ --- in other words, the set 
$[-u, u ]$ is absorbing. It is clearly convex and balanced, 
so its Minkowski functional defines a seminorm $\| \cdot \|_{u}$. 

\begin{exercise} Show that $[-u, u]$ is radially compact 
iff $\E$ is Archimedean. 
\end{exercise}

%\footnote{The still more general statement is that $u$ is an order unit iff it belongs to the {\em algebraic interior} of the cone. See Wikipedia for more on this. } 

\begin{definition} An {\em order-unit space} is a pair $(\E,u)$ where $\E$ is an Archimedean ordered vector space and 
$u$ is an order unit. The norm $\| \cdot \|_u$ is the 
{\em order-unit norm} on $\E$. 
%with a closed [what top??] cone, and $u$ is an order-unit. 
\end{definition}

\begin{exercise} Show that the order-unit norm on 
$(B(X),1)$ is the usual supremum norm on $B(X)$. 
\end{exercise} 

A {\em state} on an OUS $(\E,u)$ is a positive linear 
functional $\alpha : \E \rightarrow \R$ with 
$\alpha(u) = 1$. One can show that such a function is automatically bounded with respect to the order-unit norm, so the set $K$ of states --- the {\em state-space} of $(\E,u)$ ---  
lies in $\E^{\ast}$, where it is a base for the dual cone 
$\E^{\ast}_{+}$. Moreover, $K$ is closed, and hence compact, 
in the weak-$\ast$ topology on $\E^{\ast}$.

% Moreover, it is not difficult 
%to show that the state space is closed in the weak-$\ast$ 
%topology on $\E^{\ast}$, and hence, compact, and 
%that every functional $f \in \E^{\ast}$ is a linear combination 
%of states. Thus, by \ref{lemma: compact base}, we 
%see that $(\E^{\ast},K)$ is a complete BNS. 
For details on this, and a proof of the following, see 
\cite{AS}:

\begin{theorem}[Ellis]\label{thm: Ellis} If 
$(\E,u)$ is an order-unit space with state-space $K$, 
then $(\E^{\ast},K)$ is a complete BNS, and the base 
norm is equivalent to the dual norm. If $(\V,K)$ is a base-normed space, then $(\V^{\ast},u)$ is a complete OUS, and the order-unit 
norm is equivalent to the dual norm. 
\end{theorem}

\section{Completeness results for $\V(A)$}

In the following, $A = (\M,\Omega)$ is a fixed probabilistic model with outcome-set $\bigcup \M = X$, event space $\Ev$, and with $\Omega$ convex. As usual, $\V(A)$ is the span of $\Omega$ in $\R^{X}$.  As discussed in Appendix B, if $\Omega$ is compact in some linear topology, then $\V(A)$ is compact in its base-norm. Unfortunately, many infinite-dimensional examples --- prominently including infinite Borel models and infinite-dimensional Hilbert models! --- have non-compact state spaces. Fortunately, more general completeness results are available that cover these and many other examples.  

%The results in this Appendix are mostly due to Tim Cook \cite{Cook}. 
%ur goal in the following is to prove that $\V(A)$ is a complete 
%BNS under conditions weaker than that $\Omega$ be compact. 

We start with the case in which $\Omega = \Pr(\M)$ (this covers both Borel and Hilbert models).  Here, the result we are after is due to Tim Cook \cite{Cook}. The main idea 
is to embedd $\V(A)$ into the following, {\em a priori larger}, Banach space. 

\begin{definition} Let $\W$ be the space of functions $X \rightarrow \R$ such that 
\[ \|\mu\|_{1} := \sup_{E} \sum_{x \in E} |\mu(x)| < \infty\]
and there exists a constant $K$ with $\sum_{x \in E} \mu(x) = K$ for all $E \in \M$. 
\end{definition}

It is straightforward that $\W$ is, indeed, a subspace of $\R^{X}$. We refer to $\|\mu\|_1$ the {\em variation norm} of $\mu \in \W$. That it {\em is} a norm is also straightforward to check. Note that $\sum_{x \in E} |\mu(x)| = \sup\{ \, \sum_{x \in a} \, |\mu(x)| \ | \ a \subseteq E , \ a \ \mbox{finite} \, \}$, so we can equally well define $\|\mu\|_{1}$ as $\sup_{a \in \Ev(A)_{o}} \sum_{x \in a} |\mu(x)|$ where $\Ev(A)_{o}$ is the set of finite events of $\M$.   We will show that $\W(A)$ is complete in $\|~\cdot~\|_1$, but first we consider two examples. 

%\end{appendix} 
%\end{document} 

{\em Remark:} Elements of $\W$ are called {\em bounded weights} on $\M$.  Normed spaces of bounded complex-valued weights or, more generally, of bounded weights with values in a normed vector space, can be defined in the same way. 

\begin{example} In the case in which $\M = \M(S,\Sigma)$ is the test space of finite or countable partitions of a measurable space $S$ by measurable subsets, $\W$ is the space of 
bounded countably additive signed measures on $S$. It is well known that this is complete in the variation norm  (cf, e.g., Dunford and Schwartz \cite{Dunford-Schwartz}, p. 161.)  
\end{example}

\begin{example} Let $X$ be the unit sphere of an infinite-dimensionional Hilbert space $\H$, and let $\M = {\F}(\H)$ be the corresponding frame manual.  It is a standard fact (see, e.g., \ref{Conway}) that the space $\V(\H)$ of self-adjoint trace-class operators is a Banach space in the trace norm, given by $\|T\|_{\tr} = \Tr(|T|)$, where $|T| = (T^{\ast} T)^{1/2}$.  Every such operator determines a weight on $\M$, call it $\mu_{T}$, given by $\mu_{T}(x) = \Tr(TP_{x})$. This gives us a linear mapping $T \mapsto \mu_{T}$ from $\V(\H)$ into $\W$; this mapping is clearly injective, and the Bunce-Maitland-Wright extension of Gleason's theorem shows it's surjective. It's also an isometry. 
Using the Spectral Theorem, it's not hard to see that the trace-norm of $T$ is the same as 
the norm $\|\mu_{T}\|_1$ introduced above. Indeed, if $|T| = \sum_{x \in E} t_{x} P_{x}$ 
$E$ is an ONB, $t_{x} > 0$, and $P_{x}$ is the rank one projection attached to the unit vector $x \in E$, we have 
\[\sum_{x \in E} |\Tr(|T|P_{x})| = \sum_{x \in E} |t_{x}| = \Tr(|T|),\]
so $\|T\|_{\tr} \leq \|T\|_1$. On the other hand, of $F$ is any other orthonormal 
basis, then since $\Tr(P_{x}P_{y}) = |\langle x, y \rangle|^2$, we have 
\begin{eqnarray*}
\sum_{y \in F} |\mu_{T}(y)| & = & \sum_{y \in F} \left | \sum_{x \in E} t_{x} \Tr(P_{x} P_{y}) \right |\\
&  \leq & 
\sum_{y \in F } \sum_{x \in E} |t_{x}| |\Tr(P_{x}P_{y})| \\
& \leq & 
\sum_{x \in E}  |t_{x}| \left (\sum_{y \in F}  \Tr(P_x P_y) \right ) \\ 
& \leq & 
\sum_{x \in E} |t_{x}| = \Tr(|T|).
\end{eqnarray*}
Thus, $\|\mu_T\|_{1} \leq \|T\|_{\tr}$.  $\Box$ 

%{\gray It remains to show that $T \mapsto \mu_T$ is surjective.  
%This isn't entirely trivial, 
%as it requires either proving that $\W$ is generated by its positive cone, and then 
%applying Gleason's Theorem, or else extending Gleason's Theorem directly to handle weights on the frame manual that are not necessarily linear combinations of probability 
%weights.  
%The Bunce-Maitland Wright \cite{B-MW} extension of Gleason's Theoremfor a more-than-sufficiently powerful result in this direction. 
%$\Box$ }
\end{example}

\tempout{
L. J. Bunce and J. D. M. Wright showed that, for any von Neumann algebra 
${\mathfrak A}$ with no type-$I_2$ factor --- notably, for ${\mathfrak A} = \B(\H)$ --- any 
finitely-orthogonally-additive mapping from projections into {\em any} Banach space 
extends uniquely to a bounded linear operator from ${\mathfrak A}$ to this Banach space. 
It is not hard to see that any weight $\mu \in \W(\H)$ satisfies has a well-defined 
value $\mu(a) := \sum_{x \in a} \mu(x)$ on every orthonormal set in $\H$, that 
is, on every event of $\F(\H)$, and that if $a \sim b$, then $\mu(a) = \mu(b)$. 
Thus, if $p$ is a projection, and $a$ is an orthonormal basis for the range of $p$, 
we can set $\mu(p) = \mu(a)$. The resulting function is additive on orthogonal 
projections, and thus, there is a bounded operator $\phi \in {\mathcal B}(\H)^{\ast}$ with 
$\mu(p) = \phi(p)$.  Since  $\mu$ is absolutely summable over maximal orthogonal sets, 
it's summable over countable pairwise disjoint sets of projections, and therefore $\phi$ 
shares this property. This means that $\phi$ is a normal functional [ref] whence, 
arises from a trace-class operator $T$ by $\phi(p) = \Tr(Tp)$ for all $p$. It 
follows that $\mu = \mu_T$. }

We are now going to show, following Cook, that $\W(A)$ is always complete in the variation norm. If $\mu \in \W$, then 
\[\mu(a) \ = \ \sum_{x \in a} \mu(x)\] 
is well-defined for all events $a$, with $|\mu(a)| \leq \|\mu\|_1$. Thus, we can define 
a linear functional $\hat{a} : \W \rightarrow \R$ by setting $\hat{a}(\mu) = \mu(a)$, 
and we see that $\hat{a}$ is bounded with $\|\hat{a}\| \leq 1$. Since these functionals 
separate points of $\W$, this gives us another norm on $\W$, namely,  
\[\|\mu\|_{\Ev} := \sup_{a \in \Ev} |\mu(a)|.\]
Note that convergence with respect to this norm is the same thing as uniform convergence on events.

%Since $\mu(E \setminus a) = \mu(F \setminus a)$ for all $a \in \Ev$ and all tests $E, F$ with $a \subseteq E \cap F$,  we can define $u = \hat{E}$ for all $E$ and $\hat{a}' = \hat{E \setminus a} = u - \hat{a}$. 

\begin{lemma} For every $\mu \in \W$, 
\[\|\mu\|_1 = \sup_{a \subseteq E \in \M} \mu(a) - \mu(E \setminus a).\]
%=  \sup_{a \in \Ev} (2\hat{a} - u)(\mu).\] 
Hence, $\| ~\cdot~ \|_{\Ev}$ is equivalent to $\|~\cdot~\|_1$. 
\end{lemma} 
%\end{document} 
 
Thus, $\mu$ has finite variation norm iff it is bounded on events, and a sequence $(\mu_n)$ 
in $\W$ converges to $\mu \in \W$  with respect to the norm $\| \cdot \|_1$ iff $\mu_n \rightarrow \mu$ uniformly on events. 
 
\begin{proof} For the moment, write $\|\mu\|$ for 
$\sup_{a \in E \in \M} \mu(a) - \mu(E \setminus a)$ (noting that this is 
necessarily non-negative).  Let $E \in \M$ and let $\mu \in \W$. Set 
$a = \{ x \in E | \mu(x) > 0\}$. Then 
\[\sum_{x \in E} |\mu(x)| = \mu(a) - \mu(E \setminus a) \leq \|\mu\|.\]
Taking the sup over $E \in \M$ gives us $\|\mu\|_1 \leq \|\mu\|$. 
On the other hand, 
\[|\mu(a) - \mu(E \setminus a)| \leq |\mu(a)| + |\mu(E \setminus a)| \leq \sum_{x \in E} |\mu(x)|.\]
Taking the supremum over all $a \subseteq E \in \M$ gives us 
$\|\mu\| \leq \|\mu\|_1$, proving the first statement. For the second statement, simply note that or any event $a$, we have 
\[|\mu(a)|  \leq |\mu(a)| + |\mu(E \setminus a)| \leq 2\|\mu\|_{\Ev},\] so 
$\|\mu\|_{\Ev} \leq \|\mu\| \leq 2\|\mu\|_{\Ev}$. Since $\|\mu\| = \|\mu_1\|$, 
we are done. $\Box$ 
%$\|\mu\|_{\Ev} \leq \|\mu\|_{1} \leq 2\|\mu\|_{\Ev}$.  $\Box$ 
%{\blue [Clarify $\| \cdot \|$ vs $\|\cdot\|_{\Ev}$!]}
\end{proof} 

For the sake of convenience, except in cases where there might be some ambiguity, from 
now on we'll write $\| \cdot \|$ for $\| \cdot \|_{1}$.

%\begin{lemma} $\|\mu\| = \sup_{a \in \Ev} ~|~ (2\hat{a} - u)(\mu)|$ for all $\mu \in \W$. (Note the absolute values!)
%\end{lemma} 
 
%{\em Proof:} Let $\mu(a) - \mu(E \setminus a) = k$. Then 
 %\[\|\mu\| \geq \mu(a) - \mu(E \setminus a) = k\]
%and also  
%\[\|\mu\| \geq \mu(E \setminus a) - \mu(a) = -k.\]
%Thus, $\|\mu\| \geq |\mu(a) - \mu(E \setminus a)|$, and the 
%inequality $\|\mu\| \geq \sup_{a} |(2a - u)(\mu)|$ follows. The reverse  inequality is trivial. $\Box$ 
 
%Thus, $\|~\cdot~|$ is the sup of $|\mu| \in \W"$ over the %functionals of the form $2\hat{a} - u$, $a \in \Ev(A)$. This is clearly a seminorm. 
 
% {\bf Lemma:} $\| ~ \|$ is a norm.  
%{\em Proof:} Suppose $\|\mu\| = 0$. If $a \subseteq E$ and $\mu(E) = c$, we have $\mu(a) = \mu(E \setminus a) = c - \mu(a)$ for every $a$, whence, $2\mu(a) = c$ for all $a$, whence, $2\mu(E) = c$, whence $2c = c$, whence $c = 0$. But now $\mu(a) = -\mu(a)$ for every event, and thus, $\mu = 0$. $\Box$ 

\tempout{
%\begin{lemma} For every $a \in \Ev$, $\hat{a}$ is bounded with $\|\hat{a}\| \leq 1$.
%\end{lemma} 

%{\em Proof:} For each $a \in \Ev$, let $\phi_{a} = 2\hat{a} - u$: then $|\phi(\mu)| \leq \|\mu\|$ by definition of $\| \cdot\|$, so $\phi_{a}$ is bounded with $\|\phi_{a}\| \leq 1$. 
%In particular, $u = \phi_{E}$ for any $E \in \M$.  Thus, 
%$\hat{a} = (\phi_{a} + u)/2$ is also bounded, with 
%$\|\hat{a}\| \leq \frac{1}{2}(\|\phi_{a}| + \|u\|) = 1$. $\Box$ 

%\begin{lemma} $\| ~\cdot~ \|$ is equivalent to the supremum norm 
%$\|\mu\|_{s} := \sup_{a \in \Ev} |\mu(a)|$. In particular, $\mu$ has finite variation norm iff it is bounded on events. 
%\end{lemma} 

{\em Proof:} For any event $a$, we have 
\[|\mu(a)\| \leq |\mu(a)| + |\mu(E \setminus a)| \leq 2\|\mu\|_{s},\] so 
$\|\mu\|_{s} \leq \|\mu\|_{1} \leq 2\|\mu\|_{s}$.  $\Box$ 
}

\begin{proposition}\label{prop: W complete} $(\W,\|\cdot\|_1)$ is complete.
\end{proposition}

{\em Proof:} Let $(\mu_n)$ be Cauchy with respect to $\|~ \cdot~\|_1$. Then 
$(\mu_n)$ is Cauchy with respect to the supremum norm on events, so 
$\mu_n \rightarrow \mu$ in the space $B(\Ev,\R)$ of bounded real-valued functions 
on $\Ev$. Since $\mu$ bounded on events, it is also bounded in the 1-norm. Let $\mu_n(E) = m_n$ for all $E \in \M$ (
since $\mu_n \in \W$, this is independent of $E$). Then $|m_n - m_k| 
\leq \|\mu_n - \mu_k\|_{\Ev}$; since the latter approaches $0$, $(m_n)$ is Cauchy, 
and thus, converges to some constant $m$, independent of $E$. It remains 
to show that, for each $E \in \M$, $\sum_{x \in E} \mu(x) = m$.  Let $\epsilon > 0$, and 
choose $n$ with $\|\mu - \mu_1\|_{s} < \epsilon$ and $|m_n - m| < \epsilon/3$. Also choose a finite event 
$a \subseteq E$ with $|\mu_n(a) - m_n| < \epsilon/3$. Then 
%\begin{eqnarray*}
%\left |\left (\sum_{x \in a} \mu(x)\right ) - m \right | & = & 
\[|\mu(a) - m| 
\ < \ |\mu(a) - \mu_n(a)| + |\mu_n(a) - m_n| + |m_n - m| < \epsilon.\]
%\end{eqnarray*}
Thus, $\sum_{x \in E} \mu(x) = m$, as promised. $\Box$

%\begin{exercise} Prove this. Hint: First show that if $\mu_n \rightarrow \mu$ uniformly on 
%events, with $\mu_n \in \W$, then $\mu \in \W$. 

%Cook gives a fairly direct proof. The following is less direct, but shorter. 

\tempout{
{\em Proof \#1 (Cookish):} Suppose $(\mu_n)$ is Cauchy in $\W$ w.r.t. this norm. Then for any $a \in \Ev(A)$, $\mu_{n}(a) = \hat{a}(\mu_n)$ is Cauchy in $\R$, and thus, converges to a value $\mu(a)$. In particular, we define $\mu(x) = \mu(\{x\})$. We claim that $\mu \in \W$ and $\mu_n \rightarrow \mu$ in norm. 

It is clear that $\mu(a \cup b) = \mu(a) + \mu(b)$ for all events $a, b$ with $a \perp b$. In particular, $\mu(a) = \sum_{x \in a} \mu(x)$ for all finite events $a$. Note also that $\mu(E) = \lim_{n} \mu_{n}(E) = \lim_{n} u(\mu_n)$, that is, 
$\mu(E)$ is independent of $E \in \M$.  Now choose $N$ large enough so that $|\mu(E) - \mu_N(E)|$ and $|\mu_N(a) - \mu(a)|$ are both less than $\epsilon/3$, and let $a \subseteq N$ be 
a finite event such that $|\mu_N (E) - \mu_N(a)| < \epsilon$. Then 
\[|\mu(E) - \mu(a)|  \leq  
|\mu(E) - \mu_N(E)| + |\mu_{N}(E) - \mu_{N}(a)| + |\mu_{N}(a) - \mu(a)| 
\leq  3\epsilon/3 = \epsilon.\]
It follows that $\mu$ is summable over $E$ with $\sum_{x \in E} \mu(x) = \mu(E)$ independent of $E$. 

To see that $\mu$ has finite variation norm, let $\epsilon > 0$ and for $a \subseteq E$ be finite, choose for each $x \in a$ a constant $\epsilon_x$ so that $\sum_{x \in a} \epsilon_{x} = \epsilon$. For each $x \in a$, choose $N_x$ such that 
$x \in E$ choose $n$ with $|\mu_n(x) - \mu(x)| < \epsilon_x$ for all $n \geq N_x$. Let $n \geq \max_{x \in a} N_x$ (remembering here that $a$ is finite): then $\|\mu_n (x) - \mu(x)| < \epsilon_x$ for ever $x \in a$, and thus, 
\[\sum_{x \in a} |\mu(x)| \leq \sum_{x \in a} |\mu_{n}(x)| + |\mu(x) - \mu_n(x)| 
\leq \sum_{x \in a} |\mu_n(a)| + \epsilon \leq \sum_{x \in E} |\mu_n(a)| + \epsilon.\]
It follows that $\|\mu\|_{1} < \infty$, and thus, that $\mu \in \W$. 

It remains to show that $\mu_n \rightarrow \mu$ in norm. Let 
$\epsilon > 0$ and choose $N$ so that $\|\mu_m - \mu_n\| < \epsilon$ for all $m, n \geq N$. Let $a \in \Ev$. 
For a fixed $m \geq N$, $\mu - \mu_m \in \Ev$, so we can find a finite set 
$a_o \subseteq a$ with 
\[(\mu - \mu_{m})(a) < (\mu - \mu_m)(a_o) + \epsilon.\]
Since $a_o$ is finite and $\mu_n(x) \rightarrow \mu(x)$, we can choose $m$ large enough that $|\mu(a_o) - \mu_m(a_o)| < \epsilon$. Then we have 
\begin{eqnarray*} 
|(\mu - \mu_n)(a)\| & \leq & |(\mu - \mu_{m})(a)| + |\mu_m(a) - \mu_n(a)| \\
& \leq & |\mu(a_o) - \mu_m(a_o)| + \epsilon + \|\mu_{m} - \mu_{n}\|\\
& < & 3\epsilon
\end{eqnarray*}
completing the proof. $\Box$

{\em Proof:}  Let $\ell^1(E)$ denote the space of absolutely summable 
functions $f : E \rightarrow \R$, understood as a Banach space in the usual 
norm $\|f\|_1 = \sum_{x \in E} |f(x)|$.  Let $\X$ be the 
set of vectors in $\Pi_{E \in \M} ~\ell^1(E)$ 
with 
\[\|\phi\| := \sup_{E \in \M} \|\phi_{E}\|_1 < \infty.\]
This last quantity is a norm on $\X$, with respect to which 
$\X$ is a Banach space.  
%[NB: can shift this to a subspace of $\ell^1(X)^{\M}$]. 
Now let $\W_o$ be the subset of $\X$ consisting of functions $\phi$ satisfying  
\begin{itemize} 
\item[(a)] $\phi(E)(x) = \phi(F)(x)$ for all $E, F \in \M$ and all $x \in E \cap F$, and 
\item[(b)] $\phi(E)(E)$ independent of $E$. 
\end{itemize} 
Clearly, $\W_o \leq \X$. Since $\hat{(x,E)} : \phi \mapsto \phi(E)(x)$ is a continuous linear 
functional on $\X$ for all $x \in E \in \M$, we see that the set of functions $\phi$ satisfying 
(a) is simply 
\[\bigcap \{\ker(\hat{(x,E)} - \hat{(x,F)}) \ | \ E, F  \in \M, \ x \in E \cap F\},\]
which is closed.  Similarly, the set of functions satisfying (b) is closed: for each $E \in \M$, the functional $f_{E} : \phi \mapsto \phi(E)(E)$ is continuous, since 
\[|\phi(E)(E)| \leq \sum_{x \in E} |\phi(E)(x)| = \|\phi(E)\|_1 \leq \|\phi\|.\]
Thus, the set of functions satisfying (b) is $\bigcap_{E, F \in \M} \ker(f_{E} - f_{F})$, which 
is closed. Now observe that there is a natural mapping $\W_o \rightarrow \W$ given by 
$\phi \mapsto \mu_{\phi}$ where $\mu_{\phi}(x) = \phi(E)(x)$ for any $E \in \M$ with 
$x \in E$. Conversely, for any $\mu \in \W$, we have $\mu = \mu_{\phi}$ 
where $\phi(E,x) = \mu(x)$ for all $x \in E \in \M$. Evidently, $\phi \mapsto \mu_{\phi}$ 
is a norm-preserving bijection, hence, an isometry. Since $\W_o$ is complete, so is 
$\W$. $\Box$ 
}

{\em Remark:} The proof extends almost verbatim to show that the space $\W(A,\X)$ of bounded $\X$-valued weights is complete for any Banach space $\X$. 

\begin{exercise} Extend the proof of Proposition \ref{prop: W complete}, almost verbatim, to show that for any Banach space 
$\X$, $\W(A,\X)$ is complete. 
\end{exercise} 

%\begin{exercise} Give a second proof that $\W(A,\X)$ is complete, as follows. Firsst, for any Banach space $\X$, and any 
set $E$, let $\ell^2(E,\X)$ be the space of functions 
%$f : E \rightarrow \X$ with  
%\[\|f\|_{E} = \sum_{x \in E} \|f(x)\| < \infty.\]
%It is standard fare that $\ell^2(E,\X)$ is a subspace of $\X^{E}%$, that $\| ~ \|_{E}$ is a norm on $\ell^1(E,\X)$, and that 
%$\ell^1(X,\E)$ is complete in this norm. Now ...

\begin{example} (a) If $(S,\Sigma)$ is a measurable space and 
$\M = \M(S,\Sigma)$ is the test space of finite or countable partitions of $S$ by $\Sigma$-measurable sets, then $\W$ is the space of countably-additive real signed measures, 
with the usual variation norm.  This is well known to be complete:   (b) If $\M = \F(\H)$, the frame manual of a Hilbert space $\H$, then $\W$ is the space of self-adjoint trace-class operators on $\H$, and $\|~\cdot~\|_1$ is the trace norm. More on this below.
\end{example}

{\bf The space $\V$}  Now let $\W_{+}$ be the cone in $\W$ consisting of 
non-negative weights. Note that this is closed, owing to Lemma 4.  We define 
\[\V = \W_{+} - \W_{+}\]
and order this by $\V_{+} = \V \cap \W_{+} = \W_{+}$.  We aim to show that if 
$\Omega \subseteq \Pr(\M)$ is any convex set of positive weights that is closed 
in the variation norm on $\W$, the conebase space $\V(\Omega)$ generated by 
$\Omega$ is complete in its base-norm. 

If $\V$ is any ordered normed space with closed unit ball 
$B$ and closed cone $\V_{+}$, the set $B \cap \V_{+} - B \cap \V_{+}$ is closed and convex. 

To prove this, we use the following special case of 
a theorem due to Klee (cf. \cite[p.194]{Peressini}), for which we include a self-contained proof. 

\begin{proposition}\label{prop: Special Klee} Let $K$ be a closed cone in a Banach space $\W$, and let $\V = K - K$. Set 
\[B_{K} = (B \cap K) - (B \cap K)\]
where $B$ is the unit ball of $\W$. Then the Minkowski 
functional of $B_{K}$ is a complete norm on 
$\V$.  
\end{proposition} 

{\em Proof:} The Minkowski functional of $B_K$ is given by 
\[\|z\| = \inf \{ t \geq 0 | z \in tB_{K}\}\]
for all $z \in \V$. 
We note first that the set above is non-empty, i.e., 
$B_{K}$ is absorbing for $\V$. To see this, let 
$z = x - y$ where $x, y \in K$. Since $B$ is absorbing, 
there is some $t$ with $x, y \in tB \cap K = t(B \cap K)$ 
(using $tK \subseteq K$), whence 
\[z \in tB \cap V - tB \cap V = t(B \cap V - B \cap V) = tB_{K}.\]
Now, $\| ~\cdot~\|_{K}$ is generally a seminorm. However, 
as $B_K \subseteq B$, 
 we have 
%\[\{t \geq 0 | z \in tB_{K}\} \subseteq \{t \geq 0 | z \in tB\},%\]Taking infs, we have 
$\|z\| \leq \|z\|_{K}$, and since $\|z\|$ vanishes only for 
$z = 0$, the same holds for $\|z\|_{K}$.  Notice that 
the norms $\|~\cdot~\|$ and $\|~\cdot~\|_K$ agree on the 
cone $K$, since if $x \in K$ then $x \in \|x\|(B \cap K)$, whence 
$\|x\|_{K} \leq \|x\|$. 

It remains to  see that $(z_{n})_{K}$ is complete. Let $(z_n)$ be Cauchy in $\V$ with respect to $\|~\cdot~\|_{K}$.  It follows that $(z_n)$ is also Cauchy w.r.t. $\|~\cdot\|_{K}$, and hence, 
as $\W$ is complete, converges to some point $z \in \W$. 
We need to show that $z \in \V$ and that $z_{n} \rightarrow z$ with respect to $\|~\cdot~\|_{K}$. To this end, it will suffice 
to show that $(z_n)$ has some $\|~\cdot~\|$-convergent subsequence. Since $(z_n)$ is Cauchy, we can always find a subsequence of $(z_n)$, say $w_k = z_{n_k}$, with 
\[\|w_{k+1} - w_{k}\|_{K} \ \leq \ 1/2^{k}\]
or, equivalently, 
\[w_{k+1} - w_{k} \ \in \ \frac{1}{2^{k}}(B \cap V - B \cap V)\]
for all $k$.  It follows that, for every $k$, 
we can find $x_k, y_k \in (\frac{1}{2^k}B) \cap V)$ with 
\[w_{k+1} - w_{k} = x_{k} - y_{k}.\] 
Then $\|x_k\|, \|y_k\| \leq 1/2^{k}$ and 
\[\|x_{k} - y_{k}\| = \|w_{k+1} - w_{k}\| \leq 1/2^{k}.\]
It follows that 
\[\sum_{k=1}^{\infty} x_{k},  \ \sum_{k=1}^{\infty} y_k\]
converge to elements, say $x$ and $y$, of $\W$, with both sums 
$\|~\cdot~\|$-convergent.  Let 
\[S_{n} = \sum_{k=1}^{n} x_k \ \ \mbox{and} \ 
T_{n} = \sum_{k=1}^{n} y_k\]
so that $S_n \rightarrow x$ and $T_n \rightarrow y$ with 
respect to $\|~\cdot~\|$. Note that since $x_n, y_n \in K$, 
$S_n, T_n \in K$ as well. Because $K$ is closed in $\|~\cdot~\|$, 
$x, y \in K$. Moreover, we have $S_n \leq x$ for every $n$ 
and likewise $T_n \leq y$ for every $n$. Thus, $x - S_n$ and 
$y - T_n$ belong to $K$. Since the two norms agree on $K$, we have $\|x - S_n\|_{K}, \|y - T_n\| \rightarrow 0$, so $S_n \rightarrow x$ and $T_n \rightarrow x$  
with respect to $\|~\cdot~\|_{K}$. Finally, we have 
\[w_{k+1} - w_{1} = S_{n+1} - T_{n+1} \rightarrow x - y\]
so $w_{k} \rightarrow x - y + w_1 \in \V$.  Since 
$(z_n)$ is Cauchy and has $(w_k)$ as a convergent subsequence, 
$(z_n)$ also converges to $x - y + w_1$, and we are done. $\Box$ 
%\end{document}

\begin{corollary}\label{cor: V(M) complete} Let $\M$ be any test space, and 
let $\Omega = \Pr(\M)$ be the set of all probability weights 
on $\M$, and let $\V = \V_{A}$ for $A = (\M,\Omega)$. Then 
$(\V,\Omega)$ is a complete BNS. 
\end{corollary} 

{\em Proof:} $\W = \W(\M)$ is a Banach space, $\W_{+}$ is a 
closed cone in $\W$, and the Minkowski functional 
of $B \cap \W_{+} - B \cap \W_+$, where $B$ is the closed unit 
ball of $\W$, is the closed unit ball for the base-norm 
on $\V$. $\Box$ 
%{\blue [A bit of extra checking here is in order...]} 

\begin{lemma}\label{lemma: closed subcone}  Let $(\V,\Omega)$ be a BNS with closed cone $K$ and base $\Omega$. Let $\Omega_o$ be a closed convex subset 
of $\Omega$.  Then the cone $K_o = \{ ta | a \in \Omega_o, t \geq 0\}$ generated by $\Omega_o$ is also closed in $\W$. Moreover, 
if $\V$ is complete in its base-norm, so is $\V_o = K_o - K_o$ 
in the base-norm generated by $\Omega_o$. 
\end{lemma} 

{\em Proof:} To see that $\V_o$ is closed, let $z_n = r_n a_n \in K_o$ where 
$r_n \geq 0$ and $a_n \in \Omega_o$, and suppose 
$r_n a_n \rightarrow ra$ where $r \geq 0$ and $a \in \Omega$. 
We wish to show that $a \in \Omega_o$. This is trivial 
if $r = 0$, so assume in what follows that $r > 0$. 
Applying the order unit $u$ associated with $\Omega$, 
and noting that this is continuous with respect to the 
base-norm on $\W$, we have $r_n \rightarrow r$. Claim: 
$ra_n \rightarrow ra$. To see this, note that 
\begin{eqnarray*}
\|ra_n - ra\| & = & \|ra_n - r_n a_n + r_n a_{n} - ra\| \\
& \leq & \|ra_n - r_n a_n\| + \|r_n a_n - r a\| \\
& \leq & |r - r_n|\|a_n\| + \|r_n a_n - ra\|,
\end{eqnarray*}
and both of the last terms go to $0$ as $n \rightarrow \infty$. 
Since $r > 0$, we have 
\[a_n = \frac{1}{r} ra_n \rightarrow \frac{1}{r} r a = a.\]
Thus, as $\Omega_o$ is closed, $a \in \Omega_o$ and we are 
done. 

Now suppose that $\V$ is complete. Let $B$ denote 
the unit ball of $\V$, that is $B = \co(\Omega \cup -\Omega)$. 
Then 
\[B \cap K_o = B \cap (K \cap K_o) = (B \cap K) \cap K_o\]
that is $B \cap K_o$ consists of elements of $\V$ of the 
form $ta_o$ for some $t \geq 0$ and $a_o \in \Omega_o$, 
with $ta_o \in B$. But since $ta_o \geq 0$ in $\V$ 
(since $\Omega_o \subseteq \Omega$), we have 
$ta_o \in B$ iff $0 \leq t \leq 1$, i.e., 
$B \cap K_o = B_o \cap K_o$. Thus, 
$B \cap K_o - B \cap K_o = B_o$, and the completeness of 
$\V_o$ 
follows from Proposition \ref{prop: Special Klee}. $\Box$

We now have 

\begin{proposition} Let $(\M,\Omega)$ be any probabilistic 
model with $\Omega$ convex and $\| \cdot\|_1$-closed in 
$\V(\Omega)$ is a complete BNS\end{proposition} 

{\em Proof:} Immediate from Corollary \ref{cor: V(M) complete} and Lemma \ref{lemma: closed subcone}. $\Box$\\

%\end{document}

\tempout{
{\bf Example 3:} Let $X$ be the unit sphere of an infinite-dimenionsional Hilbert space $\H$, let ${\F}(\H)$ be the corresponding frame manual, and let $\Omega = \Pr(\H)$. By 
Gleason's Theorem, we can associated $\Omega$ with the convex set of density operators in $\V(\H)$, the space of self-adjoint trace-class operators on $\H$. 
It is a standard fact [ref] that $\V(\H)$ is a Banach space in the trace norm, given by 
$\|T\| = \Tr(|T|)$, where $|T| = (T^{\ast} T)^{1/2}$. It follows from Klee's Theorem [details?]
that $\V(\H)$ is also complete in its base norm.  Using the Spectral Theorem, 
not hard to see that the trace-norm is the same as 
the norm $\|~\cdot~\|_1$ introduced above. Indeed, if $|T| = \sum_{x \in E} t_{x} P_{x}$ 
$E$ is an ONB, $t_{x} > 0$, and $P_{x}$ is the rank one projection attached to the unit vector $x \in E$, we have 
\[\sum_{x \in E} |\Tr(|T|P_{x})| = \sum_{x \in E} |t_{x}| = \Tr(|T|),\]
so $\|T\|_{\tr} \leq \|T\|_1$. On the other hand, of $F$ is any other orthonormal 
basis, then since $\Tr(P_{x}P_{y}) = |\langle x, y \rangle|^2$, we have 
\begin{eqnarray*}
\sum_{y \in F} |\Tr(|T|P_{y})| & = & 
\sum_{y \in F} \left | \sum_{x \in E} t_{x} \Tr(P_{x} P_{y}) \right | \\
& \leq & 
\sum_{y \in F } \sum_{x \in E} |t_{x}| |\Tr(P_{x}P_{y})| \\
& \leq & 
\sum_{x \in E}  |t_{x}| \left (\sum_{y \in F}  \Tr(P_x P_y) \right )  \\
& \leq & \sum_{x \in E} |t_{x}| = \Tr(|T|).
\end{eqnarray*}
Thus, $\|T\|_{1} \leq \|T\|_{\tr}$. 
It follows 
(but is easy to check directly) that $\Tr(W_n a) \rightarrow \Tr(Wa)$ for all bounded s.a. operators $a$ on $\H$, and 
thus, that $\Tr(W) = \lim_{n} \Tr(W_n) = 1$.  The positive 
cone of $\V(\H)$ is trace-norm closed, so the set 
$\Omega(\H)$ of density operators is closed, and the completeness 
of $\V(\H)$ can be seen as an instance of Proposition 3 above. 
\\
}

\tempout{
{\bf Unitality and Compactness} The following theorem of Cook's is really useful. Define the {\em ultraweak} topology on $\V(A)$ to be the weakest one making evaluations $\hat{x}$ continuous (i.e., the relative product topology.) 

\begin{proposition}[Cook] Let $K \subseteq \Omega$ be ultraweakly = outcome-wise compact and unital. Then 
$\M(A)$ is locally finite.
\end{proposition} 

{\em Proof:} Suppose $E \in \M$ were infinite. Because $K$ is unital, for every finite event $a \subseteq E$ we can find some $\alpha_{a} \in K$ with $\alpha_{a}(a) = 0$. If the 
net $(\alpha_{a})$, indexed by finite subsets of $E$, has a convergent subnet $\alpha_{a_{i}}$, with limit $\alpha$, then for any $x \in X$ we have $\alpha_{a_i} \rightarrow \alpha(x)$, and it follows that for any finite event $b$, we also have $\alpha_{a_i}(b) \rightarrow \alpha(b)$. For a given choice of $b$, $i$ with $b \subseteq a_i$, so we have for all $j$ with $j \geq i$ that $b \subseteq a_{i} \subseteq a_j$, whence, $\alpha_{a_j}(b) = 0$ for all $j \geq i$.It follows that $\alpha(b) = 0$.  Since this holds for all finite events, $b$ is identically $0$ on $X$, and thus, does not belong to $\Omega$, let alone $K$. In particular, $K$ is not compact. $\Box$ 

It follows that neither $\F(\H)$ nor $\B(S,\Sigma)$ is 
ultraweakly compact. 

\begin{example} In the case of the regular Borel model $(\B(S,\Sigma),\Delta(S))$ where $S$ is compact, the relevant topology is that of pointwise convergence on continuous functions $f : S \rightarrow [-1,1]$, which is quite distinct from that 
of convergence pointwise on $X = \Sigma \setminus \{\emptyset\}$.
\end{example} 
}

\end{appendix} 

\begin{thebibliography}{*}

\bibitem{AC} S. Abramky and B. Coecke, Categorical quantum mechanics, in  K. Engesser  D. Gabbay, and D. Lehman, Eds,  {\em Handbook of Quantum Logic and Quantum Structures vol II}, Elsevier, 2008; {\tt arXiv:0808.1023}


\bibitem{Alfsen} E. M. Alfsen, {\em Compact Convex Sets and Boundary Integrals}, Springer, 1971

\bibitem{AS1} E. Alfsen and F. Shultz, {\em State Spaces of 
Operator Algebras}, Birkh{\"a}user, 2001

\bibitem{AS} E. Alfsen and F. Shultz, {\em Geometry of State Spaces of Operator Algebras},  Birkh{\"a}user, 2003 
%\href{https://doi.org.10.1007/978-1-4612-0019-2}{DOI: 10.1007/978-1-4612-0019-2}\vspace{.1in}


\bibitem{AT} C. Aliprantis and D. Tourky, {\em Cones and Duality}, Springer, 2007 
%\href{https://doi.org.10.1090/gsm/084}{DOI: %10.1090/gsm/084}\vspace{.1in}

\bibitem{ALPP} T. Aubrun, L. Lami, C. Palazuelos and M. Pla\'{a}vala, Entangleability of cones, Geom. Funct. Anal. {\bf 31} (2021) {\tt arXiv:1911.09663} 

\bibitem{Baez-QQ} J. Baez, Quantum quandaries, 
in S. French, D. Rickles and J. Saatsi (Eds.) {\em Structural Foundations of Quantum Gravity}  Oxford, 2006; {\tt arXiv:quant-ph/0404040}, 


\bibitem{BBLW-Teleportation} H. Barnum, J. Barrett, M. Leifer, and A. Wilce, Teleportation in general probabilistic theories, in S. Abramsky and M. Mislove (Eds.), {\em The Mathematics of Information Flow}, Proceedings of Symposia in Applied Mathematics {\bf 71}, AMS, 
2012 {\tt arXiv:0805.3553} 

\bibitem{BFRW} H. Barnum, C. Fuchs, J. Renes, and A. Wilce, 
Influence-free states on compound quantum systems, {\tt arXiv: 
quant-ph/0507108}, 2005. 

\bibitem{BGW} H. Barnum, C. P. E. Gaebler and A. Wilce, Foundations of Physics {\bf 43} (2013) {\tt arxiv.org/0912.5532}

\bibitem{CCEJA} H. Barnum, M. Graydon and A. Wilce, Categories and composites of euclidean Jordan algebras, Quantum {\bf 4} (2020) {\tt arXiv:1606.09331}

\bibitem{BW-Information} H. Barnum and A. Wilce, 
Information processing in convex operational theories, in B. Coecke, I. Mackie
P. Panangaden, and P. Selinger (Eds.), 
{\em Proceedings of the Joint 5th International Workshop on Quantum Physics and Logic and 4th Workshop on Developments in Computational Models (QPL/DCM 2008)}, 
 Electronic Notes in Theoretical Computer Science
{\bf 270} (2011); {\tt arXiv: arXiv:0908.2352} 
%\bibitem{BH} H. Barnum and J. Hilgert, Strongly symmetric spectral convex bodies are Jordan algebra state spaces {\tt arXiv:1904.03753} 

%\bibitem{BMU} H. Barnum, M. M\"{u}ler and C. Ududec, Higher-order interference and single-system postulates characterizing quantum theory, New J. Phys. {\bf 16} (2014) {\tt arXiv:1403.4147}

%\bibitem{BUW} H. Barnum, C. Ududec and J. van de Wetering, Self-duality and Jordan structure of quantum theory follow from homogeneity and pure transitivity, {\tt arxiv.org/abs/2306.00362v1}

\bibitem{BW} H. Barnum and A. Wilce, Post-classical probability theory, in G. Chiribella and R. Spekkens, eds., {\em Quantum Theory: Informational Foundations and Foils}, Springer, 2017 
%\doi{10.1007/978-94-017-7303-4_11};   \href{http://arXiv:1205.3833}{\tt arXiv:1205.3833}\vspace{.1in}

\bibitem{Barrett} J. Barrett, Information processing in generalized probabilistic theories, Physical Review A {\bf 75} (2005) {\tt arXiv:quant-ph/0508211}

%\href{https://doi.org.10.1103/PhysRevA.75.032304}{DOI: 10.1103/PhysRevA.75.032304}; 
%\href{http://arXiv:quant-ph/0508211}{\tt arXiv:quant-ph/0508211}\vspace{.1in}

\bibitem{B-MW} L. J. Bunce and J. D. Maitland-Wright, The Mackey-Gleason problem, 
Bull. Am. Math. Soc. {\bf 26} (1992)

%\bibitem{DB}  B. Daki\v{c} and C. Brukner, Quantum theory and
%  beyond: is entanglement special?  in H. Halvorson, ed., {\em %Deep Beauty}, Princeton, 2011 , {\tt arXiv:0911.0695}


\bibitem{CDP-Purification} G. Chiribella, G. M. D'Ariano and P. Perinotti, Probabilistic theories 
with purification, Phys. Rev. A {\bf 80} (2009) {\tt arXiv:0908.1583}

\bibitem{CDP} G. Chiribella, M. D'Ariano and P. Perinotti, Informational derivation of quantum theory, Phys. Rev. A {\bf 84} (2011), {\tt arXiv:1011.6451}

\bibitem{Coecke-LQMII} B. Coecke, The logic of quantum mechanics - Take II, in 
J. Chubb, A. Eskandarian, and V. Harizonov (Eds.), 
{\em Logic and Algebraic Structures in Quantum Computing}, Cambridge, 2016 
({\tt arXiv:1204.3458})

\bibitem{CK} B. Coecke and A. Kissinger, {\em Picturing Quantum Processes}, 
Cambridge, 2017

\bibitem{Cook} T. Cook, Banach spaces of weights on quasimanuals, Int. J. Theor. Phys. 
1985



\bibitem{Pavia-book} G. M. D'Ariano, G. Chiribella and P. Perinotti, {\em Quantum Theory from First Principles}, Cambridge 2017 

\bibitem{DL} E. B. Davies and J. T. Lewis, An operational approach to quantum probability, Comm. Math. Phys. {\bf 17} (1970) 


\bibitem{Dunford-Schwartz} N. Dunford and J. Schwartz, {\em Linear Operators, Part I: General Theory}, Wiley, 1988


\bibitem{Edwards}  C. M. Edwards, The operational approach to quantum probability, I. Comm.Math. Phys. {\bf 16} (1970) 

\bibitem{FS} B. Fong and D. Spivak, {\em An Invitation to Applied Category Theory: Seven Sketches in Compositionality}, Cambridge 2019 (available 
online as {\tt arXiv:1803.05316}) 

\bibitem{HV} C. Heunen and C. Vicary, {\em Categories for Quantum Theory}, Oxford, 
2020



\bibitem{FB} D. J. Foulis and M. K. Bennett, Effect algebras 
and unsharp quantum  logics, Found. Phys. {\bf 24} (1994)

\bibitem{FR-ELPS}  D.J. Foulis and C. H. Randall, The empirical logic approach to the physical sciences,  in A. Hartk\"{a}mper 
an H Neumann (eds) {\em Foundations of Quantum Mechanics and Ordered Linear Spaces} (Lecture Notes in Physics, vol. {\bf 29}) Springer, 1974


\bibitem{FR-TP}  D. J. Foulis and C. H. Randall, Empirical logic and tensor products, Proceedings of the Colloquium on the
Interpretations and Foundations of Quantum Theories, Fachbereich Physik der Philipps Universit\"{a}t, Marburg, Germany (1979). 


\bibitem{FGR}  D. J. Foulis, R. Greechie, and G. R\"{u}ttimann, 
Filters and supports in orthoalgebras, International Journal of Theoretical Physics {\bf 31} (1992)

\bibitem{FGR-II} D. J. Foulis, R. Greechie, and G. R\"{u}ttimann, Logico-algebraic structures II: supports in test spaces, Int. J. Theor. Phys. {\bf 32} (1993)

\bibitem{FPR}  D. J. Foulis, C. Piron and C. H. Randall, 
Realism, operationalism, and quantum mechanics, Foundations of Physics {\bf 13} (1983)


%\bibitem{CS} G. Chiribella and C. M. Scandolo, EPTCS {\bf 195}, %(2015), {\tt arXiv:1506.00380} $\leftarrow$ another way to get %spectrality

%\bibitem{CDP} G. Chiribella, M. D'Ariano and P. Perinotti, Informational derivation of quantum theory, PR-A {\bf 84} (2011), {\tt arXiv:1011.6451}

\bibitem{Gleason} A. Gleason, Measures on the closed subspaces of a Hilbert space, J. Math. and Mech. {\bf 6} (1957)


\bibitem{GMCD} D. Gross, M. M\"{u}ller, R. Colbeck, and O. Dahlsten, All reversible dynamics in maximally nonlocal theories are trivial, Phys. Rev. Lett. {\bf 104}, 2010 
%\bibitem{Bush} P. Bush 

%\bibitem{Caves et al} C. Caves et al., 


%\bibitem{FR-ELQM} Synthese, Vol. 29, No. 1/4, Logic and %Probability in Quantum Mechanics (Dec., 1974), pp. 81-111

\bibitem{Hardy}  L. Hardy, Quantum theory from five reasonable axioms (2001) {\tt arXiv:quant-ph/0101012}

\bibitem{Hardy-Foilable} L. Hardy, Foilable operational structures for general probabilistic theories, {\tt arXiv.org/pdf/0912.4740} 

\bibitem{Hardy-Reconstructing} L. Hardy, Reformulating and reconstructing quantum 
mechanics {\tt arXiv:quant-ph/1104.2066v3} (2011)


\bibitem{Kelly-Namioka} J. Kelly and I. Namioka, 
{\em Linear Topological Spaces}, van Nostrand, ...

\bibitem{KFR}  M. Kl\"{a}y, D. J. Foulis and C. H. Randall, 
Tensor products and probability weights, Int. J. Theor. Phys. 
{\bf 26} (1987) 

\bibitem{Klay} M. Kl\"{a}y, Einstein-Podolski-Rosen exeriments I: the structure of the probability space I, II, Foundations of Physics Letters {\bf 1} (1988)	

\bibitem{Ludwig} G. Ludwig,  Versuch einer axiomatischen Grundlegung der Quantenmechanikund allgemeinerer physikalischer Theorien. Z. Physik {\bf 181} (1964);  Attempt of an Axiomatic Foundation of Quantum Mechanics and More General Theories, II, Comm. Math. Physics, Commun. math. Phys. {\bf 4} (1967) 

\bibitem{Mackey} G. W. Mackey, {\em Mathematical Foundations of Quantum Mechanics}, Addison-Wesley, 1957 (Reprinted by Dover)  

\bibitem{MM} Ll. Masanes and M. M\"{u}ller, A derivation of quantum theory from physically reasonable requirements, NJP {\bf 13} (2011), {\tt arXiv:1004.1483} 

\bibitem{Mielnik} B. Mielnik, Theory of filters, Commun. math. Phys. {\bf 15} (1969)

\bibitem{NP} I. Namioka and R. Phelps, Tensor products of compact convex sets, Pacific. J. Math. {\bf 31} 1969

\bibitem{Peressini} A. Peressini, {\em Ordered topological vector spaces}, Harper and Row, 1967

\bibitem{Plavala} M. Pl\'{a}vala, General probabilistic theories: an introduction (2021) {\tt arXiv:2103.07469}

\bibitem{RF70} C. H. Randall and D. J. Foulis, An approach 
to empirical logic,  Am. Math. Monthly {\bf 77}, 1970

\bibitem{RF-NoGo} C. H. Randall and D. J. Foulis, Tensor products of quantum logics do not exist, Notices Am. Math. Soc. {\bf 26} (1979), A-557

\bibitem{RF-Hooker} C. H. Randall and D. J. Foulis, The operational approach to quantum mechanics, in C. A. Hooker (Ed.), {\em Physical Theory as Logic-Operational Structure}, University of Western Ontario Series in Philosophy of Science {\bf 7}, D. Reidel, 1979 	

\bibitem{RJF} C. H. Randall, M. Janowitz and D. J. Foulis, 
Orthomodular generalizations of homogeneous boolean algebras, 
J. Aust. Math. Soc. {\bf 15} (1971)

\bibitem{Rau}  J. Rau, Ann. Phys. {\bf 324} (2009), {\tt arXiv:0710.2119}

\bibitem{Riehl} E. Riehl, {\em Category Theory in Context},  Dover 2016

\bibitem{Shulz} F. W. Shulz, A characterization of state spaces of orthomodular lattices, J. Comb. Theory {\bf 17} (1974)

\bibitem{van Fraassen} B. van Fraassen, One or two gentle remarks about Hans Halvorsen's critique of the semantic view, 
Phil. Science {\bf 81} (2014) 

\bibitem{Werner} R. Werner, Quantum states with Einstein-Podolsky-Rosen correlations admitting a hidden-variable model,  Physical Review A. {\bf 40} (1989)

\bibitem{Wigner} E. Wigner, comment during a conversation with Randall and Foulis in Amherst, ca 1978; D. J. Foulis, personal communication.

\bibitem{Wilce92} A. Wilce, Tensor products in generalized measure theory, Int. J. Theor. Phys. {\bf 31} (1992)

\bibitem{AW-SEP} A. Wilce, Quantum logic and probability theory, 
Stanford Encyclopedia of Philosophy, 2002

\bibitem{AW-handbook} A. Wilce, Test Spaces, in H. Engesser, D. Gabbay and D. Lehmann, eds., {\em Handbook of Quantum Logic and Quantum Structures}, North Holland, 

\bibitem{AW-RR} A. Wilce, A royal road to quantum theory (or thereabouts), Entropy {\bf 20} (2018), {\tt arXiv:1606.09306}

\bibitem{AW-CFQM} A. Wilce, Conjugates, filters, and quantum mechanics, Quantum {\bf 3}  (2019), {\tt arXiv:1206.2897v8}


\bibitem{AW-SC} A. Wilce, Symmetry and composition in probabilistic theories, ENTCS {\bf 270} (2011) {\tt  arXiv:0910.1527.v2} 
 

\bibitem{AW-shortcut} A. Wilce, A shortcut from categorical quantum mechanics to generalized probabilistic theories, in Bob Coecke and Aleks Kissinger (Eds.): 14th International Conference on Quantum Physics and Logic (QPL), EPTCS {\bf 266} (2018)  {\tt arXiv:1803.00707}

\bibitem{AW-CCaM} A. Wilce, Coarse-graining and compounding as monads, 
{\tt arXiv:2410.08818}
 
\bibitem{Wittstock} G. Wittstock, Ordered normed tensor products, in A. Hartk\"{a}mper and H. Neumann (Eds.), {\em Foundations of Quantum Mechanics and Ordered Linear Spaces}, 
Lecture Notes in Physics {\bf 29}, Springer, 1974

\bibitem{Wright} R. Wright, Spin manuals: empirical logic talks quantum mechanics, in A. R. Marlow, ed., {\em Mathematical Foundations of Quantum Theory}, Academic Press 1978
%\bibitem{AW-top} 

\bibitem{vN} J. von Neumann, {\em Mathematical Foundations of Quantum Mechanics}, 1937; English translation Princeton, 1957. 
\end{thebibliography}
\end{document}